\documentclass[11pt]{article}
\pdfoutput=1

\usepackage{jheppub}
\usepackage{xcolor}
\usepackage{array}
\usepackage{mathtools}             
\usepackage{booktabs, multirow}    
\usepackage{tikz}                  
\usepackage{amsthm}                
\usepackage[boxsize=0.6em]{ytableau}    
\usepackage[bb=stix, cal=boondoxupr]{mathalpha}    

\usetikzlibrary{calc, shapes}
\newtheorem{claim}{Claim}

\allowdisplaybreaks

\definecolor{typeA}{HTML}{E06836}
\definecolor{typeB}{HTML}{4B9423}
\definecolor{typeAB}{HTML}{5159B8}
\definecolor{myRed}{RGB}{150,20,0}      

\newcommand{\A}[1]{{\color{typeA}{\ifmmode\mathrm{#1}\else{#1}\fi}}}
\newcommand{\B}[1]{{\color{typeB}{\ifmmode\mathrm{#1}\else{#1}\fi}}}
\newcommand{\AB}[1]{{\color{typeAB}{\ifmmode\mathrm{#1}\else{#1}\fi}}}

\newcommand{\R}{\mathbb{R}}
\newcommand{\Z}{\mathbb{Z}}
\newcommand{\I}{\mathbb{I}}

\renewcommand{\H}{\mathcal{H}}     
\newcommand{\Gram}{\mathbb{G}}     
\newcommand{\gram}{\mathbb{g}}     

\newcommand{\irr}{\text{irr}}
\newcommand{\avbl}{\text{avbl}}
\newcommand{\af}{\text{af}}

\DeclareMathOperator{\diag}{diag}
\DeclareMathOperator{\conv}{Conv}

\DeclareMathOperator{\tr}{tr}
\DeclareMathOperator{\U}{U}
\DeclareMathOperator{\SU}{SU}
\DeclareMathOperator{\SO}{SO}
\DeclareMathOperator{\Sp}{Sp}
\newcommand{\adj}{\mathrm{Adj}}
\newcommand{\rep}[2][]{\mathbf{\underline{#2}^{#1}}}    
\newcommand{\repss}[2]{\mathbf{\underline{#1}_{#2}}}    

\newcommand{\e}{\mathcal{e}}
\renewcommand{\v}{\mathcal{v}}
\newcommand{\C}{\mathcal{C}}
\newcommand{\E}{\mathcal{E}}
\newcommand{\F}{\mathcal{F}}
\newcommand{\G}{\mathcal{G}}
\newcommand{\h}{\mathcal{h}}

\newcommand{\K}{\mathcal{K}}
\newcommand{\N}{\mathcal{N}}
\newcommand{\V}{\mathcal{V}}

\usepackage{letltxmacro}
\makeatletter
\let\oldr@@t\r@@t
\def\r@@t#1#2{%
\setbox0=\hbox{$\oldr@@t#1{#2\,}$}\dimen0=\ht0
\advance\dimen0-0.2\ht0
\setbox2=\hbox{\vrule height\ht0 depth -\dimen0}%
{\box0\lower0.4pt\box2}}
\LetLtxMacro{\oldsqrt}{\sqrt}
\renewcommand*{\sqrt}[2][\ ]{\oldsqrt[#1]{#2}}
\makeatother
\title{Towards a complete classification of 6D supergravities}

\author[\clubsuit,\spadesuit]{Yuta Hamada}
\author[\clubsuit]{and Gregory J.\ Loges}

\preprint{KEK-TH-2572}
\emailAdd{yhamada@post.kek.jp}
\emailAdd{gloges@post.kek.jp}

\affiliation[\clubsuit]{
    Theory Center, IPNS, High Energy Accelerator Research Organization (KEK),\\
    \phantom{}\quad 1-1 Oho, Tsukuba, Ibaraki 305-0801, Japan
}
\affiliation[\spadesuit]{
    Graduate Institute for Advanced Studies, SOKENDAI, 1-1 Oho, Tsukuba, Ibaraki 305-0801, Japan
}

\abstract{
    The constraints arising from anomaly cancellation are particular strong for chiral theories in six dimensions. We make progress towards a complete classification of 6D supergravities with minimal supersymmetry and non-abelian gauge group.
    First, we generalize a previously known infinite class of anomaly-free theories which has $T\gg 9$ to essentially any semi-simple gauge group and infinitely many choices for hypermultiplets. The construction relies on having many decoupled sectors all selected from a list of four simple theories which we identify.
    Second, we use ideas from graph theory to rephrase the task of finding anomaly-free theories as constructing cliques in a certain multigraph. A branch-and-bound type algorithm is described which can be used to explicitly construct, in a $T$-independent way, anomaly-free theories with an arbitrary number of simple factors in the gauge group. We implement these ideas to generate an ensemble of $\mathcal{O}(10^7)$ irreducible cliques from which anomaly-free theories may be easily built, and as a special case obtain a complete list of $19,\!847$ consistent theories for $T=0$, for which the maximal gauge group rank is $24$.
    Modulo $\U(1)$, $\SU(2)$ and $\SU(3)$ simple factors and the new infinite families, we give a complete characterization of anomaly-free theories and show that the bound $T\leq 273$ is sharp.
}

\begin{document}

\maketitle

\section{Introduction}
\label{sec:intro}

There has been much work in recent years in understanding to what extent the string theory landscape is universal, in the sense that it encompasses all low-energy theories which comply with general consistency requirements. There has been much success for $D\geq7$~\cite{Adams:2010zy,Kim:2019vuc,Kim:2019ths,Montero:2020icj,Cvetic:2020kuw,Hamada:2021bbz,Bedroya:2021fbu} where anomalies, together with supersymmetry, lead to an exact\footnote{Up to a classification of 3d $\mathcal{N}=4$ SCFTs for $D=7$. See \cite{Bedroya:2021fbu} for details.} match between consistent low-energy theories and those which can be embedded in string theory. At some point anomaly cancellation and supersymmetry alone are not enough; indeed, the swampland program~\cite{Vafa:2005ui} aims to characterize those theories which actually have a UV completion within a theory of quantum gravity (see, e.g.,~\cite{Palti:2019pca,vanBeest:2021lhn,Agmon:2022thq}).

It is natural to consider this question of string universality for 6D chiral theories where cancellation of gauge, gravitational and mixed anomalies via the Green-Schwarz-West-Sagnotti mechanism~\cite{Green:1984bx,Green:1984sg,Sagnotti:1992qw} places particularly strong conditions on supergravities with minimal supersymmetry. Knowledge of the structure of consistent theories has grown over the years~\cite{Seiberg2011,Tarazi2021}, and systematic searches with few tensor multiplets~\cite{Kumar:2009us,Kumar2011} or relatively restrictive gauge group~\cite{Avramis2005} have produced ensembles within which to look for patterns. For example, it was argued in~\cite{Kumar:2009us} that there is an exact match between some classes of consistent 6D supergravities with a single tensor multiplet and compactifications of $\mathcal{N}=1$ string theories on $K3$. More generally, in~\cite{Kumar:2010ru} it was proved that the number of anomaly-free non-abelian theories with fewer than nine tensor multiplets is finite. For $T\geq 9$ the proof fails, not because of technical limitations but rather because there actually appear infinite families as soon as $T=9$. For example, two types of infinite families were identified in~\cite{Kumar:2010ru}, distinguished by whether the number of simple gauge factors is bounded or not. There it was also shown that only a finite subset from each family can have an F-theory~\cite{Vafa:1996xn,Morrison:1996na,Morrison:1996pp} realization in the geometric regime;this aligns with the general expectation that the landscape of quantum gravity-compatible low-energy theories is finite~\cite{Acharya:2006zw,Hamada:2021yxy,Grimm:2021vpn}. Furthermore, using the BPS string probe~\cite{Kim:2019vuc,Kim:2019ths,Lee:2019skh,Tarazi2021,Martucci:2022krl,Hayashi:2023hqa}, in the context of the Swampland program~\cite{Tarazi2021} shows that various infinite families are rendered finite. See~\cite{Taylor:2011wt} for an overview of 6D F-theory compactifications, and~\cite{Baykara:2023plc} for intrinsically non-geometric compactifications which cannot be realized in the geometric F-theory framework.

In this work we make progress towards surveying a significant portion of the interior of the anomaly-free landscape for all $T$. The proofs of finiteness for $T<9$ are by contradiction~\cite{Kumar2009,Kumar:2010ru} and ultimately only require an understanding of how consistent theories behave in the foothills, so to say, where the structure simplifies drastically. Much can be learned by looking at what types of theories appear in the bulk: the results of our exploration are summarized schematically in figure~\ref{fig:map}. Our focus will be on the central cape, which we will show is separated from the mainland by a thin isthmus beyond which one must include one or more of the indicated simple factors. The mainland is populated by infinitely many infinite families of anomaly-free theories which, as we will see, require both the number of tensor multiplets and the gauge group to be very large. In enumerating consistent theories on the cape, we will uncover strong evidence that in this part of the world the number of tensor multiplets is universally bounded as $T\leq 273$.

\begin{figure}[t]
    \centering
    \begin{tikzpicture}
        \node at (0,0) {\includegraphics[width=0.98\textwidth]{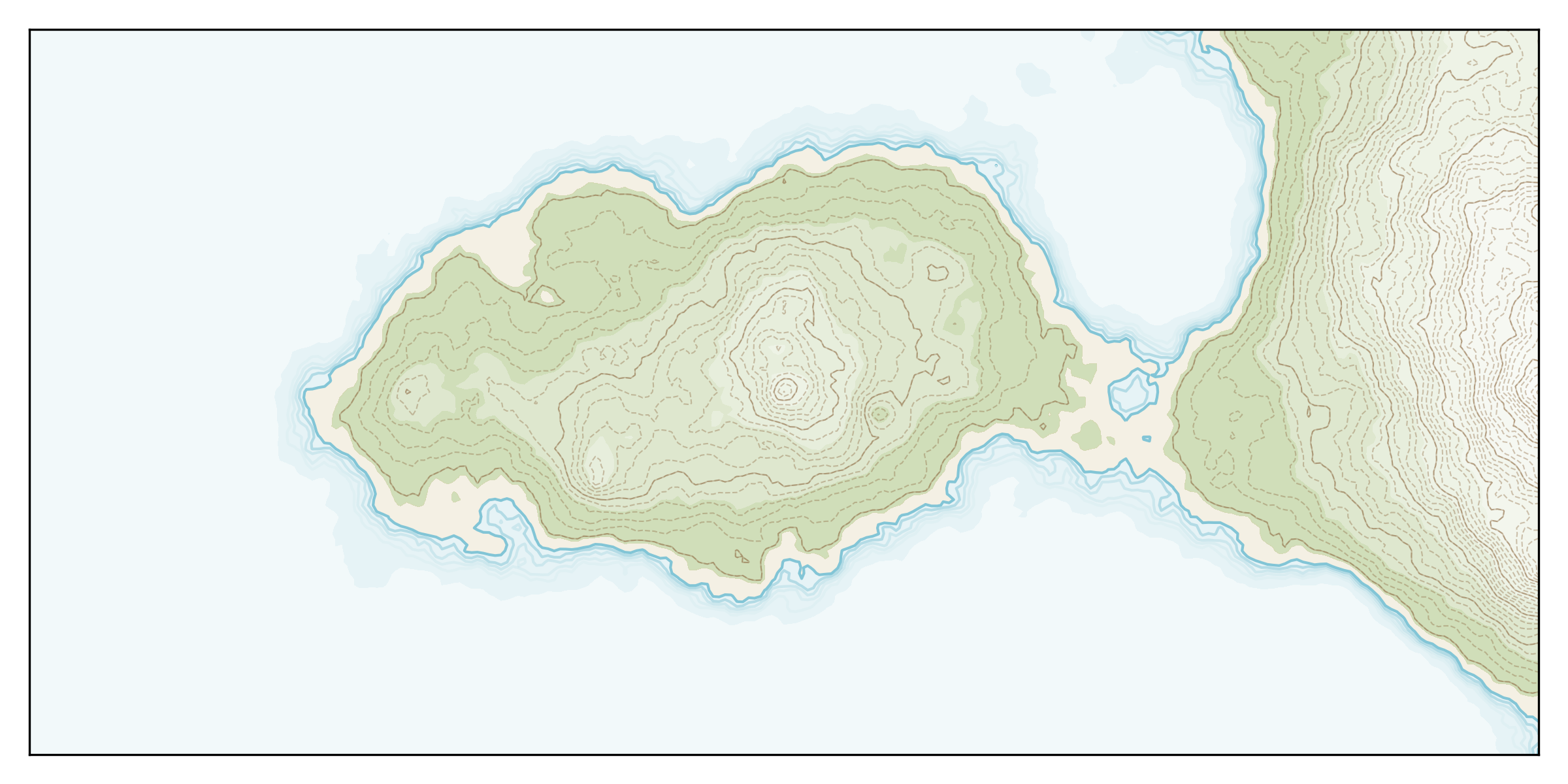}};

        \draw[myRed, dashed, thick] (-2.85,2) -- (-2.55,-1.7);
        \draw[myRed, dashed, thick] (3.5,0.8) -- (3,-1.1);
        \draw[myRed, dashed, thick] (4.25,2) .. controls (5.5,1.8) and (6.5,-1.25) .. (5.5,-2.25);

        \draw[very thick, myRed, ->] (2,0.3) .. controls (3,-0.15) and (3,-0.35) .. (4,-0.3);
        \draw[very thick, myRed, ->] (5.1,0.5) .. controls (5.3,0.6) and (6,1.4) .. (6.25,2.25);
        \draw[very thick, myRed, ->] (5.5,-0.8) .. controls(5.7,-0.9) and (6.4,-1.2) .. (6.9,-1.6);

        \node[gray!25!black,right] at (5.6,0.35) {\tiny$\{E_6,\emptyset\}$};
        \node[gray!25!black,right] at (5.75,0) {\tiny$\{E_7,\emptyset\}$};
        \node[gray!25!black,right] at (5.8,-0.35) {\tiny$\big\{E_7,\tfrac{1}{2}\!\!\!\times\!\!\rep{56}\big\}$};
        \node[gray!25!black,right] at (5.9,-0.7) {\tiny$\{E_8,\emptyset\}$};
        \node[gray!25!black,right] at (3.4,0.7) {\tiny$\{E_6,\rep{27}\}$};
        \node[gray!25!black,right] at (3.3,0.35) {\tiny$\{E_7,\rep{56}\}$};
        \node[gray!25!black,right] at (3,-0.65) {\tiny$\{E_7,\tfrac{3}{2}\!\!\!\times\!\!\rep{56}\}$};
        \node[gray!25!black,right] at (2.9,-1) {\tiny$\{F_4,\emptyset\}$};

        \node[gray!25!black] at (-3.5,-0.4) {$\mathbf{T<9}$};
        \node[gray!25!black] at (-0.2,-0.6) {$\mathbf{T\leq 273}$};
        \node[gray!25!black] at (4.85,-0.15) {$\mathbf{T\leq T_\ast}$};
    \end{tikzpicture}
    \caption{(Conjectural) map of the landscape of consistent 6D, $\N=(1,0)$ supergravities. The space of anomaly-free theories with $T<9$ is finite. The central cape contains only finitely many infinite families and has a universally bounded number of tensor multiplets, $T\leq 273$. The mainland is populated by infinitely many infinite families, all with $T\gg9$ and $|G|\gg 1$, but travelling from the cape onto the mainland requires having one or more (and often many) of the indicated simple factors.}
    \label{fig:map}
\end{figure}
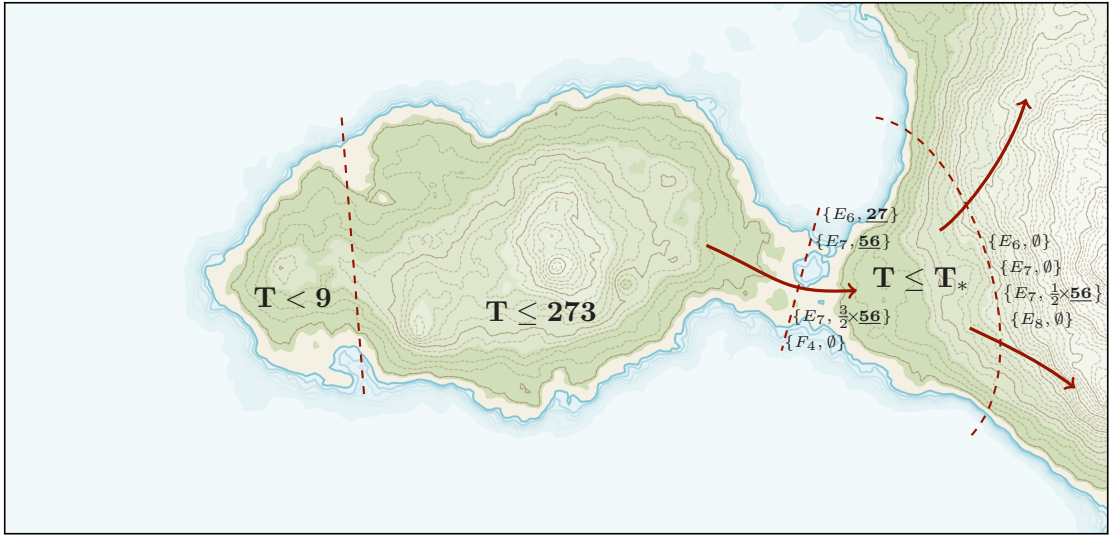

We take inspiration from the techniques of previous works such as~\cite{Bhardwaj:2013qia} and~\cite{Kumar2011} and proceed, presenting a general framework for the exhaustive enumeration of consistent theories. We implement this explicitly in two scenarios:
\begin{enumerate}
    \item $T=0$ and any gauge group built from $A_{n\geq 3}$, $B_{n\geq 3}$, $C_{n\geq 2}$, $D_{n\geq 4}$, $E_{6,7,8}$, $F_4$ and $G_2$.
    \item All $T$ and any gauge group built from $A_{4\to24}$, $B_{3\to16}$, $C_{3\to16}$, $D_{4\to16}$, $E_{6,7,8}$, $F_4$ and $G_2$, but with some choices of hypermultiplets for exceptional groups forbidden.
\end{enumerate}
In both cases we only allow for hypermultiplets charged under at most two simple factors. Notice that $A_1\sim\SU(2)$ and $A_2\sim\SU(3)$ are absent: these groups are subject to much weaker anomaly-cancellation conditions and also have irreps of very low dimension which would make the omission of hypermultiplets charged under $3+$ simple factors especially egregious. We implicitly ignore these low-rank groups except in section~\ref{sec:consistency_conditions}, which we keep fairly general, and in section~\ref{sec:finite_landscape}, where we speculate. When we turn to the more general classification (2) with $T$ unconstrained we remove $A_3\sim\SU(4)$ and $C_2\sim\Sp(2)$ for a similar reason: both of these have a four-dimensional irreducible representation that makes their inclusion computationally challenging in practice. We allow for any semi-simple gauge group built from any of the simple groups listed above, making no restrictions on either the number of simple factors or total rank. In the general classification (2) we also bound the ranks of individual simple factors: there is no fundamental obstacle to extending our results to include more simple groups.

\bigskip

The remainder of this paper is organized as follows. In section~\ref{sec:consistency_conditions} we review the consistency conditions required of non-abelian theories. Section~\ref{sec:infinite_families} is devoted to discussing the infinite families which present a technical obstacle to enumerating all consistent theories. In section~\ref{sec:classification} we outline the classification strategy in terms of $k$-cliques in a multigraph to be described, leaving many of the details to the appendices. We present results for the two classifications mentioned above in section~\ref{sec:results}, discussing features of the ensembles which retroactively justify some claims made in previous sections and highlighting some choice consistent theories. Finally, we conclude in section~\ref{sec:discussion}. In an effort to be somewhat self-contained, we review some useful results from the literature in the appendices. Accompanying code and a database of cliques may be found at~\cite{Loges:2023gh1} and~\cite{Loges:2023gh2}.

\section{Consistency conditions}
\label{sec:consistency_conditions}

In this section we review the consistency conditions that are required of 6D supergravity theories with eight supercharges. The backbone of our understanding comes from anomaly cancellation, which is very constraining for chiral theories in six dimensions. For our purposes, a ``theory'' consists of a choice of non-abelian gauge group $G$ with $k$ simple factors $G_i$, number of tensor multiplets $T$, and hypermultiplet representations $\H$.

The massless spectrum includes one gravity multiplet which contains a self-dual 2-form field, $T$ tensor multiplets each containing an anti-self-dual 2-form field, a vector multiplet in the adjoint representation of $G$, and hypermultiplets in representations of $G$ to be chosen. Charges under the $T+1$ (anti-)self-dual fields lie in a full-rank lattice $\Gamma\subset\R^{1,T}$; we will write $\Omega_{\alpha\beta}$ for the metric in this space with signature $(+,(-)^T)$ and abbreviate $\Omega_{\alpha\beta}X^\alpha Y^\beta$ as $X\cdot Y$.

We define the trace indices $A_R,B_R,C_R$ for the representation $R$ of the simple group $G_i$ via
\begin{equation}\label{eq:ABC_def}
    \beta_i\tr_R(\F_i^2) = A_R\tr(\F_i^2) \,, \qquad \beta_i^2\tr_R(\F_i^4) = B_R\tr(\F_i^4) + C_R[\tr(\F_i^2)]^2 \,,
\end{equation}
where the normalization constants $\beta_i$ are given in table~\ref{tab:lambdas} and ``$\tr$'' with no subscript refers to the fundamental or defining representation. In appendix~\ref{app:trace_indices} we review how these indices may be efficiently computed for any representation $R$~\cite{Okubo82}. With this normalization we have $A_\adj = 2h^\vee$, where $h^\vee$ is the dual Coxeter number, and $A_R,B_R,C_R$ are nearly always integers; the only exceptions are for $A_1\sim\SU(2)$, $A_2\sim\SU(3)$, $B_3\sim\SO(7)$ and $D_4\sim\SO(8)$, where $C_R\in\frac{1}{2}\Z$. In addition, all representations of $A_1\sim\SU(2)$, $A_2\sim\SU(3)$, $E_6$, $E_7$, $E_8$, $F_4$ and $G_2$ have $B_R=0$ since these groups have no independent fourth-order Casimir invariant~\cite{Okubo82}. All other groups will have $(A_R,B_R,C_R)=(\beta_i,\beta_i^2,0)$ for the fundamental representation, by definition.

Gravitational and gauge anomalies are captured by the anomaly polynomial $\hat{I}_8$ which has contributions from all of the multiplets mentioned above,
\begin{equation}
    \hat{I}_8(\mathcal{R},\F) = \hat{I}_\text{gravity}(\mathcal{R}) - T\,\hat{I}_\text{tensor}(\mathcal{R}) + \hat{I}_{1/2}^\adj(\mathcal{R},\F) - \sum_R n_R\hat{I}_{1/2}^R(\mathcal{R},\F) \,,
\end{equation}
where $n_R\geq 0$ gives the number of hypermultiplets in the representation $R$. The individual contributions are~\cite{Alvarez-Gaume:1983ihn}
\begin{equation}
    \begin{aligned}
        i(2\pi)^3\hat{I}_\text{gravity}(\mathcal{R}) &= \frac{1}{5760}\left[273\tr\mathcal{R}^4 - \frac{255}{4}(\tr\mathcal{R}^2)^2\right] \,,\\
        i(2\pi)^3\hat{I}_\text{tensor}(\mathcal{R}) &= \frac{1}{5760}\left[29\tr\mathcal{R}^4 - \frac{35}{4}(\tr\mathcal{R}^2)^2\right] \,,\\
        i(2\pi)^3\hat{I}_{1/2}^R(\mathcal{R},\F) &= \frac{H_R}{5760}\left[\tr\mathcal{R}^4 + \frac{5}{4}(\tr\mathcal{R}^2)^2\right] + \frac{1}{24}\tr_R\F^4 - \frac{1}{96}\tr_R\F^2\tr\mathcal{R}^2 \,.
    \end{aligned}
\end{equation}
\begin{table}[t]
    \centering
    \begin{tabular}{c|ccccccccc}
        $G_i$ & $A_n$ & $B_n$ & $C_n$ & $D_n$ & $E_6$ & $E_7$ & $E_8$ & $F_4$ & $G_2$\\
        \midrule
        $\beta_i$ & $1$ & $2$ & $1$ & $2$ & $6$ & $12$ & $60$ & $6$ & $2$
    \end{tabular}
    \caption{Normalization constants for the classical and exceptional Lie groups.}
    \label{tab:lambdas}
\end{table}
For brevity and uniformity in the notation we have introduced $H_R=\dim(R)$. These anomalies may be cancelled by means of the Green-Schwarz-West-Sagnotti mechanism if $\hat{I}_8$ factors in terms of a vector $Y_4$ of $4$-forms as~\cite{Green:1984sg,Green:1984bx,Sagnotti:1992qw}
\begin{equation}\label{eq:I8=YY}
    i(2\pi)^3\hat{I}_8(\mathcal{R},\F) \;\sim\; Y_4\cdot Y_4 \,, \qquad Y_4^\alpha = -\frac{1}{2}b_0^\alpha \tr(\mathcal{R}^2) + 2\sum_{i=1}^k \frac{b_i^\alpha}{\beta_i}\tr(\F_i^2) \,,
\end{equation}
where $b_I\in\R^{1,T}$ ($I=0,i$) are a collection of $k+1$ anomaly vectors.\footnote{Our ``$b_0$'' is often denoted by ``${-a}$''.} Matching terms on either side of the above gives the following:
\begin{equation}
\label{eq:anomaly_cancellation}
    \begin{aligned}
        H-V + 29T &= 273 \,, & b_0\cdot b_0 &= 9-T \,,\\
        \sum_R n_R^iB_R^i - B_\adj^i &= 0 \,, \qquad & b_0\cdot b_i &= \frac{1}{6}\Big( \sum_R n_R^iA_R^i - A_\adj^i \Big) \,,\\
        & & b_i\cdot b_i &= \frac{1}{3}\Big( \sum_R n_R^iC_R^i - C_\adj^i \Big) \,,\\
        & & b_i\cdot b_j &= \sum_{R,S}n_{(R,S)}^{i,j}A_R^iA_S^j \,, \quad (i \neq j) \,.
    \end{aligned}
\end{equation}
The total number of hypermultiplets, $H$, includes those in the trivial representation of $G$. In what follows we leave these implicit in $\H$ and introduce $\Delta=H_\text{charged}-V$, in which case the first line on the left above is equivalent to
\begin{equation}
\label{eq:gravitational_anomaly}
    \boxed{\quad \Big. \Delta + 29T \leq 273 \,. \quad}
\end{equation}
The second line on the left, which comes from requiring that the indecomposable $\tr\F_i^4$ terms are absent for each simple factor of $G$, gives us our second condition, which we will refer to as the \emph{$B$-constraint}:
\begin{equation}
\label{eq:B_constraint}
    \boxed{\quad \bigg. \sum\nolimits_Rn_R^iB_R^i = B_\adj^i \,, \quad \forall i\in\{1,2,\ldots,k\} \,. \quad}
\end{equation}
For simple groups with no non-trivial quartic invariant the index $B_R$ is zero for all representations and the $B$-constraint is trivially satisfied: this occurs only for $\SU(2)$, $\SU(3)$ and the five exceptional groups.

The anomaly lattice $\Lambda=\bigoplus_{I=0}^kb_I\Z$ is a sublattice of $\Gamma\subset\R^{1,T}$ and the four lines on the right side of~\eqref{eq:anomaly_cancellation} determine the inner products of the anomaly vectors $b_I$ only in terms of the matter spectrum. These may be collected into a $(k+1)\times(k+1)$ Gram matrix $\Gram$ with entries
\begin{equation}
    \Gram_{IJ} \equiv b_I\cdot b_J \,.
\end{equation}
It was shown in~\cite{Kumar:2010ru} that, remarkably, the entries of $\Gram$ are always integers, i.e.\ $\Lambda$ is an integral lattice, whenever all local anomalies are cancelled.\footnote{This is a statement of representation theory and does not actually impose any additional constraints on the allowed hypermultiplet representations.} Checking for anomaly-cancellation then amounts to determining if there actually exist $k+1$ vectors $b_I\in\R^{1,T}$ which realize the required inner products. A necessary and sufficient condition for $\Gram$ to be realizable is
\begin{equation}
\label{eq:eigenvalue_bounds_G}
    \boxed{\quad \Big.n_+^\Gram\leq1 \quad\text{and}\quad n_-^\Gram\leq T \,, \quad}
\end{equation}
where $n_\pm^\Gram$ denote the number of positive and negative eigenvalues of $\Gram$. This is fairly restrictive since a real-symmetric matrix having so few positive eigenvalues is non-generic. In a similar vein, we should expect the number of anomaly-free theories to drop very quickly for $k\gtrsim T$ since this requires a conspiracy in the entries of $\Gram$ to produce many zero-eigenvalues.

For the groups $\SU(2)$, $\SU(3)$ and $G_2$ there are potential global anomalies that must be cancelled. One way to understand this is as a consequence of their non-trivial homotopy groups,
\begin{equation}
    \pi_6\big(\SU(2)\big)\cong\Z/12\Z \,, \qquad \pi_6\big(\SU(3)\big)\cong\Z/6\Z \,, \qquad \pi_6(G_2)\cong\Z/3\Z \,,
\end{equation}
which impose that $\sum_Rn_RC_R-C_\adj$ is zero modulo $12$, $6$ or $3$, respectively~\cite{Bershadsky:1997sb}. However, in the modern understanding of anomalies in terms of spin bordism groups this is puzzling in view of the following facts:\footnote{See also \cite{Basile:2023zng} for the very recent study.}
\begin{equation}
    \Omega_7^\text{spin}\big(B\!\SU(2)\big) \cong 0 \,, \qquad \Omega_7^\text{spin}\big(B\!\SU(3)\big) \cong 0 \,, \qquad \Omega_7^\text{spin}(BG_2) \cong 0 \,.
\end{equation}
A careful analysis of the Green-Schwarz-West-Sagnotti mechanism or imposing Dirac quantization (see, e.g.,~\cite{Ohmori:2014kda} section 3.1.1) nevertheless leads to the same modular conditions. Detailed comparisons of the two perspectives are provided in~\cite{Lee:2020ewl,Davighi:2020kok}. Of these three simple groups we will only be considering $G_2$, which gives us our next condition:
\begin{equation}
\label{eq:G2_global}
    \boxed{\quad \bigg.\sum\nolimits_Rn_R^iC_R^i - C_\adj^i \equiv 0\mod{3} \,, \quad \forall i:G_i=G_2 \,. \quad}
\end{equation}
For quaternionic (a.k.a.\ pseudo-real) representations there is the possibility of forming half-hypermultiplets, i.e.\ taking $n_R$ to be half-integer:
\begin{equation}
\label{eq:half-integers}
    \boxed{\quad \Big.n_R\text{ half-integer} \quad\implies\quad R\text{ quaternionic} \,.  \quad}
\end{equation}
Quaternionic representations only appear for $A_{4n+1}\sim\SU(4n+2)$, $B_{4n+1}\sim\SO(8n+3)$, $B_{4n+2}\sim\SO(8n+5)$, $C_n\sim\Sp(n)$, $D_{4n+2}\sim\SO(8n+4)$ and $E_7$. Some quaternionic representations for low-rank groups are shown in table~\ref{tab:quaternionic_irreps}. In what follows we will refer to the \emph{integers} $n_R$ but it should be understood that $n_R$ can always be half-integer for quaternionic representations.\footnote{Equivalently, one could redefine $(H_R,A_R,B_R,C_R)\to(H_R/2,A_R/2,B_R/2,C_R/2)$ for these representations and let $n_R$ denote the integer number of \emph{half}-hypermultiplets.}

\begin{table}[t]
    \centering
    \begin{tabular}{ll}
        \toprule
        Group & Quaternionic irreps\\
        \midrule
        $A_5\sim\SU(6)$ & $\rep{20}$, $\rep{540}$, $\rep{980}$, \ldots \\
        $A_9\sim\SU(10)$ & $\rep{252}$, $\rep{21000}$, \ldots \\[7pt]
        $B_5\sim\SO(11)$ & $\rep{32}$, $\rep{320}$, $\rep{1408}$, \ldots \\
        $B_6\sim\SO(13)$ & $\rep{64}$, $\rep{768}$, $\rep{4160}$, \ldots \\[7pt]
        $C_2\sim\Sp(2)$ & $\rep{4}$, $\rep{16}$, $\rep{20}$, $\rep{40}$, $\rep{56}$, $\rep{64}$, $\rep{80}$, $\rep{120}$, $\rep{140}$, $\rep{160}$, \ldots \\
        $C_3\sim\Sp(3)$ & $\rep{6}$, $\rep[\prime]{14}$, $\rep{56}$, $\rep{64}$, $\rep{126}$, $\rep{216}$, $\rep{252}$, $\rep{330}$, \ldots \\
        $C_4\sim\Sp(4)$ & $\rep{8}$, $\rep{48}$, $\rep{120}$, $\rep{160}$, $\rep{288}$, $\rep{792}$, $\rep[\prime\prime]{792}$, \ldots \\
        $C_5\sim\Sp(5)$ & $\rep{10}$, $\rep{110}$, $\rep{132}$, $\rep{220}$, $\rep{320}$, $\rep{1408}$, \ldots \\
        $C_6\sim\Sp(6)$ & $\rep{12}$, $\rep{208}$, $\rep{364}$, $\rep{560}$, $\rep{572}$, \ldots \\[7pt]
        $D_6\sim\SO(12)$ & $\rep{32}$, $\rep{352}$, $\rep{1728}$, \ldots \\
        $D_{10}\sim\SO(20)$ & $\rep{512}$, $\rep{9728}$, \ldots \\[7pt]
        $E_7$ & $\rep{56}$, $\rep{912}$, $\rep{6480}$, \ldots \\
        \bottomrule
    \end{tabular}
    \caption{Quaternionic irreps for simple groups of low rank.}
    \label{tab:quaternionic_irreps}
\end{table}

The next consistency condition comes from insisting that the gauge kinetic functions have the correct sign. By supersymmetry, the kinetic terms for the gauge fields are of the form $-j\cdot b_i\tr(\F_i^2)$ for some common, moduli-dependent, time-like vector $j$. Therefore we should require
\begin{equation}
\label{eq:j_condition}
    \boxed{\quad \Big.\exists\;j\in\R^{1,T} \text{ such that } j\cdot j > 0 \text{ and } j\cdot b_i > 0 \,. \quad}
\end{equation}
Of course we do not have direct access to the anomaly vectors $b_i$, only to their inner products which are computed from the low-energy spectrum alone, and indeed there can be multiple (inequivalent) choices for the vectors which give the same Gram matrix $\Gram$. In appendix~\ref{app:j_exists} we rephrase~\eqref{eq:j_condition} in terms of the convex hull of the vectors $b_i$ and prove that if neither inequality in~\eqref{eq:eigenvalue_bounds_G} is saturated then such a $j$ can \emph{always} be found which satisfies not only $j\cdot b_i>0$ but also the analogous inequality for $b_0$, namely $j\cdot b_0>0$. This additional inequality ensures that the coefficient of the Gauss-Bonnet term is strictly positive, and can be motivated to hold based on scattering amplitudes~\cite{Cheung:2016wjt,Hamada:2018dde}, the second law of thermodynamics~\cite{Jacobson:1993xs,Liko:2007vi, Sarkar:2010xp}, weak gravity conjecture~\cite{Aalsma:2022knj}, consistency of EFT strings~\cite{Martucci:2022krl}, and holography~\cite{Ong:2022mmm}. For cases where one or more of the bounds in~\eqref{eq:eigenvalue_bounds_G} is saturated, we will not insist upon having $j\cdot b_0>0$. 

The next condition stems not from anomaly cancellation but from the unimodularity of $\Gamma$, which is required by the self-consistency of the theory upon reduction to two or four dimensions~\cite{Seiberg2011}. Since $\Lambda$ is a sublattice of $\Gamma$, we should have
\begin{equation}
\label{eq:unimodular_embedding}
    \boxed{\quad \Big.
    \exists \text{ unimodular, full-rank lattice }\Gamma\subset \R^{1,T} \text{ with } \Lambda\subseteq\Gamma \,.
    \quad}
\end{equation}
In appendix~\ref{app:lattices} we review how this condition can be checked, again given only the Gram matrix $\Gram$ and not the vectors $b_I$ directly. This is facilitated by there being a very simple classification of unimodular lattices with indefinite signature.

All of equations~\eqref{eq:B_constraint},~\eqref{eq:eigenvalue_bounds_G} and~\eqref{eq:unimodular_embedding} impose conditions on the lattice $\Lambda$ and its associated Gram matrix and, importantly, clearly must hold for any sublattice as well. This property will be crucial for us when we turn to constructing anomaly-free theories recursively in later sections. Of particular interest will be the sublattice generated only by the $b_i$ since the corresponding inner products are independent of $T$. The associated Gram matrix we write as
\begin{equation}
    \gram_{ij} \equiv b_i\cdot b_j
\end{equation}
and can be obtained from $\Gram$ by simply deleting the first row and column which are associated with $b_0$. Like for $\Gram$, we write the number of positive and negative eigenvalues of $\gram$ as $n_\pm^\gram$: clearly we have $n_+^\gram \leq n_+^\Gram$ and $n_-^\gram \leq n_-^\Gram$.

\medskip

The consistency conditions discussed in this section are summarized in table~\ref{tab:conditions}, where we also first introduce terminology that will be motivated at the beginning of section~\ref{sec:classification}. We should emphasize that there very well may be additional consistency conditions which should be imposed (such as those argued for in~\cite{Tarazi2021}), shrinking the list of allowed theories.

\begin{table}[t]
    \centering
    \textbf{Admissible theories}\\[5pt]
    \begin{tabular}{cccccccc}
        $\tr\F^4$ anomaly & ${}^\exists b_I$ & global anomaly & positivity & unimodularity \\
        \midrule
        \eqref{eq:B_constraint} & \eqref{eq:eigenvalue_bounds_G} & (\ref{eq:G2_global}, \ref{eq:half-integers}) & \eqref{eq:j_condition} & \eqref{eq:unimodular_embedding} \\
        & &
    \end{tabular}
    
    \textbf{Anomaly-free theories} \\[5pt]
    \begin{tabular}{cccccccc}
        $\tr\mathcal{R}^4$ anomaly & $\tr\F^4$ anomaly & ${}^\exists b_I$ & global anomaly & positivity & unimodularity \\
        \midrule
        \eqref{eq:gravitational_anomaly} & \eqref{eq:B_constraint} & \eqref{eq:eigenvalue_bounds_G} & (\ref{eq:G2_global}, \ref{eq:half-integers}) & \eqref{eq:j_condition} & \eqref{eq:unimodular_embedding} 
    \end{tabular}
    \caption{Summary of consistency conditions we impose. Theories satisfying (\ref{eq:B_constraint}, \ref{eq:eigenvalue_bounds_G}, \ref{eq:G2_global}, \ref{eq:half-integers}, \ref{eq:j_condition}, \ref{eq:unimodular_embedding}) are said to be admissible, and theories satisfying (\ref{eq:gravitational_anomaly}, \ref{eq:B_constraint}, \ref{eq:eigenvalue_bounds_G}, \ref{eq:G2_global}, \ref{eq:half-integers}, \ref{eq:j_condition}, \ref{eq:unimodular_embedding}) are said to be anomaly-free. See also section~\ref{sec:classification} for the terminology. Note that the condition for global anomaly (\ref{eq:G2_global}, \ref{eq:half-integers}) is redundant given $\Lambda$ is unimodular~\cite{Lee:2020ewl,Davighi:2020kok}, but it is convenient to impose in practice anyways.}
    \label{tab:conditions}
\end{table}

\section{Infinite families}
\label{sec:infinite_families}

It was shown in~\cite{Kumar:2010ru} that the number of anomaly-free theories with $T<9$ is finite. However, for $T\geq9$ the finiteness proof does not go through for a simple reason: there are in fact infinitely many anomaly-free theories with $T\geq 9$. In this section we discuss two infinite classes of theories which appear for $T\geq 9$. The first serves as a quick warm-up, having been discussed previously (e.g.\ see~\cite{Kumar:2010ru}) and is easily accounted for in the classification of section~\ref{sec:results}. The second is a generalization of the family described in section~4.2 of~\cite{Kumar:2010ru} which has gauge group $G=E_8^k$ and $T\gtrsim k$ for arbitrarily large $k$. The basic idea is that because the exceptional groups trivially satisfy the $B$-constraint, one can augment the gauge group by many exceptional factors while including no additional hypermultiplets, incurring a very large negative contribution to $\Delta$. In fleshing out this idea, we will identify a complete list of four ``culprit'' theories for which this construction works more generally.

To warm up, consider the class of theories formed in the following way; starting with a ``seed'' theory with gauge group $G$, hypermultiplets $\H$ and $n_+^\Gram=0$ (which in particular requires $T\geq 9$) which is anomaly-free, then the following theory,
\begin{equation}
    G' = G\times G_\text{aux} \,, \qquad T'=T \,, \qquad \H'=(\H,\rep{1}) + (\rep{1},\adj) \,,\label{eq:infinite_series1}
\end{equation}
is also anomaly-free for \emph{any} auxiliary semi-simple group $G_\text{aux}$. This is easy to understand since the auxiliary adjoint vector multiplets and hypermultiplets combine into a full (non-chiral) $\mathcal{N}=(1,1)$ vector multiplet and hence have no effect on anomaly cancellation: the $B$-constraint is trivially satisfied for all simple factors of $G_\text{aux}$, $\Delta$ and $T$ are unchanged and the Gram matrix is augmented by rows and columns of all zeros. The condition $n_+^\Gram=0$ is required because having $b_i$ orthogonal to all $b_I$ implies $b_i=0$ if $n_+^\Gram=1$, violating $j\cdot b_i>0$. However, for every choice of seed theory there is an upper bound on $\dim G_\text{aux}$ provided by BPS string probes, which we review in appendix~\ref{app:brane_probe}.

\medskip

The second infinite class of anomaly-free theories is constructed in a similar way by starting with a seed theory, but as we will see the seed theory is far less constrained and need not be anomaly-free itself. As mentioned above, the idea will be to augment the seed theory with enough exceptional groups, decreasing $\Delta$ dramatically, so as to become anomaly-free. This is possible only because exceptional groups automatically satisfy the $B$-constraint so that additional vector multiplets can be added without any accompanying hypermultiplets. For concreteness in the following we motivate the construction using $E_8$ with no hypermultiplets, but as we will see the argument works essentially identically for the four simple theories of table~\ref{tab:culprits} and can be extended to arbitrary combinations thereof in the obvious way.

To be more precise, the seed theory (with gauge group $G$ of $k$ simple factors, $T$ tensor multiplets and hypermultiplets $\H$) needs to satisfy only~\eqref{eq:B_constraint}, \eqref{eq:eigenvalue_bounds_G}, \eqref{eq:G2_global}, \eqref{eq:half-integers} and \eqref{eq:j_condition}: the conditions of equations~\eqref{eq:gravitational_anomaly} and~\eqref{eq:unimodular_embedding} are not necessary. In particular, vectors $b_I$ and $j$ exist which both realize the required inner products of $\Gram$ of the seed and satisfy $j\cdot j>0$ and $j\cdot b_i>0$. We augment the seed theory by $m$ auxiliary $E_8$ factors and allow the number of tensor multiplets to increase:
\begin{equation}
\label{eq:aug_theory}
    G' = G\times E_8^m \,, \qquad T'\geq T \,, \qquad \H' = (\H,\rep{1}^m) \,.
\end{equation}
Our task is to show that one can always find $T'$ and $m$ which ensure the above theory is anomaly-free. The simplest condition to state is for the gravitational anomaly, which reads
\begin{equation}
    \Delta' + 29T' = \Delta - 248m + 29T' \leq 273 \,.
\end{equation}
The other conditions revolve around the anomaly lattice, for which the Gram matrix is now
\begin{equation}
    \Gram' = \begin{pmatrix}
        9-T' & b_0\cdot b_i & -10_m\\
        b_0\cdot b_i & b_i\cdot b_j & 0_{k\times m}\\
        -10_m & 0_{m\times k} & -12\,\I_{m\times m}
    \end{pmatrix} \,,
\end{equation}
where $i,j=1,2,\ldots,k$ run over the $k$ factors of $G$. We can construct $k+m+1$ vectors $b_I'\in\R^{1,T'}$ which realize $\Gram'$ as
\begin{equation}
    \begin{aligned}
        b_0' &= (b_0,c_0) \,,\\
        b_i' &= (b_i,0) \,, & i &= 1,2,\ldots,k \,,\\
        b_{k+r}' &= (0,c_r) \,, \qquad & r &= 1,2,\ldots,m \,,
    \end{aligned}
\end{equation}
using the decomposition $\R^{1,T'}=\R^{1,T}\times\R^{(T'-T)}$. The $m$ vectors $c_r$ must be orthogonal, each with norm-squared equal to $12$, and $c_0$ must have norm-squared $T'-T$ and satisfy $c_0\cdot c_r=10$. This clearly requires $T'-T\geq m$, but if we write $c_0=c_\perp + \frac{10}{12}\sum_rc_r$ with $c_\perp\cdot c_r=0$, then we actually need the stronger condition
\begin{equation}
    \|c_\perp\|^2 = (T'-T) - \frac{10^2}{12}m \geq 0 \,.
\end{equation}
If the above inequality is satisfied then we have explicitly constructed vectors $b_I'$ which realize $\Gram'$ and therefore equation~\eqref{eq:eigenvalue_bounds_G} must be satisfied. We can similarly explicitly construct a vector $j'$ satisfying $j'\cdot j'>0$ and $j'\cdot b_i'>0$ as
\begin{equation}
    j' = (j,0) - \epsilon\sum_{r=1}^mb_{k+r}' \,, \qquad 0<\epsilon < \sqrt{\frac{j\cdot j}{12m}} \,.
\end{equation}
It is straightforward to check that all of $j'\cdot b_i'$, $j'\cdot b_{k+r}'$ and $j'\cdot j'$ are positive and that $j'\cdot b_0'$ is positive if $j\cdot b_0$ is. Finally, $T'\geq (k+m+1)+3$ is sufficient to ensure that~\eqref{eq:unimodular_embedding} is satisfied as well (see appendix~\ref{app:lattices}).

\medskip

In summary, the inequalities
\begin{equation}
\label{eq:Tm_region_E8}
    \max\left\{T+\frac{10^2}{12}\,m,\;k+m+4\right\} \;\leq\; T' \;\leq\; \frac{273-\Delta + 248m}{29}
\end{equation}
are sufficient to ensure the theory of~\eqref{eq:aug_theory} is anomaly-free. The region of parameters which gives anomaly-free theories is not only non-empty but unbounded: for large enough $m$ the interval for $T'$ is always non-trivial, no matter the values of $\Delta$, $T$ and $k$ for the seed theory. This relies crucially on the inequality $\frac{248}{29} > \frac{10^2}{12}$ for $E_8$, i.e.
\begin{equation}
    \frac{V}{29} > \frac{(A_\adj/6)^2}{(C_\adj/3)} \,.
\end{equation}
Rerunning the above construction for a general auxiliary simple theory in place of $\{E_8,\emptyset\}$, one finds the requirement $b_\text{aux}\cdot b_\text{aux} < 0$ and that equation~\eqref{eq:Tm_region_E8} should be replaced by
\begin{equation}
\label{eq:Tm_region}
    \max\left\{T+\frac{(b_0\cdot b_\text{aux})^2}{(-b_\text{aux}\cdot b_\text{aux})}\,m,\;k+m+4\right\} \;\leq\; T' \;\leq\; \frac{273-\Delta + (-\Delta_\text{aux})m}{29} \,.
\end{equation}
The set of $T'$ and $m$ satisfying the above inequalities will be unbounded provided
\begin{equation}
\label{eq:infinite_family_inequalities}
    b_\text{aux}\cdot b_\text{aux} < 0 \,, \qquad \frac{(-\Delta_\text{aux})}{29} > \max\left\{\frac{(b_0\cdot b_\text{aux})^2}{(-b_\text{aux}\cdot b_\text{aux})},\;1\right\} \,.
\end{equation}
The only choices for simple group and hypermultiplets which satisfy these inequalities are those listed in table~\ref{tab:culprits}.

\begin{table}[t]
    \centering
    \begin{tabular}{lcccccc}
        \toprule
        $G_i$ & $\H_i$ & $\Delta_i$ & $b_0\cdot b_i$ & $b_i\cdot b_i$ & $\frac{(-\Delta_i)}{29}$ & $\frac{(b_0\cdot b_i)^2}{(-b_i\cdot b_i)}$ \\
        \midrule
        $E_6$ &                             & $-78$  & $-4$  & $-6$  & $2.6896...$ & $2.6666...$ \\
        $E_7$ &                             & $-133$ & $-6$  & $-8$  & $4.5862...$ & $4.5\phantom{000...}$ \\
        $E_7$ & $\frac{1}{2}\times\rep{56}$ & $-105$ & $-5$  & $-7$  & $3.6206...$ & $3.5714...$ \\
        $E_8$ &                             & $-248$ & $-10$ & $-12$ & $8.5517...$ & $8.3333...$\\
        \bottomrule
    \end{tabular}
    \caption{All simple theories which satisfy the inequalities of equation~\eqref{eq:infinite_family_inequalities} and thus lead to an infinite family of anomaly-free theories for each seed theory.}
    \label{tab:culprits}
\end{table}

\medskip

It is worth emphasizing just how unruly this infinite family of anomaly-free theories potentially is. If one picks any gauge group $G$ and hypermultiplets $\H$ subject only to the $B$-constraints, $n_+^\gram\leq 1$ and $\det\gram\neq 0$, then provided a vector $j$ can be found satisfying $j\cdot j>0$ and $j\cdot b_i>0$ we can ensure that $n_+^\Gram\leq 1$ by taking $T$ large enough, since\footnote{Here $\gram^\mathrm{c}$ is the cofactor matrix of $\gram$.}
\begin{equation}
    \det\Gram = (9-T)\det\gram - (b_0\cdot b_i)(\gram^\mathrm{c})_{ij}(b_0\cdot b_j) \quad\xrightarrow{T\gg 9}\quad -T\det\gram \,.
\end{equation}
This shows that $\Gram$ will have one additional negative eigenvalue compared to $\gram$ and thus must have $n_+^\Gram=n_+^\gram\leq 1$.\footnote{The condition $\det\gram\neq0$ can be relaxed: if $\det\gram$ vanishes then $\det\Gram$ is independent of $T$, but it can still happen that $n_+^\Gram\leq 1$.} The above theory can then serve as the seed in the above construction. However, notice from table~\ref{tab:culprits} that the inequality $\frac{(-\Delta_i)}{29}>\frac{(b_0\cdot b_i)^2}{(-b_i\cdot b_i)}$ is only ever marginally satisfied: an allowable window for $T'$ will appear only for large $m$. Theories built in this way will thus generically have a huge number of tensor multiplets and gauge groups with very large rank.

\medskip

Let us conclude this section with an explicit example, using one of the other rows from table~\ref{tab:culprits}. Start with the following seed theory,
\begin{equation}
\label{eq:infinite_family_example}
    \newcommand{\myspacer}{\hspace{-11pt}}
    \begin{aligned}
        G &= \SU(10) \,,\myspacer\\
        \H &= 7\times\rep{10} &&\myspacer+ 5\times\rep{45} &&\myspacer+ 8\times\rep{120} &&\myspacer+ 3\times\rep{210} &&\myspacer+ \tfrac{1}{2}\times\rep{252} &&\myspacer+ 2\times\rep{825} &&\myspacer+ \rep{990} &&\myspacer+ 2\times\rep{1848}\\
        &= 7\times\ydiagram{1} &&\myspacer+ 5\times\ydiagram{1,1} &&\myspacer+ 8\times\ydiagram{1,1,1} &&\myspacer+ 3\times\ydiagram{1,1,1,1} &&\myspacer+ \tfrac{1}{2}\times\ydiagram{1,1,1,1,1} &&\myspacer+ 2\times\ydiagram{2,2} &&\myspacer+ \ydiagram{2,1,1} &&\myspacer+ 2\times\ydiagram{2,1,1,1} \,,\\[-18pt]
        \Delta &= 8248 \,,\\
        \Gram &= \begin{psmallmatrix*}
            9-T & 473\\
            473 & 895
        \end{psmallmatrix*} \,,\myspacer\myspacer
    \end{aligned}
\end{equation}
which satisfies the $B$-constraint and can serve as the seed theory for $T\geq 3$ since one can choose the vectors $b_I$ and $j$ to be the following:
\begin{equation}
    b_0 = (3;1^T) \,, \qquad b_1 = (103;-77,-44,-43,0^{T-3}) \,, \qquad j = (1;0^T) \,.
\end{equation}
Augmenting this seed theory to $G'=\SU(10)\times E_6^m$ with no additional hypermultiplets, equation~\eqref{eq:Tm_region} in this case can be brought to the form
\begin{equation}
    232(m-12,\!093) \leq 87(T'-32,\!251) \leq 234(m-12,\!093) \,.
\end{equation}
The smallest choice for $m$ and $T'$ which yields an anomaly-free theory is $m=12,\!093$ and $T'=32,\!251$. Indeed, for these values one has $\Delta'+29T'=273$ exactly and the $12,\!095$-dimensional lattice $\Lambda'$ with Gram matrix
\begin{equation}
    \Gram' = \begin{pmatrix}
        -32,\!242 & 473 & -4_{12,093}\\
        473 & 895 & 0_{12,093}\\
        -4_{12,093} & 0_{12,093} & -6\,\I_{12,093\times12,093}
    \end{pmatrix}
\end{equation}
easily embeds into the odd unimodular lattice of signature $(1,32251)$. For example, one can choose
\begin{equation*}
    \begin{array}{rrrrrrrrrrrl}
        b_0 = \big( & 3; & 1, & 1, & 1; & & (\phantom{-}2, & 2, & 0, & 0)^{4031}\hspace{-18pt}\phantom{,} & ; & 0^{16,124} \big) \,, \\
        b_1 = \big( & 103; & -77, & -44, & -43; & & \multicolumn{4}{c}{\phantom{{}^{16,124}}0^{16,124}} & ; & 0^{16,124} \big) \,,\\
        b_{3\ell-1} = \big( & 0; & 0, & 0, & 0; & (0,0,0,0)^{\ell-1}, & (\phantom{-}1, & 1, & -2, & 0), & (0,0,0,0)^{4031-\ell} ;& 0^{16,124} \big) \,,\\
        b_{3\ell} = \big( & 0; & 0, & 0, & 0; & (0,0,0,0)^{\ell-1}, & (\phantom{-}2, & 0, & 1, & 1), & (0,0,0,0)^{4031-\ell} ;& 0^{16,124} \big) \,,\\
        b_{3\ell+1} = \big( & 0; & 0, & 0, & 0; & (0,0,0,0)^{\ell-1}, & (\phantom{-}0, & 2, & 1, & -1), & (0,0,0,0)^{4031-\ell} ;& 0^{16,124} \big) \,,\\[5pt]
        j = \big( & 90; & 0, & 0, & 0; & & (-1, & -1, & 0, & 0)^{4031}\hspace{-18pt}\phantom{,} & ; & 0^{16,124} \big) \,,
    \end{array}
\end{equation*}
with $\ell=1,2,\ldots,4031$.

\section{Classification}
\label{sec:classification}

The guiding principle for the classification is the fact that anomaly-free theories with semi-simple gauge groups with $k$ factors can always be decomposed into $k$ theories with simple gauge groups. This has been used by previous analyses under several guises, e.g.\ as the ``block'' decomposition in~\cite{Kumar2011}. For example, the following theory
\begin{equation}
\label{eq:decompose_example}
    \begin{aligned}
        G &= \SU(8)\times\SU(10)\times\SU(9) \,,\\
        \H &= 5\times(\rep{8},\rep{1},\rep{1}) + (\rep{28},\rep{1},\rep{1}) + (\rep{56},\rep{1},\rep{1})\\
        &\qquad + 2\times(\rep{1},\rep{45},\rep{1}) + (\rep{1},\rep{1},\rep{9}) + (\rep{1},\rep{1},\rep{45}) + 2\times(\rep{8},\rep{10},\rep{1})\\
        \Delta &= 186\\
        T &= 3\\
        \Gram &= \begin{psmallmatrix*}
             6 & 5 & 2 & -1 \\
             5 & 3 & 2 &    \\
             2 & 2 & 0 &    \\
            -1 &   &   & -1
        \end{psmallmatrix*} \,,
    \end{aligned}
\end{equation}
is anomaly-free\footnote{With $\Omega=\diag(1,-1,-1,-1)$, one can choose, for example, $b_0=(3,1,1,1)$, $b_1=(2,0,0,1)$, $b_2=(1,1,0,0)$, $b_3=(0,0,1,0)$ and $j=(2,0,-1,0)$.} and can be decomposed as
\begin{equation}
    \begin{aligned}
        G_{1,2} &= \SU(8)\times\SU(10) \,, & \Delta_{1,2} &= 212 \,, & \Gram_{1,2} &= \begin{psmallmatrix*}
            6 & 5 & 2\\ 5 & 3 & 2\\ 2 & 2 & 0
        \end{psmallmatrix*} \,,\\[-5pt]
        \H_{1,2} &= 5\times(\rep{8},\rep{1}) + (\rep{28},\rep{1}) + (\rep{56},\rep{1}) + 2\times(\rep{1},\rep{45}) + 2\times(\rep{8},\rep{10}) \,, \hspace{-110pt} \\[5pt]
        G_3 &= \SU(9) \,, \quad \H_3 = \rep{9} + \rep{45} \,, & \Delta_3 &= -26 \,, \quad & \Gram_3 &= \begin{psmallmatrix*}
            6 & -1\\ -1 & -1
        \end{psmallmatrix*} \,,
    \end{aligned}
\end{equation}
or even further as
\begin{equation}
    \begin{aligned}
        G_1 &= \SU(8) \,, & \H_1 &= 25\times\rep{8} + \rep{28} + \rep{56} \,, & \Delta_1 &= 221 \,, & \Gram_1 &= \begin{psmallmatrix*}
            6 & 5\\ 5 & 3
        \end{psmallmatrix*} \,,\\
        G_2 &= \SU(10) \,, \quad & \H_2 &= 16\times\rep{10} + 2\times\rep{45} \,, \quad & \Delta_2 &= 151 \,, & \Gram_2 &= \begin{psmallmatrix*}
            6 & 2\\ 2 & 0
        \end{psmallmatrix*} \,,\\
        G_3 &= \SU(9) \,, & \H_3 &= \rep{9} + \rep{45} \,, & \Delta_3 &= -26 \,, \quad & \Gram_3 &= \begin{psmallmatrix*}
            6 & -1\\ -1 & -1
        \end{psmallmatrix*} \,.
    \end{aligned}
\end{equation}
This fully decomposed form must be accompanied by the following ``merging rule'' which describes which bi-charged hypermultiplets to construct in order to recover~\eqref{eq:decompose_example}:
\begin{equation}
    20\times(\rep{8},\rep{1},\rep{1}) + 16\times(\rep{1},\rep{10},\rep{1}) \quad\longrightarrow\quad 2\times (\rep{8},\rep{10},\rep{1}) \,.
\end{equation}
This merging has the effect of both decreasing the number of hypermultiplets by $160$ and introducing the non-zero off-diagonal entry $\Gram_{13} = b_1\cdot b_3 = 2$. We emphasize that although~\eqref{eq:decompose_example} satisfies the gravitational bound $\Delta+29T\leq 273$, not all of the theories encountered in its decomposition do as well.

This motivates the following terminology, which we will make extensive use of. We say that a theory with group $G$, hypermultiplets $\H$ and $T$ tensor multiplets is \emph{admissible} if it satisfies all of the boxed conditions of equations~\eqref{eq:B_constraint}, \eqref{eq:eigenvalue_bounds_G}, \eqref{eq:G2_global}, \eqref{eq:half-integers}, \eqref{eq:j_condition} and~\eqref{eq:unimodular_embedding}. Theories which additionally cancel the gravitational anomaly by satisfying equation~\eqref{eq:gravitational_anomaly} are said to be \emph{anomaly-free}. Finally, we say a theory is \emph{simple} if its gauge group is simple. Otherwise, we will refer to a theory by its value of $k$, the number of simple factors. By a slight abuse of terminology we also say a theory is admissible (anomaly-free) without specifying a value of $T$ if there exists \emph{at least one} $T\geq 0$ for which that the theory with $G$ and $\H$ is admissible (anomaly-free). In this language, the above motivating example shows that anomaly-free theories can always be decomposed into $k$ simple admissible theories (which may or may not be anomaly-free). We should emphasize that \emph{anomaly-free} theories are not just anomaly-free but also satisfy the positivity requirement of~\eqref{eq:j_condition} and unimodularity requirement of~\eqref{eq:unimodular_embedding}.

\subsection{On multigraphs and cliques}
\label{sec:graph_features}

There is a natural representation of this ``decomposition data'' as a multigraph\footnote{In graph theory, a multigraph is a graph which is permitted to have multiple edges between the same two vertices and edges from a vertex to itself.} where vertices correspond to simple admissible theories and edges correspond to different ways to merge the hypermultiplets of two simple admissible theories together to form a $k=2$ admissible theory.\footnote{One way to incorporate hypermultiplets charged under three or more simple factors would be to generalize to a hypergraph where hyperedges connect more than two vertices.} Admissible theories for higher $k$ then correspond to $k$-cliques in this multigraph since, at a minimum, they must result in an admissible theory when restricted to any two of their simple factors. Very broadly, the approach we take is to start by constructing $k\leq 2$ admissible theories and then use the resulting multigraph structure to ``bootstrap'' our way to $(3+)$-cliques (i.e.\ $k\geq3$ admissible theories).

\begin{table}[t]
    \centering
    \begin{tabular}{clcccc}
        \toprule
        $\;\;G_i$ & $\;\;\H_i$ & $\Delta_i$ & $b_i\cdot b_i$ & $b_0\cdot b_i$ & Notes\\
        \midrule
        Any      & $\adj$                          & $0$ & $0$ & $0$ \\[10pt]
        $\SU(N)$ & $(2N)\times\rep{N}$                     & $N^2+1$ & $-2$ & $0$ \\
        $\SU(N)$ & $(N-8)\times\rep{N}+\rep{N(N+1)/2}$     & $\frac{(N-7)(N-8)}{2}-27$ & $-1$ & $-1$ & $N\geq 8$ \\
        $\SU(N)$ & $(N+8)\times\rep{N}+\rep{N(N-1)/2}$     & $\frac{(N+7)(N+8)}{2}-27$ & $-1$ & $1$\\
        $\SU(N)$ & $16\times\rep{N}+2\times\rep{N(N-1)/2}$ & $15N+1$ & $0$  & $2$\\
        $\SU(N)$ & $\rep{N(N-1)/2}+\rep{N(N+1)/2}$ & $1$ & $0$ & $0$\\
        $\SU(6)$ & $15\times\rep{6}+\frac{1}{2}\times\rep{20}$          & $65$ & $-1$ & $1$\\
        $\SU(6)$ & $17\times\rep{6}+\rep{15}+\frac{1}{2}\times\rep{20}$ & $92$ & $0$  & $2$\\
        $\SU(6)$ & $18\times\rep{6}+\rep{20}$                           & $93$ & $0$  & $2$ \\
        $\SU(6)$ & $\rep{6}+\frac{1}{2}\times\rep{20}+\rep{21}$     & $2$ & $0$ & $0$ \\[10pt]
        
        $\SO(N)$ & $(N-8)\times\rep{N}$ & $\frac{(N-7)(N-8)}{2}-28$  & $-4$  & $-2$ & $N\geq 8$ \\
        $\SO(N)$ & \tiny$(N-7)\times\rep{N}+\big(2^{\lfloor\frac{10-N}{2}\rfloor}\big)\times\rep{2^{\lfloor\frac{N-1}{2}\rfloor}}$ & $\frac{(N-6)(N-7)}{2}-5$  & $-3$  & $-1$ & $7\leq N\leq 12$ \\
        $\SO(N)$ & \tiny$(N-6)\times\rep{N}+\big(2\cdot 2^{\lfloor\frac{10-N}{2}\rfloor}\big)\times\rep{2^{\lfloor\frac{N-1}{2}\rfloor}}$ & $\frac{(N-5)(N-6)}{2}+17$  & $-2$  & $0$ & $7\leq N\leq 13$ \\
        $\SO(N)$ & \tiny$(N-5)\times\rep{N}+\big(3\cdot 2^{\lfloor\frac{10-N}{2}\rfloor}\big)\times\rep{2^{\lfloor\frac{N-1}{2}\rfloor}}$ & $\frac{(N-4)(N-5)}{2}+38$  & $-1$  & $1$ & $7\leq N\leq 12$ \\
        $\SO(N)$ & \tiny$(N-4)\times\rep{N}+\big(4\cdot 2^{\lfloor\frac{10-N}{2}\rfloor}\big)\times\rep{2^{\lfloor\frac{N-1}{2}\rfloor}}$ & $\frac{(N-3)(N-4)}{2}+58$  & $0$  & $2$ & $7\leq N\leq 14$ \\[10pt]
        
        $\Sp(N)$ & $(2N+8)\times\rep{2N}$               & $N(2N+15)$ & $-1$  & $1$ \\
        $\Sp(N)$ & $16\times\rep{2N}+\rep{(N-1)(2N+1)}$ & $30N-1$    & $0$   & $2$ \\[10pt]
        
        $E_6$    & $k\times\rep{27}$                    & $27k-78$   & $k-6$ & $k-4$ & $k\leq 6$ \\
        $E_7$    & $\frac{k}{2}\times\rep{56}$          & $28k-133$  & $k-8$ & $k-6$ & $k\leq 8$ \\
        $E_8$    &                                      & $-248$     & $-12$ & $-10$ \\
        $F_4$    & $k\times\rep{26}$                    & $26k-52$   & $k-5$ & $k-3$ & $k\leq 5$ \\
        $G_2$    & $(3k+1)\times\rep{7}$                & $21k-7$    & $k-3$ & $k-1$ & $k\leq3$ \\
        \bottomrule
    \end{tabular}
    \caption{All type-\B{B} vertices, i.e.\ all simple theories with $b_i\cdot b_i\leq 0$.}
    \label{tab:typeB}
\end{table}

Let us make this idea more precise. The directed multigraph $\G=(\V,\E,\h)$ consists of a set of vertices $\V$, a set of edges $\E$ and a map $\h:\E\to\V\times\V$ which associates each edge to an ordered pair of vertices. Each vertex $\v\in\V$ represents a simple admissible theory and consists of a choice of simple gauge group and hypermultiplets (from which all other relevant data may be derived):
\begin{equation}
    \v = \Big\{G_i \,,\; \H_i \,;\; \Delta_i \,,\; b_0\cdot b_i \,,\; b_i\cdot b_i \Big\} \,.
\end{equation}
Every $k=2$ admissible theory which decomposes into the vertices $\v_i$ and $\v_j$ plus a (possibly empty) collection of merging rules,
\begin{equation}
    \Big\{\;(n_\ell \dim R_{j,\ell})\times(R_{i,\ell},\rep{1}) + (n_\ell \dim R_{i,\ell})\times(\rep{1},R_{j,\ell}) \;\;\longrightarrow\;\; n_\ell\times(R_{i,\ell},R_{j,\ell})\;\Big\}_\ell \,,
\end{equation}
is identified with an edge $\e\in\E$ incident to the participating vertices (i.e.\ $\h(\e)=(\v_i,\v_j)$). Each edge consists of the end-products of the mergings, $\H_{ij}=\sum_\ell n_\ell\times(R_{i,\ell},R_{j,\ell})$ from which the decrease in number of hypermultiplets (denoted by $\delta H\leq0$) and off-diagonal Gram matrix entry may be derived:
\begin{equation}
    \e = \Big\{ \H_{ij} \,;\; \delta H \,,\; b_i\cdot b_j \Big\} \,.
\end{equation}
Importantly, there can be multiple edges which connect the same two vertices but represent different mergings. We will refer to edges which connect a vertex to itself as \emph{self-edges} and to edges for which $\H_{ij}=\emptyset$ as \emph{trivial}.

By construction, the multigraph $\G$ directly represents admissible theories with up to two simple factors. We can identify admissible theories with $k$ simple factors with $k$-cliques in $\G$, where for our purposes ``$k$-clique'' means something slightly different than it usually does. For us a $k$-clique $\C$ is a multigraph homomorphism $\C:\K_k\to\G$, where $\K_k$ is the complete graph\footnote{The complete graph on $k$ vertices has a unique edge connecting every pair of distinct vertices.} on $k$ vertices, modulo the symmetries of $\K_k$. That is, $\C$ sends the $k$ vertices $\v_i^\K$ of $\K_k$ to vertices of $\G$ and sends the $\binom{k}{2}$ edges $\e_{ij}^\K$ of $\K_k$ to edges of $\G$ in such a way that
\begin{equation}
    \h\big(\C(\e_{ij}^\K)\big) = \big(\C(\v_i^\K),\C(\v_j^\K)\big) \quad\text{or}\quad \h\big(\C(\e_{ij}^\K)\big) = \big(\C(\v_j^\K),\C(\v_i^\K)\big)
\end{equation}
for all $1\leq i<j\leq k$.\footnote{If there were no self-edges then vertices of $\G$ would appear at most once and we could simply say that a $k$-clique is a subgraph of $\G$ isomorphic to $\K_k$. As it stands, multiple vertices of $\K_k$ may map to the same vertex in $\G$.} More colloquially, a $k$-clique is a way to label, up to symmetries, the vertices and edges of $\K_k$ using vertices and edges of $\G$ in such a way that the edge incidences are representative of $\G$: this is how we will depict (and think about) cliques going forward. We emphasize that this notion of clique is quite general and allows for collections of vertices and edges which correspond to theories which are anomalous or have negative hypermultiplet multiplicities after accounting for all of the merging rules: we will see examples of these phenomena in section~\ref{sec:graph_example}. Of course, we are only interested in admissible (and ultimately anomaly-free) theories and \emph{in}admissible cliques will be the first to be discarded when we come to the classification techniques of section~\ref{sec:clique_construction}.

\medskip

There are several features of $\G$ which are worth highlighting as they will play an important role in what follows. Vertices can be categorized into two types based on their diagonal Gram matrix entries: vertices which have $b_i\cdot b_i>0$ and $b_i\cdot b_i\leq 0$ we call \emph{type-\A{A}} and \emph{type-\B{B}}, respectively. For each simple group there are only a handful of vertices of type~\B{B}: see table~\ref{tab:typeB}. Distinguishing between type-\A{A} and type-\B{B} vertices is useful in part because a vertex's degree (i.e.\ the number of incident edges) in $\G$ depends dramatically on its type. While type-\B{B} vertices are few in number, they have very high degree, having edges to most other type-\B{B} vertices and to many type-\A{A} vertices as well. In contrast, most vertices are of type-\A{A} but generally have very low degree since there can only be \emph{non-trivial} edges between type-\A{A} vertices\footnote{Having a trivial edge between two type-\A{A} vertices would produce a $2\times2$ principle submatrix of $\Gram$ of the form $\begin{psmallmatrix*}
    + & 0\\ 0 & +
\end{psmallmatrix*}$ which immediately gives $n_+^\Gram\geq 2$.} and this requires there to be enough hypermultiplets available to merge. Figure~\ref{fig:vertices_scatter} shows some of the vertices of $\G$ organized by their values of $\Delta_i$ and $b_i\cdot b_i$. Although not immediately obvious from the above characterization in terms of $b_i\cdot b_i$, all type-\A{A} vertices have $\Delta_i$ strictly positive. The smallest value of $\Delta_i$ for a type-\A{A} vertex occurs for $\{G_2,\,\rep{7}+\rep{27}\}$ which has $\Delta_i=20$ and is admissible for $T\geq 9$. In addition, all but~$29$ type-\B{B} vertices have $\Delta_i\geq 0$.

\begin{figure}[t]
    \centering
    \includegraphics[width=\textwidth]{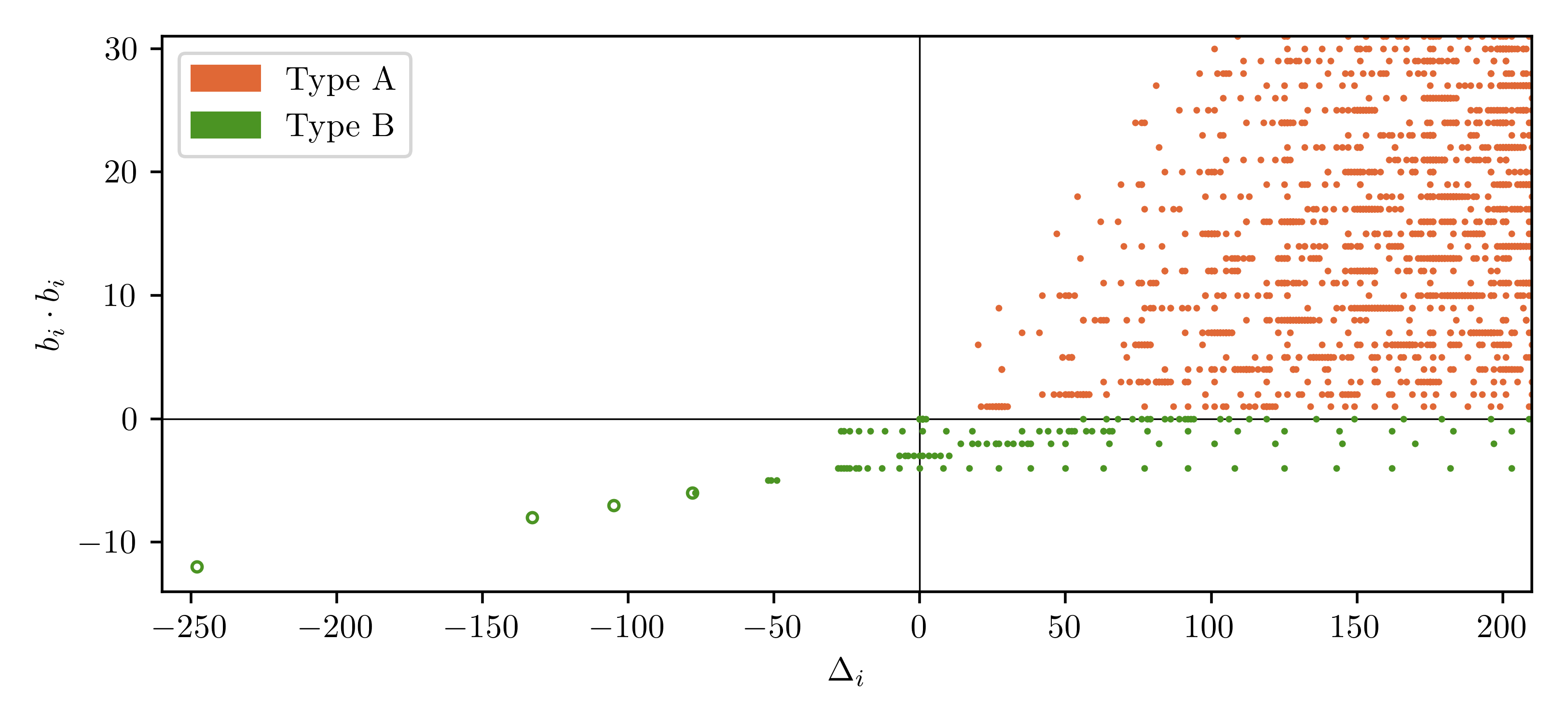}
    \caption{The low-lying vertices of $\G$ organized by their values of $\Delta_i$ and $b_i\cdot b_i$. The four vertices of table~\ref{tab:culprits} which are responsible for the infinite family discussed in section~\ref{sec:infinite_families} are shown with circles.}
    \label{fig:vertices_scatter}
\end{figure}

\medskip

It will be useful to associate to each clique a value $T_\text{min}$, defined as
\begin{equation}
\label{eq:Tmin_defn}
    T_\text{min}(\C) \equiv \inf\!\left\{T\geq 0\;\big|\;\C\text{ with }T\text{ tensors is admissible}\right\} \in\{0,1,2,\ldots\}\cup\{+\infty\} \,,
\end{equation}
where our terminology extends to cliques in the obvious way: we say a clique $\C$ with $T$ tensor multiplets is admissible (anomaly-free) if the theory it represents is admissible (anomaly-free). Similarly, a clique is admissible (anomaly-free) if there exists one or more $T\geq 0$ such that $\C$ with $T$ tensor multiplets is admissible (anomaly-free). That is, $\C$ is admissible iff $T_\text{min}(\C)<+\infty$ and anomaly-free iff $\Delta(\C)+29T_\text{min}(\C) \leq 273$. An important fact is that if $\C$ contains $\C'$ as a sub-clique then $T_\text{min}(\C)\geq T_\text{min}(\C')$, since adding additional vectors $b_i$ to the lattice can only ever increase the required value of $T$ needed to satisfy equations~\eqref{eq:eigenvalue_bounds_G}, \eqref{eq:j_condition} and~\eqref{eq:unimodular_embedding}.

\subsection{Multigraph construction}
\label{sec:graph_construction}

Strictly speaking, the multigraph $\G$ is infinite: for each simple group there are infinitely many choices of hypermultiplets which satisfy the $B$-constraint which then appear as infinitely many vertices in $\G$. In practice we will truncate to a finite list of simple groups and only include type-\A{A} vertices with values $\Delta_i$ below some bound $\Delta_\text{max}$, which we allow to depend on the corresponding simple group $G_i$. For now, we outline how the vertices and edges of $\G$ can be constructed up to some to-be-determined $\Delta_\text{max}(G_i)$, leaving some of the details to the appendices.

Constructing all type-\A{A} vertices for a simple group $G_i$ up to some bound $\Delta_\text{max}$ proceeds in a few steps. First, all irreps $R$ of $G_i$ with $H_R \leq H_\text{max} \equiv \Delta_\text{max} + V_i$ are enumerated and their trace indices $A_R$, $B_R$ and $C_R$ computed (conjugate irreps are redundant and removed). This can be done very quickly, even to very large $H_\text{max}$, by using backtracking on the highest-weight vector and the results of~\cite{Okubo82} which we review in appendix~\ref{app:trace_indices}. Having generated this list of irreps, we are faced with the following combinatorial problem: given a set of vectors $\{\vec{z}_R = (H_R,B_R)\}_R$ in the $HB$-plane, find all convex, integer, linear combinations of the form
\begin{equation}
    (H,B) = \sum_Rn_R\vec{z}_R \,, \qquad n_R\in\Z_{\geq 0} \,,
\end{equation}
which lie on the line segment $[0,H_\text{max}]\times\{B_\adj\}$. This can be solved by adding vectors one-by-one to candidate solutions in a judiciously chosen order and carefully narrowing the allowed region around the target line segment after each step: see appendix~\ref{app:vertex_construction} for more details.

\medskip

Given that all vertices of $\G$ have been constructed (up to some bounds on $\Delta_i$), adding edges is in principle straightforward: for each choice of two vertices (with repetition to allow for self-edges) and all possible mergings of their hypermultiplets, check the resulting $k=2$ theory for admissibility. However, na\"ively iterating over all $\binom{|\V|+1}{2}$ choices of vertices can become prohibitively slow when the number of vertices is large. Most vertices are of type \A{A} and ultimately have very few (if any) edges to other type-\A{A} vertices. If $\v_i$ and $\v_j$ are two type-\A{A} vertices then any edge connecting them must be non-trivial and have $b_i\cdot b_j>0$ large enough to ensure
\begin{equation}
    \det\begin{pmatrix}
        b_i\cdot b_i & b_i\cdot b_j\\
        b_i\cdot b_j & b_j\cdot b_j
    \end{pmatrix} \leq 0 \quad\implies\quad (b_i\cdot b_j)^2 \geq (b_i\cdot b_i)(b_j\cdot b_j)
\label{eq:bibj_lower_bound}\end{equation}
and thus avoid $n_+^\Gram\geq 2$. As we show in appendix~\ref{app:edge_construction}, one can find an upper bound $b_i\cdot b_j \leq \sigma_{i,j}^\text{max}$ which depends on the vertex $\v_i$ and (crucially) only the \emph{gauge group} $G_j$. Given a type-\A{A} vertex $\v_i$ one then only needs to check for edges to vertices with gauge group $G_j$ which are either type-\B{B} or are type-\A{A} and have $0<b_j\cdot b_j \leq \frac{(\sigma_{i,j}^\text{max})^2}{b_i\cdot b_i}$. In practice this upper bound is quite severe and the number of pairs of vertices to check is greatly reduced from the na\"ive $\mathcal{O}(|\V|^2)$.

\subsection{\emph{Intermission}: A three-vertex example}
\label{sec:graph_example}

\begin{figure}[t]
    \centering
    \begin{tabular}{cccccc}
        \toprule
        $\v_i$ & $G_i$ & $\H_i$ & $\Delta_i$ & $b_i\cdot b_i$ & $b_0\cdot b_i$ \\
        \midrule
        $\v_\A{1}$ & $\SU(8)$ & $16\times\rep{8}+3\times\rep{28}$ & $149$ & $1$ & $3$\\
        $\v_\B{2}$ & $\SU(16)$ & $8\times\rep{16} + \rep{136}$ & $9$ & $-1$ & $-1$\\
        $\v_\B{3}$ & $\Sp(4)$ & $16\times\rep{8}$ & $92$ & $-1$ & $1$\\
        \bottomrule
    \end{tabular}\\[15pt]
    \begin{tikzpicture}[baseline={([yshift=-0.5ex]current bounding box.center)}]
        \begin{scope}[scale=1.5]
            \node[circle, inner sep=1.5, fill=typeA, draw=black, thick] (a) at (0,1.5*1.73) {};
            \node[circle, inner sep=1.5, fill=typeB, draw=black, thick] (b) at (-1.5,0) {};
            \node[circle, inner sep=1.5, fill=typeB, draw=black, thick] (c) at (1.5,0) {};
    
            \node[left] at (a) {$\v_\A{1}\,$};
            \node[left] at (b) {$\v_\B{2}\,$};
            \node[below=2pt] at (c) {$\v_\B{3}\;\;$};
    
            \draw[very thick] (a) to node[fill=white,inner sep=0pt] {$\e_{\A{1}\B{2}}$} (b);
            
            \draw[very thick] (a) to[out=-60,in=120] node[fill=white,inner sep=0pt] {$\e_{\A{1}\B{3}}$} (c);
            \draw[very thick] (a) to[out=-90,in=150] node[fill=white,inner sep=0pt] {$\e_{\A{1}\B{3}}'$} (c);
            
            \draw[very thick] (b) to node[fill=white,inner sep=0pt] {$\e_{\B{23}}$} (c);
    
            \draw[very thick] (a) to[out=0,in=50,looseness=30] node[right] {$\!\!\e_{\A{11}}$} (a);
            \draw[very thick] (a) to[out=70,in=120,looseness=30] node[left] {$\!\e_{\A{11}}'$} (a);
            \draw[very thick] (c) to[out=75,in=35,looseness=30] node[right] {$\!\!\e_{\B{33}}$} (c);
            \draw[very thick] (c) to[out=25,in=-15,looseness=30] node[right] {$\!\!\e_{\B{33}}'$} (c);
    
            \draw (a) to[out=-150,in=90] (b);
            \draw (a) to[out=-30,in=90] (c);
            \draw (b) to[out=-30,in=-150] (c);
            \draw (b) to[out=-80,in=-130,looseness=30] (b);
            \draw (c) to[out=-25,in=-65,looseness=30] (c);
        \end{scope}
    \end{tikzpicture}
    \hspace{20pt}
    \begin{tabular}{cccccc}
        \toprule
        $\e_{ij}$ & $\h(\e_{ij})$ & $\H_{ij}$ & $\delta H$ & $b_i\cdot b_j$ \\
        \midrule
        $\e_{\A{11}}$      & $(\v_\A{1},\v_\A{1})$ & $(\rep{8},\rep{8})$        & $-64$  & $1$ \\
        $\e_{\A{11}}'$     & $(\v_\A{1},\v_\A{1})$ & $2\times(\rep{8},\rep{8})$ & $-128$ & $2$ \\
        $\e_{\A{1}\B{2}}$  & $(\v_\A{1},\v_\B{2})$ & $(\rep{8},\rep{16})$       & $-128$ & $1$ \\
        $\e_{\A{1}\B{3}}$  & $(\v_\A{1},\v_\B{3})$ & $(\rep{8},\rep{8})$        & $-64$  & $1$ \\
        $\e_{\A{1}\B{3}}'$ & $(\v_\A{1},\v_\B{3})$ & $2\times(\rep{8},\rep{8})$ & $-128$ & $2$ \\
        $\e_{\B{23}}$      & $(\v_\B{2},\v_\B{3})$ & $(\rep{16},\rep{8})$       & $-128$ & $1$ \\
        $\e_{\B{33}}$      & $(\v_\B{3},\v_\B{3})$ & $(\rep{8},\rep{8})$        & $-64$  & $1$ \\
        $\e_{\B{33}}'$     & $(\v_\B{3},\v_\B{3})$ & $2\times(\rep{8},\rep{8})$ & $-128$ & $2$ \\
        \bottomrule
    \end{tabular}
    \caption{The subgraph of the multigraph $\G$ induced by the three vertices $\v_\A{1}$, $\v_\B{2}$ and~$\v_\B{3}$. The five un-named edges (thin lines in the left panel) are trivial, i.e.~have $\H_{ij}=\emptyset$. Note that there is no trivial self-edge for the type-\A{A} vertex $\v_\A{1}$.}
    \label{fig:3vertex_example}
\end{figure}

Before turning to the classification of anomaly-free cliques let us pause here to exemplify how admissible/anomaly-free theories are encoded in $\G$ by discussing the three-vertex subgraph shown in figure~\ref{fig:3vertex_example}. In this example there is one type-\A{A} vertex $\v_\A{1}$ with gauge group $\SU(8)$ and two type-\B{B} vertices $\v_\B{2}$ and $\v_\B{3}$ with gauge groups $\SU(16)$ and $\Sp(4)$, respectively. In what follows we will omit all but the relevant data in order to streamline our discussion, but we encourage the motivated reader to fill in all of the details about the cliques we highlight: everything one could hope to know about a theory can be derived from its labelling of $\K_k$ and the information in figure~\ref{fig:3vertex_example}.

With eight non-trivial edges and five trivial edges, it is manifest that there are exactly thirteen $k=2$ admissible cliques which are built from pairs of these three simple theories. In contrast, a $(3+)$-clique can fail to be admissible for several reasons. For example, the following clique is clearly inadmissible because the number of $(\rep{1},\rep{8},\rep{1})$ hypermultiplets cannot be negative:
\begin{equation}
\label{eq:example_neg_nR}
    \begin{aligned}
        \begin{tikzpicture}[baseline={([yshift=-0.5ex]current bounding box.center)}]
            \node[circle, inner sep=1.5, fill=typeA, draw=black, thick] (a) at (0,0.75*1.73) {};
            \node[circle, inner sep=1.5, fill=typeB, draw=black, thick] (b) at (-0.75,0) {};
            \node[circle, inner sep=1.5, fill=typeB, draw=black, thick] (c) at (0.75,0) {};
            
            \node[left] at (a) {$\v_\A{1}$};
            \node[left] at (b) {$\v_\B{3}$};
            \node[right] at (c) {$\v_\B{3}$};

            \draw[very thick] (a) to node[fill=white,inner sep=0pt] {\footnotesize$\e_{\A{1}\B{3}}$} (b);
            \draw (a) to (c);
            \draw[very thick] (b) to node[fill=white,inner sep=0pt] {\footnotesize$\e_{\B{33}}'$} (c);
        \end{tikzpicture} \;\; &: &\quad \H &= 8\times(\rep{8},\rep{1},\rep{1}) + 3\times(\rep{28},\rep{1},\rep{1}) + (-8)\times(\rep{1},\rep{8},\rep{1})\\[-18pt]
        && &\quad + (\rep{8},\rep{8},\rep{1}) + 2\times(\rep{1},\rep{8},\rep{8})
    \end{aligned}
\end{equation}
It can also happen that $n_+^\gram$, and thus $n_+^\Gram$, is too large:
\begin{equation}
    \begin{aligned}
        \begin{tikzpicture}[baseline={([yshift=-.5ex]current bounding box.center)}]
            \node[circle, inner sep=1.5, fill=typeA, draw=black, thick] (a) at (0,0.75*1.73) {};
            \node[circle, inner sep=1.5, fill=typeA, draw=black, thick] (b) at (-0.75,0) {};
            \node[circle, inner sep=1.5, fill=typeB, draw=black, thick] (c) at (0.75,0) {};
            
            \node[left] at (a) {$\v_\A{1}$};
            \node[left] at (b) {$\v_\A{1}$};
            \node[right] at (c) {$\v_\B{3}$};

            \draw[very thick] (a) to node[fill=white,inner sep=0pt] {\footnotesize$\e_{\A{11}}$} (b);
            \draw (a) to (c);
            \draw[very thick] (b) to node[fill=white,inner sep=0pt] {\footnotesize$\e_{\A{1}\B{3}}$} (c);
        \end{tikzpicture} \;\; &: &\quad \gram &= \begin{psmallmatrix*}
            1 & 1 &   \\
            1 & 1 &  1\\
            & 1 & -1
        \end{psmallmatrix*} \,, \quad (n_+^\gram,n_-^\gram) = (2,1) \,.
    \end{aligned}
\end{equation}
There can also be more subtle cases where $n_+^\gram \leq 1$ but $n_+^\Gram>1$ for all $T\geq 0$, such as for
\begin{equation}
    \begin{aligned}
        \begin{tikzpicture}[baseline={([yshift=-0.5ex]current bounding box.center)}]
            \node[circle, inner sep=1.5, fill=typeA, draw=black, thick] (a) at (0,0.75*1.73) {};
            \node[circle, inner sep=1.5, fill=typeB, draw=black, thick] (b) at (-0.75,0) {};
            \node[circle, inner sep=1.5, fill=typeB, draw=black, thick] (c) at (0.75,0) {};
            
            \node[left] at (a) {$\v_\A{1}$};
            \node[left] at (b) {$\v_\B{3}$};
            \node[right] at (c) {$\v_\B{3}$};

            \draw (a) to (b);
            \draw (a) to (c);
            \draw[very thick] (b) to node[fill=white,inner sep=0pt] {\footnotesize$\e_{\B{33}}$} (c);
        \end{tikzpicture} \;\; &: &\quad \Gram &= \begin{psmallmatrix*}
            9-T &  3 &  1 &  1\\
            3 &  1 &    &   \\
            1 &    & -1 &  1\\
            1 &    &  1 & -1
        \end{psmallmatrix*} \,, \quad (n_+^\Gram,n_-^\Gram) = (2,2) \,,
    \end{aligned}
\end{equation}
where $n_\pm^\Gram$ are independent of $T$. Finally, the vector $j$ may not exist, such as for the clique
\begin{equation}
    \begin{aligned}
        \begin{tikzpicture}[baseline={([yshift=-0.5ex]current bounding box.center)}]
            \node[circle, inner sep=1.5, fill=typeA, draw=black, thick] (a) at (0,0.75*1.73) {};
            \node[circle, inner sep=1.5, fill=typeB, draw=black, thick] (b) at (-0.75,0) {};
            \node[circle, inner sep=1.5, fill=typeB, draw=black, thick] (c) at (0.75,0) {};
            
            \node[left] at (a) {$\v_\A{1}$};
            \node[left] at (b) {$\v_\B{2}$};
            \node[right] at (c) {$\v_\B{3}$};

            \draw (a) to (b);
            \draw (a) to (c);
            \draw[very thick] (b) to node[fill=white,inner sep=0pt] {\footnotesize$\e_{\B{23}}$} (c);
        \end{tikzpicture} \;\; &: &\quad \Gram &= \begin{psmallmatrix*}
            9-T &  3 &  1 & -1\\
            3 &  1 &    &   \\
            1 &    & -1 &  1\\
            -1 &    &  1 & -1
        \end{psmallmatrix*} \,, \quad (n_+^\Gram,n_-^\Gram) = \begin{cases}
            (2,1) & T=0 \,,\\
            (1,1) & T=1 \,,\\
            (1,2) & T\geq 2 \,,
        \end{cases}
    \end{aligned}
\end{equation}
which can never satisfy $j\cdot b_i>0$ since here $n_+^\Gram=1$ implies that $b_2+b_3=0$.

\medskip

Of course there are also admissible cliques with $k\geq 3$, such as the following three:
\begin{equation}
    \begin{aligned}
        \newcommand{\s}{1}
        \begin{tikzpicture}[baseline={([yshift=-0.5ex]current bounding box.south)}]
            \node[circle, inner sep=1.5, fill=typeA, draw=black, thick] (a) at (0,1.753\s) {};
            \node[circle, inner sep=1.5, fill=typeA, draw=black, thick] (b) at (-\s,0) {};
            \node[circle, inner sep=1.5, fill=typeA, draw=black, thick] (c) at (\s,0) {};
            
            \node[above] at (a) {$\v_\A{1}$};
            \node[left] at (b) {$\v_\A{1}$};
            \node[right] at (c) {$\v_\A{1}$};

            \draw[very thick] (a) to node[fill=white,inner sep=0pt] {\footnotesize$\e_{\A{11}}$} (b);
            \draw[very thick] (b) to node[fill=white,inner sep=0pt] {\footnotesize$\e_{\A{11}}$} (c);
            \draw[very thick] (c) to node[fill=white,inner sep=0pt] {\footnotesize$\e_{\A{11}}$} (a);

            \node at (0,-1) {$\Delta=255 \,, \quad T_\text{min}=0$};
        \end{tikzpicture}
        \qquad
        \begin{tikzpicture}[baseline={([yshift=-0.5ex]current bounding box.south)}]
            \node[circle, inner sep=1.5, fill=typeB, draw=black, thick] (a) at (0,\s) {};
            \node[circle, inner sep=1.5, fill=typeB, draw=black, thick] (b) at (0.951*\s,0.309*\s) {};
            \node[circle, inner sep=1.5, fill=typeB, draw=black, thick] (c) at (0.588*\s,-0.809*\s) {};
            \node[circle, inner sep=1.5, fill=typeB, draw=black, thick] (d) at (-0.588*\s,-0.809*\s) {};
            \node[circle, inner sep=1.5, fill=typeB, draw=black, thick] (e) at (-0.951*\s,0.309*\s) {};
            
            \node[above] at (a) {$\v_\B{3}$};
            \node[right] at (b) {$\v_\B{3}$};
            \node[right] at (c) {$\v_\B{3}$};
            \node[left] at (d) {$\v_\B{3}$};
            \node[left] at (e) {$\v_\B{3}$};

            \draw (a) to (c) to (e) to (b) to (d) to (a);

            \draw[very thick] (a) to node[fill=white,inner sep=0pt] {\footnotesize$\e_{\B{33}}$} (b);
            \draw[very thick] (b) to node[fill=white,inner sep=0pt] {\footnotesize$\e_{\B{33}}$} (c);
            \draw[very thick] (c) to node[fill=white,inner sep=0pt] {\footnotesize$\e_{\B{33}}$} (d);
            \draw[very thick] (d) to node[fill=white,inner sep=0pt] {\footnotesize$\e_{\B{33}}$} (e);
            \draw[very thick] (e) to node[fill=white,inner sep=0pt] {\footnotesize$\e_{\B{33}}$} (a);

            \node at (0,-0.809*\s-1) {$\Delta=140 \,, \quad T_\text{min}=4$};
        \end{tikzpicture}
        \qquad
        \newcommand{\ww}{1.5}
        \begin{tikzpicture}[baseline={([yshift=-0.5ex]current bounding box.south)}]
            \node[circle, inner sep=1.5, fill=typeB, draw=black, thick] (a) at (-\ww,0.5*\ww) {};
            \node[circle, inner sep=1.5, fill=typeB, draw=black, thick] (b) at (0,0) {};
            \node[circle, inner sep=1.5, fill=typeB, draw=black, thick] (c) at (\ww,0) {};
            \node[circle, inner sep=1.5, fill=typeB, draw=black, thick] (d) at (\ww,\ww) {};
            \node[circle, inner sep=1.5, fill=typeB, draw=black, thick] (e) at (0,\ww) {};
            
            \node[left] at (a) {$\v_\B{2}$};
            \node[below] at (b) {$\v_\B{3}$};
            \node[right] at (c) {$\v_\B{3}$};
            \node[right] at (d) {$\v_\B{3}$};
            \node[above] at (e) {$\v_\B{3}$};

            \draw (b) to (a) to (c);
            \draw (d) to (a) to (e);
            
            \draw[very thick] (b) to node[fill=white,inner sep=0pt] {\footnotesize$\e_{\B{33}}$} (c);
            \draw[very thick] (c) to node[fill=white,inner sep=0pt] {\footnotesize$\e_{\B{33}}$} (d);
            \draw[very thick] (d) to node[fill=white,inner sep=0pt] {\footnotesize$\e_{\B{33}}$} (e);
            \draw[white, line width=3pt] (e) to (b);
            \draw[very thick] (e) to node[fill=white,inner sep=0pt] {\footnotesize$\e_{\B{33}}$} (b);

            \draw[white, line width=3pt] (b) to (d);
            \draw[white, line width=3pt] (c) to (e);
            \draw (b) to (d);
            \draw (c) to (e);
            
            \node at (0,-1) {$\Delta=121 \,, \quad T_\text{min}=6$};
        \end{tikzpicture}
    \end{aligned}
\end{equation}
The $3$-clique on the left is the only type-\A{A} admissible clique with $k\geq 3$ vertices and is in fact anomaly-free for $T=0$. The $5$-clique in the middle is in fact anomaly-free for $T=4$, and is ``irreducible'' in the sense that all vertices are connected via non-trivial edges. In contrast, the $5$-clique on the right has been drawn to emphasize that it is better thought of as a sort of disjoint union of a $1$-clique and $4$-clique since there are only trivial edges between the two sets of vertices. Similarly, the anomaly-free clique with the largest number of vertices ($k=12$) is best thought of as a sort of disjoint union of six $2$-cliques:
\begin{equation}
\label{eq:12-clique-example}
    \begin{aligned}
        \newcommand{\s}{1}
        \newcommand{\ww}{1.5}
        \begin{tikzpicture}[baseline={([yshift=-0.5ex]current bounding box.south)}]
            \foreach \x in {1, 2, 3, 4, 5} {
            \foreach \y in {0,...,\x} {
                \draw[opacity=0.25] (\x*\ww,0) to (\y*\ww,\s);
                \draw[opacity=0.25] (\x*\ww,\s) to (\y*\ww,0);
            };};
            \draw[opacity=0.25] (0,0) to (5*\ww,0);
            \draw[opacity=0.25] (0,\s) to (5*\ww,\s);
            
            \foreach \x in {0, 1, 2, 3, 4, 5} {
                \node[circle, inner sep=1.5, fill=typeB, draw=black, thick] (a) at (\x*\ww,0) {};
                \node[circle, inner sep=1.5, fill=typeB, draw=black, thick] (b) at (\x*\ww,\s) {};
                
                \node[below] at (a) {$\v_\B{2}$};
                \node[above] at (b) {$\v_\B{3}$};

                \draw[very thick] (a) to node[fill=white,inner sep=0pt] {\footnotesize$\e_{\B{23}}$} (b);
            };
        \end{tikzpicture}
    \end{aligned}
\end{equation}
This clique has $\Delta=-162$ and is anomaly-free for $T=15$.

\subsection{The anatomy of anomaly-free cliques}
\label{sec:clique_anatomy}

The multigraph $\G$ is dense enough that the total number of $k$-cliques grows extremely quickly with $k$. However, because we are interested in admissible (and ultimately anomaly-free) cliques, $\G$ has a lot of structure that we can leverage. In general, a clique $\C$ can be written as a sort of disjoint union,
\begin{equation}
\label{eq:clique_decomp}
    \C = \biguplus_\alpha \C_\alpha^\irr \,,
\end{equation}
where vertices of different $\C_\alpha^\irr$ are joined by trivial edges and each $\C_\alpha^\irr$ is \emph{irreducible}, by which we mean that it has a spanning tree consisting of only non-trivial edges. For example, schematically (with non-trivial edges shown with higher-weight lines),
\begin{equation}
    \C = \;
    \begin{tikzpicture}[baseline={([yshift=-0.5ex]current bounding box.center)}]
        \node[circle, inner sep=1.5, fill=typeA, draw=black, thick] (a) at (0,0) {};
        \node[circle, inner sep=1.5, fill=typeA, draw=black, thick] (b) at (0,0.5) {};
        \node[circle, inner sep=1.5, fill=typeA, draw=black, thick] (c) at (0.5,0.5) {};
        \node[circle, inner sep=1.5, fill=typeA, draw=black, thick] (d) at (0.5,0) {};
        
        \node[circle, inner sep=1.5, fill=typeB, draw=black, thick] (e) at (0.866,0.25) {};
        \node[circle, inner sep=1.5, fill=typeB, draw=black, thick] (f) at (0,-0.5) {};
        \node[circle, inner sep=1.5, fill=typeB, draw=black, thick] (g) at (0,-1) {};
        \node[circle, inner sep=1.5, fill=typeB, draw=black, thick] (h) at (0.5,-0.5) {};
        \node[circle, inner sep=1.5, fill=typeB, draw=black, thick] (i) at (1,-0.5) {};
        \node[circle, inner sep=1.5, fill=typeB, draw=black, thick] (j) at (1.366,-0.15) {};
        \node[circle, inner sep=1.5, fill=typeB, draw=black, thick] (k) at (0.75,-1) {};

        \foreach \xx in {a, b, c, d} {
            \foreach \yy in {e, f, g, h, i, j, k} {
                \draw[opacity=0.25] (\xx.center) -- (\yy.center);
            };
        };
        \foreach \xx in {e, f, g, h, i, j, k} {
            \foreach \yy in {e, f, g, h, i, j, k} {
                \draw[opacity=0.25] (\xx.center) -- (\yy.center);
            };
        };
        \draw[opacity=0.25] (a) to[out=240,in=120] (g);
        \draw[opacity=0.25] (b) to[out=240,in=120] (f);
        \draw[opacity=0.25] (b) to[out=240,in=120] (g);

        \node[circle, inner sep=1.5, fill=typeA, draw=black, thick] (a) at (0,0) {};
        \node[circle, inner sep=1.5, fill=typeA, draw=black, thick] (b) at (0,0.5) {};
        \node[circle, inner sep=1.5, fill=typeA, draw=black, thick] (c) at (0.5,0.5) {};
        \node[circle, inner sep=1.5, fill=typeA, draw=black, thick] (d) at (0.5,0) {};
        
        \node[circle, inner sep=1.5, fill=typeB, draw=black, thick] (e) at (0.866,0.25) {};
        \node[circle, inner sep=1.5, fill=typeB, draw=black, thick] (f) at (0,-0.5) {};
        \node[circle, inner sep=1.5, fill=typeB, draw=black, thick] (g) at (0,-1) {};
        \node[circle, inner sep=1.5, fill=typeB, draw=black, thick] (h) at (0.5,-0.5) {};
        \node[circle, inner sep=1.5, fill=typeB, draw=black, thick] (i) at (1,-0.5) {};
        \node[circle, inner sep=1.5, fill=typeB, draw=black, thick] (j) at (1.366,-0.15) {};
        \node[circle, inner sep=1.5, fill=typeB, draw=black, thick] (k) at (0.75,-1) {};

        \draw[very thick] (a) to (b) to (c) to (d) to (a) to (c);
        \draw[very thick] (b) to (d);

        \draw[very thick] (c) to (e) to (d);
        \draw[very thick] (a) to (f);
        \draw[very thick] (f) to (g);

        \draw[very thick] (h) to (i) to (j);

    \end{tikzpicture}
    \quad=\quad
    \left(
    \begin{tikzpicture}[baseline={([yshift=-0.5ex]current bounding box.center)}]
        \node[circle, inner sep=1.5, fill=typeA, draw=black, thick] (a) at (0,0) {};
        \node[circle, inner sep=1.5, fill=typeA, draw=black, thick] (b) at (0,0.5) {};
        \node[circle, inner sep=1.5, fill=typeA, draw=black, thick] (c) at (0.5,0.5) {};
        \node[circle, inner sep=1.5, fill=typeA, draw=black, thick] (d) at (0.5,0) {};
        
        \node[circle, inner sep=1.5, fill=typeB, draw=black, thick] (e) at (0.866,0.25) {};
        \node[circle, inner sep=1.5, fill=typeB, draw=black, thick] (f) at (0,-0.5) {};
        \node[circle, inner sep=1.5, fill=typeB, draw=black, thick] (g) at (0,-1) {};

        \foreach \xx in {a, b, c, d} {
            \foreach \yy in {e, f, g} {
                \draw[opacity=0.25] (\xx.center) -- (\yy.center);
            };
        };
        \foreach \xx in {e, f, g} {
            \foreach \yy in {e, f, g} {
                \draw[opacity=0.25] (\xx.center) -- (\yy.center);
            };
        };
        \draw[opacity=0.25] (a) to[out=240,in=120] (g);
        \draw[opacity=0.25] (b) to[out=240,in=120] (f);
        \draw[opacity=0.25] (b) to[out=240,in=120] (g);
        
        \node[circle, inner sep=1.5, fill=typeA, draw=black, thick] (a) at (0,0) {};
        \node[circle, inner sep=1.5, fill=typeA, draw=black, thick] (b) at (0,0.5) {};
        \node[circle, inner sep=1.5, fill=typeA, draw=black, thick] (c) at (0.5,0.5) {};
        \node[circle, inner sep=1.5, fill=typeA, draw=black, thick] (d) at (0.5,0) {};
        
        \node[circle, inner sep=1.5, fill=typeB, draw=black, thick] (e) at (0.866,0.25) {};
        \node[circle, inner sep=1.5, fill=typeB, draw=black, thick] (f) at (0,-0.5) {};
        \node[circle, inner sep=1.5, fill=typeB, draw=black, thick] (g) at (0,-1) {};
        
        \draw[very thick] (a) to (b) to (c) to (d) to (a) to (c);
        \draw[very thick] (b) to (d);

        \draw[very thick] (c) to (e) to (d);
        \draw[very thick] (a) to (f);
        \draw[very thick] (f) to (g);
    \end{tikzpicture}
    \right) \;\;\uplus\;\; \Big(\;\;
    \begin{tikzpicture}[baseline={([yshift=-0.5ex]current bounding box.center)}]
        \node[circle, inner sep=1.5, fill=typeB, draw=black, thick] (h) at (0.5,-0.5) {};
        \node[circle, inner sep=1.5, fill=typeB, draw=black, thick] (i) at (1,-0.5) {};
        \node[circle, inner sep=1.5, fill=typeB, draw=black, thick] (j) at (1.366,-0.15) {};

        \draw[very thick] (h) to (i) to (j);
        
        \node[circle, inner sep=1.5, fill=typeB, draw=black, thick] (h) at (0.5,-0.5) {};
        \node[circle, inner sep=1.5, fill=typeB, draw=black, thick] (i) at (1,-0.5) {};
        \node[circle, inner sep=1.5, fill=typeB, draw=black, thick] (j) at (1.366,-0.15) {};

        \draw[opacity=0.25] (h) to (j);
    \end{tikzpicture}
    \;\;\Big) \;\;\uplus\;\; \big(\;\;
    \begin{tikzpicture}[baseline={([yshift=-0.5ex]current bounding box.center)}]
        \node[circle, inner sep=1.5, fill=typeB, draw=black, thick] (k) at (0.75,-1) {};
    \end{tikzpicture}
    \;\;\big) \,,
\end{equation}
is a disjoint union of three irreducible cliques. Notice that type-\A{A} vertices must all appear in the same irreducible component because they are only ever connected to each other via non-trivial edges. After a suitable reordering of the $b_i$, these irreducible components are reflected in the Gram matrix as a block-diagonal structure,
\begin{equation}
    \Gram = \left(\begin{array}{c|ccc}
        b_0\cdot b_0 & b_0\cdot b_{\color{typeAB}j} & b_0\cdot b_{\color{typeB}l} & b_0\cdot b_{\color{typeB}n}\\ \hline
        b_0\cdot b_{\color{typeAB}i} & \gram_{{\color{typeAB}ij}}^{\color{typeAB}(1)} & 0 & 0\\
        b_0\cdot b_{\color{typeB}k} & 0 & \gram_{{\color{typeB}kl}}^{\color{typeB}(2)} & 0\\
        b_0\cdot b_{\color{typeB}m} & 0 & 0 & \gram_{{\color{typeB}mn}}^{\color{typeB}(3)}
    \end{array}\right) \quad
    \gram^{\color{typeAB}(1)} = \begin{psmallmatrix*}
        \A{+} & \A{+} & \A{+} & \A{+} & \\
        \A{+} & \A{+} & \A{+} & \A{+} & \AB{+} \\
        \A{+} & \A{+} & \A{+} & \A{+} & \AB{+} \\
        \A{+} & \A{+} & \A{+} & \A{+} & & \AB{+} \\
        & \AB{+} & \AB{+} & & \B{\leq} \\
        & & & \AB{+} & & \B{\leq} & \B{+} \\
        & & & & & \B{+} & \B{\leq}
    \end{psmallmatrix*} \quad
    \begin{aligned}
        \gram^{\color{typeB}(2)} &= \begin{psmallmatrix*}
            \B{\leq} & \B{+} \\
            \B{+} & \B{\leq} & \B{+} \\
            & \B{+} & \B{\leq}
        \end{psmallmatrix*}\\[3pt]
        \gram^{\color{typeB}(3)} &= \begin{psmallmatrix*}
            \B{\leq}
        \end{psmallmatrix*}
    \end{aligned}
\end{equation}
where `$\leq$' denotes a non-positive integer. We see that in general $\gram$, rather then $\Gram$, decomposes nicely and its spectrum has useful additive properties:
\begin{equation}
    \gram(\C) = \bigoplus_\alpha \gram(\C_\alpha^\irr) \,, \qquad n_\pm^\gram(\C) = \sum_\alpha n_\pm^\gram(\C_\alpha^\irr) \,.
\end{equation}
Similarly, since the trivial edges between the different irreducible components in~\eqref{eq:clique_decomp} do not change the total number of hypermultiplets we also have
\begin{equation}
     \Delta(\C) = \sum_\alpha \Delta(\C_\alpha^\irr) \,.
\end{equation}
In contrast, the only general statement that can be made about $T_\text{min}$ is
\begin{equation}
    T_\text{min}(\C) \geq \max\big\{n_-^\gram(\C),\,\max_\alpha \{T_\text{min}(\C_\alpha^\irr)\}\big\} \,,
\end{equation}
reflecting the fact that if $\C$ with $T$ tensor multiplets is admissible then each $\C_\alpha^\irr$ it can be decomposed into must be as well. Given a list of irreducible cliques, anomaly-free cliques can then be formed by taking all disjoint unions subject to
\begin{equation}
\label{eq:union_bounds}
    \sum_\alpha n_+^\gram(\C_\alpha^\irr) \leq 1 \,, \qquad \sum_\alpha \Delta(\C_\alpha^\irr) \leq 273
\end{equation}
and checking if $T_\text{min}$ is small enough.

\begin{table}[t]
    \begin{tabular}{llccc}
        \toprule
        $\;G$ & $\;\H$ & $\Delta$ & $\Delta+28n_-^\gram$ & $\Gram$\\
        \midrule
        $G_i$ & $\adj$ & $0$ & $0$ & $\begin{psmallmatrix}
            9-T & 0\\
            0 & 0
        \end{psmallmatrix}$\\[4pt]
        $\SU(N)$ & $\rep{N(N-1)/2}+\rep{N(N+1)/2}$ & $1$ & $1$ & $\begin{psmallmatrix}
            9-T & 0\\
            0 & 0
        \end{psmallmatrix}$\\[4pt]
        $\SU(N)\times\SU(N)$ & $2\times(\rep{N},\rep{N})$ & $2$ & $30$ & $\begin{psmallmatrix}
            9-T & 0 & 0\\
            0 & -2 & 2\\
            0 & 2 & -2
        \end{psmallmatrix}$\\[4pt]
        $\SO(2N+8)\times\Sp(N)$ & $(\rep{2N+8},\rep{2N})$ & $-28$ & $0$ & $\begin{psmallmatrix}
            9-T & -2 & 1\\
            -2 & -4 & 2\\
            1 & 2 & -1
        \end{psmallmatrix}$\\
        $\qquad\vdots$ & $\qquad\vdots$ & $\vdots$ & $\vdots$ & $\vdots$\\
        \bottomrule
    \end{tabular}
    \caption{Some well-known examples of irreducible cliques $\C_{\infty,\alpha}^\irr$ which fall into infinite families with fixed $k$, $\Delta$ and $\Gram$. In the first row $G_i$ can be any simple group (with two parameters: the Killing-Cartan series and rank).}
    \label{tab:irreducible_cliques_infty}
\end{table}

\medskip

In light of the ``unruly'' infinite family discussed in section~\ref{sec:infinite_families}, it would be fruitless to try to enumerate all irreducible admissible cliques since we know there can be arbitrarily-negative contributions to the $\Delta$-sum in~\eqref{eq:union_bounds} while remaining anomaly-free. To make progress, we can refine the decomposition of equation~\eqref{eq:clique_decomp} for anomaly-free cliques $\C^\af$ in the following way:
\begin{equation}
    \C^\af = \Big( \biguplus_\alpha \C_{<0,\alpha}^\irr \Big) \uplus \Big( \biguplus_\alpha \C_{\infty,\alpha}^\irr \Big) \uplus \Big( \biguplus_\alpha \C_{\text{gen},\alpha}^\irr \Big) \,.
\end{equation}
The irreducible cliques $\C_{<0,\alpha}^\irr$ are those for which the combination $(\Delta+28n_-^\gram)(\C_{<0,\alpha}^\irr)$ is strictly negative (the utility of this $T$-independent combination will be made evident shortly). The irreducible cliques $\C_{\infty,\alpha}^\irr$ are those which fall into infinite families with fixed $k$, $\Delta$ and $\Gram$, such as the well-known examples shown in table~\ref{tab:irreducible_cliques_infty}. In light of there being only a finite number of anomaly-free theories with $T<9$~\cite{Kumar:2010ru}, all cliques $\C_{\infty,\alpha}^\irr$ must have $T_\text{min}\geq 9$ for one reason or another. For the examples listed in table~\ref{tab:irreducible_cliques_infty}, this occurs because~\eqref{eq:j_condition} is violated unless $n_+^\Gram=0$. Finally, the irreducible cliques $\C_{\text{gen},i}^\irr$ are generic, in the sense that they do not fall into the previous two categories.

\medskip

If irreducible cliques with $\Delta+28n_-^\gram < 0$ are absent in $\C^\af$ then
\begin{equation}
\label{eq:delta_n_sum}
    (\Delta+28n_-^\gram)(\C^\af) = \sum_\alpha(\Delta+28n_-^\gram)(\C_\alpha^\irr)
\end{equation}
is a sum of non-negative terms and we have the following double-sided bounds,
\begin{equation}
    0 \leq (\Delta+28n_-^\gram)(\C_\alpha^\irr) \leq (\Delta+28n_-^\gram)(\C^\af) \leq (\Delta+28T_\text{min})(\C^\af) \leq 273 - T_\text{min}(\C^\af) \,,
\end{equation}
for each $\C_\alpha^\irr$ of type ``$\infty$'' or ``gen''. Since $T_\text{min}(\C^\af) \geq T_\text{min}(\C_\alpha^\irr)$, we conclude that
\begin{equation}
\label{eq:delta_28n_double_sided_bound}
    \boxed{\quad\Big .\C_{<0,\alpha}^\irr\text{ absent} \quad\implies\quad 0\leq (\Delta+28n_-^\gram)(\C_\alpha^\irr) \leq 273 - T_\text{min}(\C_\alpha^\irr) \quad}
\end{equation}
for all $\C_\alpha^\irr$ which appear in anomaly-free cliques. We note in passing that this implies an upper bound of $T\leq 273$ for anomaly-free cliques with no irreducible $\C_{<0,\alpha}^\irr$ components: in section~\ref{sec:finite_landscape} we will expand on this further.

\medskip

As we discussed above, because of the finiteness for $T<9$ it must be that all $\C_{\infty,\alpha}^\irr$ have $T_\text{min}\geq 9$ and there are only finitely many $\C_{\text{gen},\alpha}^\irr$ with $\Delta+28n_-^\gram\leq 273-T_\text{min}$ and $T_\text{min}<9$. We now make the following stronger claims, which will be justified \emph{ex post facto} in section~\ref{sec:results}:
\begin{itemize}
    \item There are exactly eight irreducible cliques with $\Delta+28n_-^\gram < 0$, each consisting of just a single vertex from the following list:\footnote{This list grows to include $(2+)$-cliques when $\SU(2)$ and $\SU(3)$ are reintroduced, although we expect that the list remains finite. For example, $G=\SU(2)\times E_7$ with $\H=\frac{1}{2}\times(\rep{2},\rep{56})+12\times(\rep{2},\rep{1})$ and $G=\SU(3)\times E_7$ with $\H=\tfrac{1}{2}\times(\rep{3},\rep{56}) + 2\times(\rep{3},\rep{1})$ are both admissible and have have $\Delta+28n_-^\gram=-28$ and $\Delta+28n_-^\gram=-23$, respectively. However, these two examples have $n_+^\gram=1$ and can thus appear at most once.}
    \begin{equation}
    \label{eq:removed_vertices}
        \begin{aligned}
            &\{E_6,\,\emptyset\} \,, & &\{E_7,\,\emptyset\} \,, & &\{E_7,\,\rep{56}\} \,, & &\{E_8,\,\emptyset\} \,,\\
            &\{E_6,\,\rep{27}\} \,, & &\{E_7,\,\tfrac{1}{2}\times\rep{56}\} \,, & &\{E_7,\,\tfrac{3}{2}\times\rep{56}\} \,, & &\{F_4,\,\emptyset\} \,.
        \end{aligned}
    \end{equation}
    This includes the four simple theories from table~\ref{tab:culprits} which allow for the infinite families discussed in section~\ref{sec:infinite_families}.

    \item There are only finitely many infinite families of $\C_{\infty,\alpha}^\irr$ with $\Delta+28n_-^\gram \leq 273 - T_\text{min}$.

    \item There are only finitely many $\C_{\text{gen},\alpha}^\irr$ with $\Delta+28n_-^\gram \leq 273 - T_\text{min}$.
\end{itemize}
There are only two rows in table~\ref{tab:typeB} which have $\Delta+28n_-^\gram=0$: $\{G_i,\adj\}$ has $\Delta=n_-^\gram=0$ and $\{\SO(8),\emptyset\}$ has $\Delta=-28$ and $n_-^\gram=1$ (hence the choice of `28' for the coefficient of $n_-^\gram$). That the bound $\Delta+28n_-^\gram\geq 0$ is satisfied for irreducible $(2+)$-cliques even once the vertices of~\eqref{eq:removed_vertices} have been removed is non-trivial; for example, it is marginally satisfied by the infinitely many $\SO(2N+8)\times\Sp(N)$ $2$-cliques in table~\ref{tab:irreducible_cliques_infty}.

\subsection{Constructing irreducible cliques}
\label{sec:clique_construction}

Now we finally come to the crux of the classification. Having constructed the multigraph $\G$ for some choice of simple groups, we would then like to construct all admissible, irreducible cliques which have $\Delta+28n_-^\gram$ satisfying the upper bound in equation~\eqref{eq:delta_28n_double_sided_bound}. We do this recursively, using a ``branch-and-prune'' algorithm to generate cliques of ever-increasing size: see figure~\ref{fig:branch_and_prune} for a schematic of the procedure. During the branching step $(k-1)$-cliques are augmented by vertices in their neighborhood to form $k$-cliques which are then discarded, i.e.\ pruned, if they are not admissible or have a value of $\Delta+28n_-^\gram$ which is irredeemably large. Here we will flesh out the two main steps of this algorithm; the main obstacle to overcome is how to quickly identify cliques with $\Delta+28n_-^\gram > 273$ which should \emph{not} be pruned.

First, notice that because an irreducible clique $\C$ has pre-existing non-trivial edges, the vertices of $\C$ have fewer hypermultiplets available to merge than they would otherwise. Not all of a vertex's non-trivial edges in $\G$ will be viable: recall the example of~\eqref{eq:example_neg_nR}. We call vertices of a clique \emph{active} if they still have one or more viable non-trivial edge and write $n_R^\avbl(\v)$ for the number of hypermultiplets associated to a vertex $\v$ which remain available to merge.

\medskip

Let us now discuss the branching step. Given an irreducible $(k-1)$-clique $\C$ we first identify the neighborhood of $\C$ in $\G$, i.e.\ all vertices of $\G$ which are connected by one or more (non-)trivial edge to each vertex of $\C$, and restrict to those which are connected via a viable non-trivial edge to at least one active vertex of $\C$. For each of these candidate vertices in the neighborhood of $\C$ there may be some choice in picking edges to the pre-existing vertices of $\C$; each choice of edges (with at least one being non-trivial) gives rise to an irreducible $k$-clique.

\begin{figure}[t]
    \centering
    \newcommand{\s}{0.9}
    \begin{equation*}
        \begin{tikzpicture}[baseline={([yshift=-0.5ex]current bounding box.center)}]
            \node[circle, inner sep=1.5, fill=typeB, draw=black, thick] (a) at (0,0) {};
            \node[circle, inner sep=1.5, fill=typeB, draw=black, thick] (b) at (0,\s) {};
            \node[circle, inner sep=1.5, fill=typeB, draw=black, thick] (c) at (\s,\s) {};
            \node[circle, inner sep=1.5, fill=typeB, draw=black, thick] (d) at (\s,0) {};
            
            \node[circle, inner sep=1.5, fill=typeB, draw=black, thick] (e) at (2*\s,0.5*\s) {};
            
            \draw[thick] (a) to (b) to (c) to (d) to (a) to (c);
            \draw[white, line width=3pt] (b) to (d);
            \draw[thick] (b) to (d);

            \draw[dashed, gray] (a) to[out=-45, in=-90] (e);
            \draw[dashed, gray] (b) to[out=45, in=90] (e);
            \draw[dashed, gray] (c) to[out=0, in=135] (e);
            \draw[dashed, gray] (c) to[out=-40, in=164] (e);
            \draw[dashed, gray] (d) to[out=0, in=-135] (e);
            \draw[dashed, gray] (d) to[out=40,in=195] (e);
        \end{tikzpicture}
        \;\;\longrightarrow\;\;
        \left\{
        \begin{tikzpicture}[baseline={([yshift=-0.5ex]current bounding box.center)}]
            \node[circle, inner sep=1.5, fill=typeB, draw=black, thick] (a) at (0,0) {};
            \node[circle, inner sep=1.5, fill=typeB, draw=black, thick] (b) at (0,\s) {};
            \node[circle, inner sep=1.5, fill=typeB, draw=black, thick] (c) at (\s,\s) {};
            \node[circle, inner sep=1.5, fill=typeB, draw=black, thick] (d) at (\s,0) {};
            
            \node[circle, inner sep=1.5, fill=typeB, draw=black, thick] (e) at (2*\s,0.5*\s) {};
            
            \draw[thick] (a) to (b) to (c) to (d) to (a) to (c);
            \draw[white, line width=3pt] (b) to (d);
            \draw[thick] (b) to (d);

            \draw[thick] (a) to[out=-45, in=-90] (e);
            \draw[thick] (b) to[out=45, in=90] (e);
            \draw[thick] (c) to[out=0, in=135] (e);
            \draw[thick] (d) to[out=0, in=-135] (e);
        \end{tikzpicture}
        \;,\quad
        \begin{tikzpicture}[baseline={([yshift=-0.5ex]current bounding box.center)}]
            \node[circle, inner sep=1.5, fill=typeB, draw=black, thick] (a) at (0,0) {};
            \node[circle, inner sep=1.5, fill=typeB, draw=black, thick] (b) at (0,\s) {};
            \node[circle, inner sep=1.5, fill=typeB, draw=black, thick] (c) at (\s,\s) {};
            \node[circle, inner sep=1.5, fill=typeB, draw=black, thick] (d) at (\s,0) {};
            
            \node[circle, inner sep=1.5, fill=typeB, draw=black, thick] (e) at (2*\s,0.5*\s) {};
            
            \draw[thick] (a) to (b) to (c) to (d) to (a) to (c);
            \draw[white, line width=3pt] (b) to (d);
            \draw[thick] (b) to (d);

            \draw[thick] (a) to[out=-45, in=-90] (e);
            \draw[thick] (b) to[out=45, in=90] (e);
            \draw[thick] (c) to[out=-40, in=164] (e);
            \draw[thick] (d) to[out=40,in=195] (e);

            \draw[red, line width=2pt] (0,-0.4) to (2*\s,\s+0.4);
        \end{tikzpicture}
        \;,\quad
        \begin{tikzpicture}[baseline={([yshift=-0.5ex]current bounding box.center)}]
            \node[circle, inner sep=1.5, fill=typeB, draw=black, thick] (a) at (0,0) {};
            \node[circle, inner sep=1.5, fill=typeB, draw=black, thick] (b) at (0,\s) {};
            \node[circle, inner sep=1.5, fill=typeB, draw=black, thick] (c) at (\s,\s) {};
            \node[circle, inner sep=1.5, fill=typeB, draw=black, thick] (d) at (\s,0) {};
            
            \node[circle, inner sep=1.5, fill=typeB, draw=black, thick] (e) at (2*\s,0.5*\s) {};
            
            \draw[thick] (a) to (b) to (c) to (d) to (a) to (c);
            \draw[white, line width=3pt] (b) to (d);
            \draw[thick] (b) to (d);

            \draw[thick] (a) to[out=-45, in=-90] (e);
            \draw[thick] (b) to[out=45, in=90] (e);
            \draw[thick] (c) to[out=0, in=135] (e);
            \draw[thick] (d) to[out=40,in=195] (e);

            \draw[red, line width=2pt] (0,-0.4) to (2*\s,\s+0.4);
        \end{tikzpicture}
        \;,\quad
        \begin{tikzpicture}[baseline={([yshift=-0.5ex]current bounding box.center)}]
            \node[circle, inner sep=1.5, fill=typeB, draw=black, thick] (a) at (0,0) {};
            \node[circle, inner sep=1.5, fill=typeB, draw=black, thick] (b) at (0,\s) {};
            \node[circle, inner sep=1.5, fill=typeB, draw=black, thick] (c) at (\s,\s) {};
            \node[circle, inner sep=1.5, fill=typeB, draw=black, thick] (d) at (\s,0) {};
            
            \node[circle, inner sep=1.5, fill=typeB, draw=black, thick] (e) at (2*\s,0.5*\s) {};
            
            \draw[thick] (a) to (b) to (c) to (d) to (a) to (c);
            \draw[white, line width=3pt] (b) to (d);
            \draw[thick] (b) to (d);

            \draw[thick] (a) to[out=-45, in=-90] (e);
            \draw[thick] (b) to[out=45, in=90] (e);
            \draw[thick] (c) to[out=-40, in=164] (e);
            \draw[thick] (d) to[out=0, in=-135] (e);
        \end{tikzpicture}
        \right\}
    \end{equation*}
    \caption{A branch-and-prune algorithm is used to generate $k$-cliques from $(k-1)$-cliques. During the branching step a $(k-1)$-clique is augmented by each vertex in its neighborhood in turn with all possible choices of edges to the pre-existing vertices considered. Any resulting $k$-clique which is not admissible or which has $\Delta+28n_-^\gram$ irredeemably large is pruned.}
    \label{fig:branch_and_prune}
\end{figure}
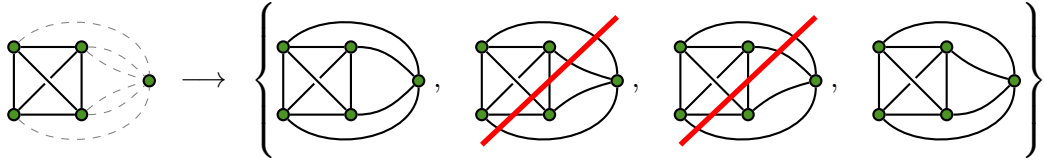

Actually, we can again leverage the differences between type-\A{A} and type-\B{B} vertices in admissible cliques to refine the branching step. The method just described (and depicted in figure~\ref{fig:branch_and_prune}) can be used to build type-\A{A} cliques and type-\B{B} cliques. However, when it comes to forming type-\A{A}\B{B} cliques which contain vertices of both types, it is beneficial to utilize the strong constraint that $n_+^\gram=1$ is entirely fixed.\footnote{If $-\gram$ is positive semi-definite, all the leading principal minors are non-negative. However, this is not the case for the type-\A{A} and type-\A{A}\B{B} cliques.} Rather than adding type-\B{B} vertices from the neighborhood one-by-one, a better strategy is to use the already-constructed type-\B{B} cliques as indivisible units that are joined to a type-\A{A} or type-\A{A}\B{B} clique all at once. To be precise, in a type-\A{A}\B{B} clique we only consider type-\A{A} vertices to be active and during the branching step we take non-trivial edges from the active type-\A{A} vertices to active vertices of a type-\B{B} clique. That is, type-\A{A}\B{B} cliques are branched schematically like the following example,
\begin{equation}
\label{eq:branch_AB}
    \begin{tikzpicture}[baseline={([yshift=-0.5ex]current bounding box.center)}]
        \node[circle, inner sep=1.5, fill=typeA, draw=black, thick] (a0) at (-0.5,0) {};
        \node[circle, inner sep=1.5, fill=typeA, draw=black, thick] (b0) at (0,0.866) {};
        \node[circle, inner sep=1.5, fill=typeA, draw=black, thick] (c0) at (0.5,0) {};
        \node[circle, inner sep=1.5, fill=typeB, draw=black, thick] (d0) at (-1.15,-0.3) {};
        \node[circle, inner sep=1.5, fill=typeB, draw=black, thick] (e0) at (-1.5,-0.9) {};
        \node[circle, inner sep=1.5, fill=typeB, draw=black, thick] (f0) at (-0.3,-0.7) {};
        \node[circle, inner sep=1.5, fill=typeB, draw=black, thick] (g0) at (-0.7,1) {};
        \node[circle, inner sep=1.5, fill=typeB, draw=black, thick] (h0) at (-0.2,1.5) {};
        
        \draw[very thick] (f0) to (a0) to (b0) to (c0) to (a0) to (d0) to (e0);
        \draw[very thick] (g0) to (b0) to (h0) to (g0);
        
        \draw (b0) circle (0.15);
        \draw (c0) circle (0.15);
        
        \node[circle, inner sep=1.5, fill=typeB, draw=black, thick] (a1) at (0.9,1.3) {};
        \node[circle, inner sep=1.5, fill=typeB, draw=black, thick] (b1) at (1.8,1.2) {};
        \node[circle, inner sep=1.5, fill=typeB, draw=black, thick] (c1) at (1.3,0.5) {};
        \draw[very thick] (a1) to (b1) to (c1);
        \draw (a1) circle (0.15);
        \draw (c1) circle (0.15);
        
        \draw[dashed, gray] (a1) to (b0) to (c1);
        \draw[dashed, gray] (c1) to (c0) to (a1);
    \end{tikzpicture}
    \;\;\longrightarrow\;\;
    \left\{
    \begin{tikzpicture}[baseline={([yshift=-0.5ex]current bounding box.center)}]
        \node[circle, inner sep=1.5, fill=typeA, draw=black, thick] (a0) at (-0.5,0) {};
        \node[circle, inner sep=1.5, fill=typeA, draw=black, thick] (b0) at (0,0.866) {};
        \node[circle, inner sep=1.5, fill=typeA, draw=black, thick] (c0) at (0.5,0) {};
        \node[circle, inner sep=1.5, fill=typeB, draw=black, thick] (d0) at (-1.15,-0.3) {};
        \node[circle, inner sep=1.5, fill=typeB, draw=black, thick] (e0) at (-1.5,-0.9) {};
        \node[circle, inner sep=1.5, fill=typeB, draw=black, thick] (f0) at (-0.3,-0.7) {};
        \node[circle, inner sep=1.5, fill=typeB, draw=black, thick] (g0) at (-0.7,1) {};
        \node[circle, inner sep=1.5, fill=typeB, draw=black, thick] (h0) at (-0.2,1.5) {};
        
        \draw[very thick] (f0) to (a0) to (b0) to (c0) to (a0) to (d0) to (e0);
        \draw[very thick] (g0) to (b0) to (h0) to (g0);
        
        \draw (b0) circle (0.15);
        \draw (c0) circle (0.15);
        
        \node[circle, inner sep=1.5, fill=typeB, draw=black, thick] (a1) at (0.9,1.3) {};
        \node[circle, inner sep=1.5, fill=typeB, draw=black, thick] (b1) at (1.8,1.2) {};
        \node[circle, inner sep=1.5, fill=typeB, draw=black, thick] (c1) at (1.3,0.5) {};
        \draw[very thick] (a1) to (b1) to (c1);
        
        \draw[very thick] (c1) to (c0);
    \end{tikzpicture}
    \;,\;
    \begin{tikzpicture}[baseline={([yshift=-0.5ex]current bounding box.center)}]
        \node[circle, inner sep=1.5, fill=typeA, draw=black, thick] (a0) at (-0.5,0) {};
        \node[circle, inner sep=1.5, fill=typeA, draw=black, thick] (b0) at (0,0.866) {};
        \node[circle, inner sep=1.5, fill=typeA, draw=black, thick] (c0) at (0.5,0) {};
        \node[circle, inner sep=1.5, fill=typeB, draw=black, thick] (d0) at (-1.15,-0.3) {};
        \node[circle, inner sep=1.5, fill=typeB, draw=black, thick] (e0) at (-1.5,-0.9) {};
        \node[circle, inner sep=1.5, fill=typeB, draw=black, thick] (f0) at (-0.3,-0.7) {};
        \node[circle, inner sep=1.5, fill=typeB, draw=black, thick] (g0) at (-0.7,1) {};
        \node[circle, inner sep=1.5, fill=typeB, draw=black, thick] (h0) at (-0.2,1.5) {};
        
        \draw[very thick] (f0) to (a0) to (b0) to (c0) to (a0) to (d0) to (e0);
        \draw[very thick] (g0) to (b0) to (h0) to (g0);
        
        \draw (c0) circle (0.15);
        
        \node[circle, inner sep=1.5, fill=typeB, draw=black, thick] (a1) at (0.9,1.3) {};
        \node[circle, inner sep=1.5, fill=typeB, draw=black, thick] (b1) at (1.6,1) {};
        \node[circle, inner sep=1.5, fill=typeB, draw=black, thick] (c1) at (1,0.6) {};
        \draw[very thick] (a1) to (b1) to (c1);
        
        \draw[very thick] (c1) to (b0) to (a1);
    \end{tikzpicture}
    \;,\;
    \ldots\;
    \right\}
\end{equation}
where the circled vertices are active and trivial edges have been suppressed. Notice that after branching all type-\B{B} vertices are made inactive and type-\A{A} vertices may or may not remain active.

The list of type-\B{B} irreducible cliques that needs to be considered when branching type-\A{A}\B{B} cliques is quite limited. If we write $\overline{\gram}$ for the matrix formed by deleting rows and columns corresponding to active vertices from $\gram$, then the signature of the Gram matrix which results from the above branching step $\C_1^{\A{A}\B{B}}+\C_2^{\B{B}}\to\C^{\A{A}\B{B}}$ can be bounded as
\begin{equation}
    n_\pm^\gram(\C^{\A{A}\B{B}}) \geq n_\pm^\gram(\C_1^{\A{A}\B{B}}) + n_\pm^{\overline{\gram}}(\C_2^{\B{B}}) \quad \text{and}\quad n_\pm^\gram(\C^{\A{A}\B{B}}) \geq n_\pm^{\overline{\gram}}(\C_1^{\A{A}\B{B}}) + n_\pm^\gram(\C_2^{\B{B}}) \,,
\end{equation}
since the active vertices of either of $\C_1^{\A{A}\B{B}}$ or $\C_2^{\B{B}}$ serve as a vertex cover of the newly-added non-trivial edges. In particular, since $n_+^\gram(\C_1^{\A{A}\B{B}})=1$, we must have $n_+^{\overline{\gram}}(\C_2^{\B{B}})=0$.

\begin{figure}[t]
    \centering
    \begin{tikzpicture}
        \node[draw,thick,fill=gray!15!white] (nodetypeA) at (-4,0) {\begin{minipage}{2.5cm}
            \centering
            \footnotesize Type-\A{A} clique\\
            \scriptsize pruning:~\eqref{eq:prune_rule_naive}
        \end{minipage}};
        
        \node[draw,thick,fill=gray!15!white] (nodetypeB) at (4,0) {\begin{minipage}{2.5cm}
            \centering
            \footnotesize Type-\B{B} clique\\
            \scriptsize pruning:~\eqref{eq:prune_rule_improved}
        \end{minipage}};
        
        \node[draw,thick,fill=gray!15!white] (nodetypeAB) at (0,-2) {\begin{minipage}{2.5cm}
            \centering
            \footnotesize Type-\A{A}\B{B} clique\\
            \scriptsize pruning:~\eqref{eq:prune_rule_improved}
        \end{minipage}};

        \draw[thick, ->] (nodetypeA.west) to[out=180,in=90,looseness=4] node[above=-3,rotate=30] {\scriptsize add type-\A{A} vertex} (nodetypeA.north);
        \draw[thick, ->] (nodetypeB.east) to[out=0,in=90,looseness=4] node[above=-3,rotate=-30] {\scriptsize add type-\B{B} vertex} (nodetypeB.north);
        \draw[thick] (nodetypeA.east) to[out=0,in=90] ($(nodetypeAB.north)+(0,1)$);
        \draw[thick] (nodetypeB.west) to[out=180,in=90] node[pos=0.22, rotate=-4, above=-3] {\scriptsize$n_+^{\overline{\gram}}=0$} ($(nodetypeAB.north)+(0,1)$);
        \draw[thick, ->] ($(nodetypeAB.north)+(0,1)$) to (nodetypeAB.north);
        \draw[thick] (nodetypeAB.east) to[out=0,in=60] ($(nodetypeAB.east)+(1,-1)$);
        \draw[thick] (nodetypeB.west) to[out=180,in=60] ($(nodetypeAB.east)+(1,-1)$);
        \draw[thick, ->] ($(nodetypeAB.east)+(1,-1)$) to[out=240,in=-90] (nodetypeAB.south);
        \node (note) at (6,-3) {\scriptsize\begin{minipage}{2.5cm}
            \centering
            join via only \A{A}--\B{B}\\ non-trivial edges
        \end{minipage}};
        \draw[line width=3, white] (note) to[out=150,in=-30] ($(nodetypeAB.north)+(0.1,1)$);
        \draw[->, gray] (note) to[out=150,in=-30] ($(nodetypeAB.north)+(0.1,1)$);
        \draw[->, gray] (note) to[out=180,in=-30] ($(nodetypeAB.east)+(1.1,-1)$);
    \end{tikzpicture}
    \caption{Process for generating irreducible cliques of different types. Type-\A{A} and type-\B{B} cliques are generated independently according to the branch-and-prune process sketched in figure~\ref{fig:branch_and_prune}. Type-\A{A}\B{B} cliques are formed by combining type-\A{A} or type-\A{A}\B{B} cliques with type-\B{B} cliques which have $n_+^{\overline{\gram}}=0$ by introducing only \A{A}--\B{B} non-trivial edges, according to the procedure sketched in~\eqref{eq:branch_AB}.}
    \label{fig:flow-chart}
\end{figure}
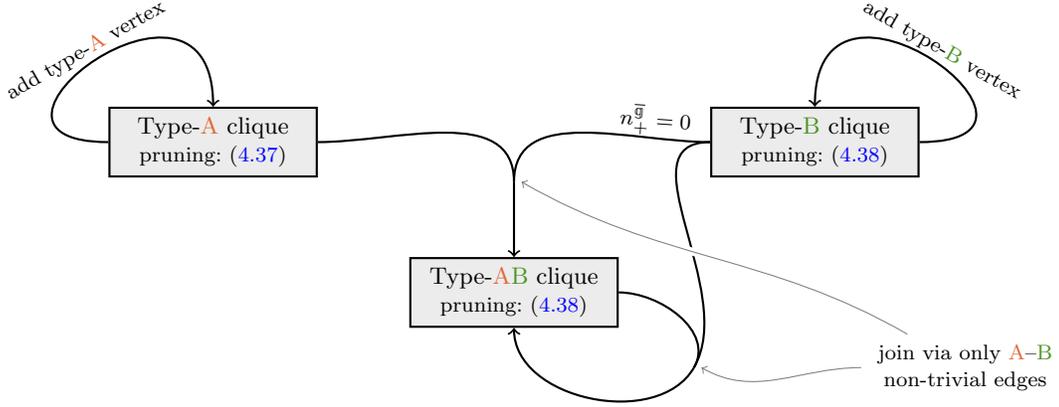

\medskip

During the pruning step, $k$-cliques generated during the branching step which are not admissible should be immediately pruned. The more delicate condition to determine is when to prune cliques which have $\Delta+28n_-^\gram > 273 - T_\text{min}$. There are cases where $\Delta+28n_-^\gram$ is very large but can be brought into accordance with the upper bound of~\eqref{eq:delta_28n_double_sided_bound} by the inclusion of additional vertices and non-trivial edges. As a small example,
\begin{equation}
    \begin{tikzpicture}[baseline={([yshift=-0.5ex]current bounding box.center)}]
        \begin{scope}[xshift=-4cm]
            \node[circle, inner sep=1.5, fill=typeA, draw=black, thick] (a) at (-1,0) {};
            \node[circle, inner sep=1.5, fill=typeA, draw=black, thick] (b) at (1,0) {};
    
            \draw[very thick] (a) to node[above=-3] {\tiny $(\rep{21},\rep{7})$} (b);
    
            \node[below=2] at (a) {\tiny\begin{minipage}{2cm}
                \centering
                $\SU(7)$\\
                $23\times\rep{7}+9\times\rep{21}$
            \end{minipage}};
            \node[below=2] at (b) {\tiny\begin{minipage}{2cm}
                \centering
                $G_2$\\
                $43\times\rep{7}$
            \end{minipage}};
    
            \node[below] at (0,-1) {$\Delta+28n_-^\gram=470\,,\;\; T_\text{min}=3$};
            \node[below] at (0,-1.5) {$\Gram = \begin{psmallmatrix}
                9-T & 9 & 13\\
                9 & 7 & 10\\
                13 & 10 & 11
            \end{psmallmatrix}$};
        \end{scope}

        \draw[->,thick] (-1,0) to node[above] {\footnotesize branch} (1,0);

        \begin{scope}[xshift=4cm]
            \node[circle, inner sep=1.5, fill=typeA, draw=black, thick] (a) at (-1,0) {};
            \node[circle, inner sep=1.5, fill=typeA, draw=black, thick] (b) at (1,0) {};
            \node[circle, inner sep=1.5, fill=typeB, draw=black, thick] (c) at (0,2*0.866) {};
    
            \draw[very thick] (a) to node[above=-3] {\tiny $(\rep{21},\rep{7})$} (b);
            \draw[very thick] (a) to node[above=-3,rotate=60] {\tiny $(\rep{7},\rep{22})$} (c);
            \draw[very thick] (b) to node[above=-3,rotate=-60] {\tiny $(\rep{22},\rep{7})$} (c);
    
            \node[below=2] at (a) {\tiny\begin{minipage}{2cm}
                \centering
                $\SU(7)$\\
                $23\times\rep{7}+9\times\rep{21}$
            \end{minipage}};
            \node[below=2] at (b) {\tiny\begin{minipage}{2cm}
                \centering
                $G_2$\\
                $43\times\rep{7}$
            \end{minipage}};
            \node[above] at (c) {\tiny\begin{minipage}{2cm}
                \centering
                $\SU(22)$\\
                $14\times\rep{22}+\rep{253}$
            \end{minipage}};
    
            \node[below] at (0,-1) {$\Delta+28n_-^\gram=268\,,\;\; T_\text{min}=3$};
            \node[below] at (0,-1.5) {$\Gram = \begin{psmallmatrix}
                9-T & 9 & 13 & -1\\
                9 & 7 & 10 & 1\\
                13 & 10 & 11 & 2\\
                -1 & 1 & 2 & -1
            \end{psmallmatrix}$};
            
        \end{scope}
    \end{tikzpicture}
\end{equation}
so any pruning rule should \emph{keep} the $2$-clique on the left. As an aside, the above cliques serve both as a good example with bi-charged hypermultiplets which are not a bi-fundamentals and as an example where $T_\text{min}$ does not increase when a clique grows.

\medskip

Consider a irreducible clique $\C_0^\irr$ for which $(\Delta+28n_-^\gram)(\C_0^\irr) > 273 - T_\text{min}(\C_0^\irr)$. Any irreducible clique $\widetilde{\C}^\irr$ which contains $\C_0^\irr$ can be viewed as having been formed by joining $\C_0^\irr$ to $m$ irreducible cliques via a set $\E_\text{new}$ of new non-trivial edges between active vertices. For example, schematically,
\begin{equation}
\label{eq:prune_decompose_1}
    \begin{tikzpicture}[baseline={([yshift=-0.5ex]current bounding box.center)}]
        \node at (-2.2,0.6) {$\widetilde{\C}^\irr \;=$};
    
        \node[circle, inner sep=1.5, fill=typeB, draw=black, thick] (a0) at (-0.9,0) {};
        \node[circle, inner sep=1.5, fill=typeB, draw=black, thick] (b0) at (0,0) {};
        \node[circle, inner sep=1.5, fill=typeB, draw=black, thick] (c0) at (0.5,0.8) {};
        \node[circle, inner sep=1.5, fill=typeB, draw=black, thick] (d0) at (1.3,1.3) {};
        \node[circle, inner sep=1.5, fill=typeB, draw=black, thick] (e0) at (0,-0.9) {};
        \draw[very thick] (a0) to (b0) to (c0) to (d0);
        \draw[very thick] (e0) to (b0);
        \draw (b0) circle (0.15);
        \draw (d0) circle (0.15);
        \draw (e0) circle (0.15);
        \node at (-0.2,0.6) {$\C_0^\irr$};

        \node[circle, inner sep=1.5, fill=typeB, draw=black, thick] (a1) at (2.9,1.2) {};
        \node[circle, inner sep=1.5, fill=typeB, draw=black, thick] (b1) at (3.8,1) {};
        \node[circle, inner sep=1.5, fill=typeB, draw=black, thick] (c1) at (4.7,1) {};
        \node[circle, inner sep=1.5, fill=typeB, draw=black, thick] (d1) at (5.5,0.8) {};
        \draw[very thick] (a1) to (b1) to (c1) to (d1);
        \draw (a1) circle (0.15);
        \draw (d1) circle (0.15);
        \node at (4.3,1.3) {$\C_1^\irr$};

        \node[circle, inner sep=1.5, fill=typeB, draw=black, thick] (a2) at (2.4,-0.3) {};
        \node[circle, inner sep=1.5, fill=typeB, draw=black, thick] (b2) at (2.8,-0.9) {};
        \node[circle, inner sep=1.5, fill=typeB, draw=black, thick] (c2) at (1.9,-1) {};
        \node[circle, inner sep=1.5, fill=typeB, draw=black, thick] (d2) at (3.6,-1.1) {};
        \draw[very thick] (a2) to (b2) to (c2);
        \draw[very thick] (b2) to (d2);
        \draw (a2) circle (0.15);
        \draw (c2) circle (0.15);
        \node at (3.3,-0.65) {$\C_2^\irr$};

        \draw[dashed, thick] (a1) to (d0) to (a2);
        \draw[dashed, thick] (e0) to (c2);

        \node[gray] (ee) at (1.2,-0.05) {$\E_\text{new}$};
        \draw[gray,->] (ee.south) to[out=-90,in=85] (1,-0.9);
        \draw[gray,->] (1.5,0.1) to[out=30,in=210] (1.9,0.4);
        \draw[gray,->] (ee.north) to[out=60,in=-100] (2.1,1.2);
    \end{tikzpicture}
\end{equation}
where active vertices are circled and we have suppressed trivial edges for clarity. (There are no non-trivial edges between $\C_1^\irr$ and $\C_2^\irr$ because otherwise we would just call the larger, combined irreducible clique $\C_1^\irr$.) We clearly have
\begin{equation}
    \Delta(\widetilde{\C}^\irr) = \sum_{\alpha=0}^m\Delta(\C_\alpha^\irr) - \sum_{\e\in\E_\text{new}}|\delta H(\e)| \,.
\end{equation}
Although $\gram(\widetilde{\C}^\irr)$ does not split into blocks associated to the $\C_\alpha^\irr$ because of the non-trivial edges connecting them, we can nevertheless bound
\begin{equation}
\label{eq:n-_inequality}
    n_-^\gram(\widetilde{\C}^\irr) \geq n_-^{\overline{\gram}}(\C_0^\irr) + \sum_{\alpha=1}^m n_-^\gram(\C_\alpha^\irr) \,.
\end{equation}
Like before, because the active vertices of $\C_0^\irr$ are a vertex cover of $\E_\text{new}$, when the corresponding rows and columns of $\gram$ are deleted to give $\overline{\gram}$ then $\gram(\widetilde{\C}^\irr)$ \emph{does} have a block structure, hence the above bound. If we bound the decrease in hypermultiplets from the new non-trivial edges by a sum over hypermultiplets of active vertices which are still available to merge, then we find
\begin{align}
\label{eq:prune_lower_bound}
    (\Delta+28n_-^\gram)(\widetilde{\C}^\irr) &\geq (\Delta+28n_-^{\overline{\gram}})(\C_0^\irr) + \sum_{\alpha=1}^m (\Delta+28n_-^\gram)(\C_\alpha^\irr) - \sum_{\e\in\E_\text{new}}|\delta H(\e)| \notag\\
    &\geq (\Delta+28n_-^{\overline{\gram}})(\C_0^\irr) - \lambda\sum_{\substack{\v\in\C_0^\irr\\ \text{active}}}\sum_Rn_R^\avbl(\v) H_R\\
    &\quad + \sum_{\alpha=1}^m \bigg[(\Delta+28n_-^\gram)(\C_\alpha^\irr) - (1-\lambda)\sum_{\substack{\v\in\C_\alpha^\irr\\ \text{active}}}\sum_Rn_R^\avbl(\v) H_R\bigg] \notag
\end{align}
for any $\lambda\in[0,1]$ (the factors of $\lambda$ and $1-\lambda$ are to avoid double-counting). By picking $\lambda=1$ and using $(\Delta+28n_-^\gram)(\C_\alpha^\irr)\geq 0$ (see \eqref{eq:delta_28n_double_sided_bound}) and $T_\text{min}(\widetilde{\C}^\irr)\geq T_\text{min}(\C_0^\irr)$ (recall the discussion below equation~\eqref{eq:Tmin_defn}) we get the following pruning condition:
\begin{equation}
\label{eq:prune_rule_naive}
    (\Delta+28n_-^{\overline{\gram}} + T_\text{min})(\C_0^\irr) - \sum_{\substack{\v\in\C_0^\irr\\ \text{active}}}\sum_Rn_R^\avbl(\v) H_R > 273 \;\;\implies\;\; \text{prune }\C_0^\irr \,.
\end{equation}
Namely, when the pruning condition above is satisfied, there are no irreducible cliques $\widetilde{\C}^\irr\supset\C_0^\irr$ compatible with \eqref{eq:delta_28n_double_sided_bound}. This na\"ive pruning rule suffices when building type-\A{A} cliques since their numbers drop quickly with increasing $k$ regardless. However, more generally this condition is fairly weak since the discarded terms, $(\Delta+28n_-^\gram)(\C_\alpha^\irr)$, are generally well above zero when $\C_\alpha^\irr$ has active vertices. With hind-sight, the term in brackets in the last line of~\eqref{eq:prune_lower_bound} is non-negative for all but a handful of admissible, irreducible cliques for $\lambda$ as low as $\frac{1}{2}$. This means that we can improve~\eqref{eq:prune_rule_naive} to
\begin{equation}
\label{eq:prune_rule_improved}
    (\Delta+28n_-^{\overline{\gram}} + T_\text{min})(\C_0^\irr) - \frac{1}{2}\sum_{\substack{\v\in\C_0^\irr\\ \text{active}}}\sum_Rn_R^\avbl(\v) H_R > 273\;``{+ \epsilon}" \;\;\implies\;\; \text{prune }\C_0^\irr \,.
\end{equation}
The ``$+\epsilon$'' term captures the corrections needed to address the small number of exceptions where irreducible cliques $\C_\alpha^\irr$ have
\begin{equation}
    (\Delta+28n_-^\gram)(\C_\alpha^\irr) - \frac{1}{2}\sum_{\substack{\v\in\C_\alpha^\irr\\ \text{active}}}\sum_Rn_R^\avbl(\v) H_R < 0 \,,
\end{equation}
and is meant to suggest that such terms are often absent altogether. In appendix~\ref{app:pruning} we discuss these exceptions further and give a precise expression for ``$+\epsilon$'' in \eqref{eq:+epsilon}.

\medskip

This branch-and-prune algorithm lends itself very nicely to parallelization as long as when a $k$-clique is produced in more than one thread the redundant copies are removed. For example, type-\A{A}\B{B} cliques can be constructed starting from each type-\A{A} clique completely independently since cliques with different type-\A{A} sub-cliques clearly never coincide. Also, although determining if two cliques are the same is an instance of the classic, and famously difficult, graph isomorphism problem, in practice checking for equivalence is not computationally demanding because the vertices and edges are labelled and $k$ never gets too large.

\section{Results}
\label{sec:results}

In this section we present results for the two different classifications mentioned in the introduction.  First, we focus on $T=0$ where the extremely constraining condition $n_-^\Gram=0$ allows us to enumerate \emph{all} anomaly-free theories with few caveats. We then turn to a more general classification where we allow for any $T$. In doing so we omit the groups $A_3\sim\SU(4)$ and $C_2\sim\Sp(2)$ which have four-dimensional irreps and contribute an unwieldy number of vertices to $\G$. Here we will discuss some global aspects of the two ensembles of theories and provide some hand-picked examples. The accompanying data sets are available at~\cite{Loges:2023gh2}.

\subsection{\texorpdfstring{$T=0$}{T=0}}
\label{sec:results_T=0}

\begin{table}[t]
    \centering
    \begin{tabular}{c|ccccccc}
        $\#(\text{simple gauge factors})$ & 1 & 2 & 3 & 4 & 5 & 6 & $7+$\\
        \midrule
        $\#(\text{anomaly-free theories})$ & 783 & 6130 & 8644 & 4004 & 279 & 7 & 0
    \end{tabular}
    \caption{Total number of anomaly-free theories for $T=0$ for each $k=\#(\text{simple gauge factors})$.}
    \label{tab:anom_free_T=0}
\end{table}

With no tensor multiplets the anomaly-cancellation conditions are especially strong and we are able to exhaustively construct \emph{all} anomaly-free theories, subject only to (i) $\SU(2)$ and $\SU(3)$ gauge factors are absent and (ii) hypermultiplets are charged under at most two simple factors. Since $T=0$ forces the number of negative eigenvalues $n_-^\Gram$ to vanish, only vertices with $9(b_i\cdot b_i)=(b_0\cdot b_i)^2>0$ need be kept in $\G$. In particular, there are only type-\A{A} vertices and all edges must be non-trivial. It turns out that there are no $\C_{<0,\alpha}^\irr$ irreducible cliques with $T_\text{min}=0$, as we will see.

To begin, we generate all solutions to the $B$-constraint for $A_{3\to24}$, $B_{3\to24}$, $C_{2\to24}$, $D_{4\to24}$ $E_{6,7,8}$, $F_4$ and $G_2$ out to $\Delta_\text{max}(G_i)$ according to
\begin{equation}
    \begin{aligned}
        &\Delta_\text{max}(A_3) = 550 \,, \quad \Delta_\text{max}(A_4) = 600 \,, \quad \Delta_\text{max}(A_5) = 800 \,, \quad \Delta_\text{max}(C_2) = 450 \,,\\
        &\Delta_\text{max}(D_4) = 600 \,, \quad \Delta_\text{max}(E_6)=\Delta_\text{max}(E_7) = 4000 \,, \quad \Delta_\text{max}(E_8) = 10,\!000 \,,
    \end{aligned}
\end{equation}
and $\Delta_\text{max}(G_i) = 1000$ otherwise. After restricting to solutions for which $T_\text{min}=0$ and deleting degree-zero vertices with $\Delta_i>273$, we are left with the following numbers of vertices for each simple group:
\begin{center}
    \begin{tabular}{c|rrrrrrccccccc}
        $n$ & \multicolumn{1}{c}{2} & \multicolumn{1}{c}{3} & \multicolumn{1}{c}{4} & \multicolumn{1}{c}{5} & \multicolumn{1}{c}{6} & \multicolumn{1}{c}{7} & 8 & 9 & 10 & 11 & 12 & $13\to23$ & 24 \\
        \midrule
        $A_n$ &      & 3893 & 186 & 270 & 28 & 22 & 8 & 7 & 4 & 2 & 1 & 1 & 0 \\
        $B_n$ &      &   48 &  15 &   6 &  2 &  2 & 1 & 1 & 0 & 0 & 0 & 0 & 0 \\
        $C_n$ & 9397 &  146 &  34 &   6 &  4 &  1 & 1 & 1 & 1 & 1 & 1 & 0 & 0 \\
        $D_n$ &      &      &  58 &  12 &  5 &  3 & 2 & 1 & 0 & 0 & 0 & 0 & 0 \\
        $E_n$ &      &      &     &     &  6 &  4 & 0 \\
        $F_n$ &      &      &   7 \\
        $G_n$ &   84
    \end{tabular}
\end{center}
The number of vertices for $\SU(N)$ groups in the first row correlates with the number of ``blocks'' identified in~\cite{Kumar2011} (see their table~2). Notably, for $A_5\sim\SU(6)$ and $A_9\sim\SU(10)$ the number of vertices we find is higher and the discrepancy can be attributed to the presence of half-hypermultiplets for these groups which appear not to be accounted for in~\cite{Kumar2011}. Indeed, we would get an exact match with~\cite{Kumar2011} by removing vertices which feature half-hypermultiplets: for $A_9\sim\SU(10)$ there is exactly one vertex with a half-hypermultiplet,
\begin{equation}
    \{G_i=\SU(10) \,,\;\H_i = 27\times\rep{10} + \rep{55} + \tfrac{1}{2}\times\rep{252} \,;\; \Delta_i = 352 \,,\; b_i\cdot b_i=9 \,,\; b_0\cdot b_i=9 \}
\end{equation}
and for $A_5\sim\SU(6)$ exactly $123$ vertices contain an odd number of $\tfrac{1}{2}\rep{20}$.

\begin{figure}[t]
    \centering
    \includegraphics[width=\textwidth]{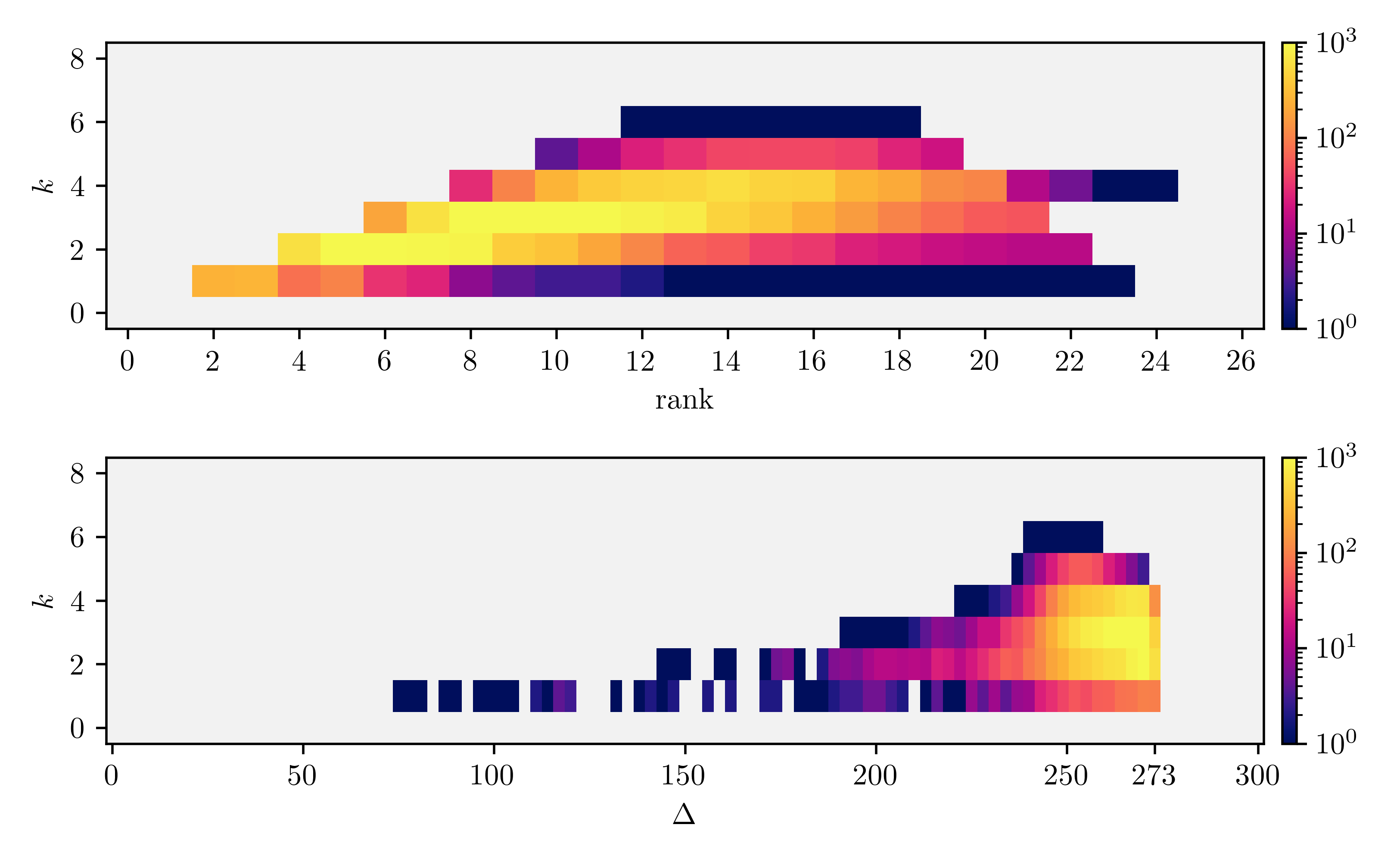}
    \caption{Distribution of gauge group rank (top) and $\Delta=H_\text{charged}-V$ (bottom) for anomaly-free cliques with $T=0$. The maximum rank is~$24$ and the largest anomaly-free cliques contain $k=6$ vertices. The lower bound $\mathrm{rank}\geq 2k$ is simply a consequence of not including rank-1 $\U(1)$ or $\SU(2)$ gauge groups.}
    \label{fig:anom_free_T=0_distributions}
\end{figure}

\medskip

Following the straightforward branch-and-prune algorithm sketched in figure~\ref{fig:branch_and_prune}, where type-\A{A} $k$-cliques are formed by adding vertices to $(k-1)$-cliques one-by-one, we find a total of $19,\!847$ anomaly-free theories. The number of anomaly-free $k$-cliques for each $k$ is given in table~\ref{tab:anom_free_T=0}, and distributions for both $\Delta$ and the gauge group rank as a function of $k$ are shown in figure~\ref{fig:anom_free_T=0_distributions}. Let us highlight a few notable features:
\begin{itemize}
    \item The values of $\Delta=\Delta+28n_-^\gram$ are always well above zero. The smallest value is $\Delta=77$, achieved by the simple clique with lone vertex $\{G_2,13\times\rep{7}\}$.

    \item The largest value of $k$ is $6$, achieved by having $r=0,1,\ldots,6$ copies of $\{\SU(4),\,20\times\rep{4}+3\times\rep{6}\}$ and $6-r$ copies of $\{\Sp(2),\,20\times\rep{4}+2\times\rep{5}\}$ and forming one bi-fundamental $(\rep{4},\rep{4})$ between each of the $15$ pairs of vertices. These have $b_i\cdot b_j=1$, $b_0\cdot b_i=3$ and $\Delta= 240+3r\leq 273$.
    
    \item The vertex with largest $\Delta_i$ which participates in an anomaly-free clique is
    \begin{equation}
        \v = \big\{ G_i = \SU(4) \,, \;\; \H_i = 66\times\rep{6} + \rep{35} \,; \;\; \Delta_i=416 \,,\;\; b_i\cdot b_i = 100 \,,\;\; b_0\cdot b_i = 30 \big\} \,,
    \end{equation}
    and appears exactly once:
    \begin{equation}
        \begin{aligned}
            G = \SU(4)\times\SU(11) \,, \quad \H = (\rep{6},\rep{66}) + (\rep{35},\rep{1}) \,, \quad \Delta = 273 \,, \quad \Gram = \begin{psmallmatrix}
                9 & 30 & 6\\
                30 & 100 & 20\\
                6 & 20 & 4
            \end{psmallmatrix} \,.
        \end{aligned}
    \end{equation}
    Figure~\ref{fig:T=0_A03_C02_vertices_used} shows the $A_3\sim\SU(4)$ and $C_2\sim\Sp(2)$ vertices which appear in anomaly-free cliques and serves as evidence that the bounds $\Delta_\text{max}(A_3)=550$ and $\Delta_\text{max}(C_2)=450$ are sufficiently large.

    \begin{figure}[t]
        \centering
        \includegraphics[width=\textwidth]{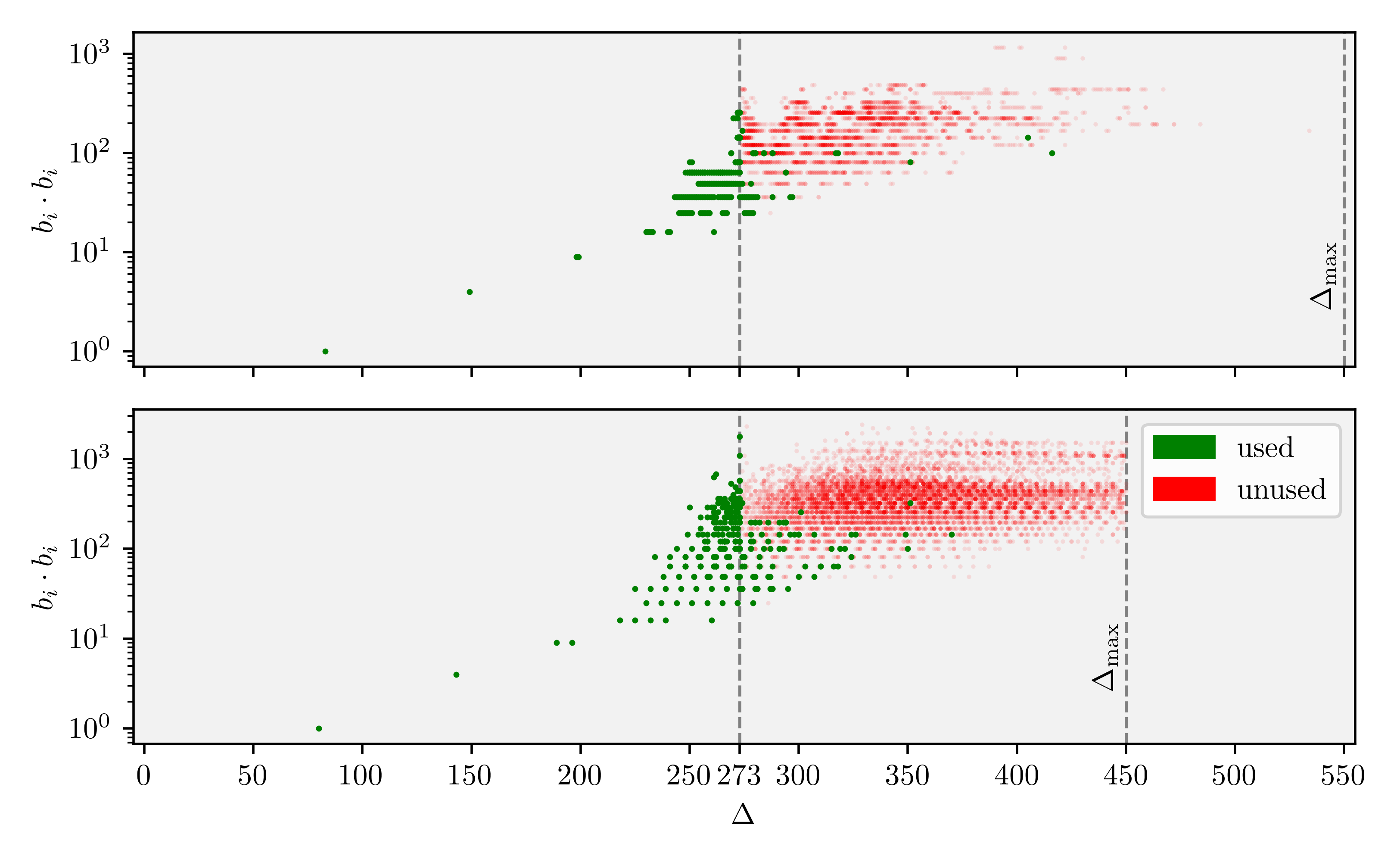}
        \caption{Vertices for $A_3\sim\SU(4)$ (top) and $C_2\sim\Sp(2)$ (bottom) which are admissible for $T=0$ and have either $\Delta\leq 273$ or degree $\geq 1$ in $\G$. Those vertices which appear in one or more anomaly-free clique are shown in green while those that go unused are shown in red. This is retroactive justification that $\Delta_\text{max}(A_3)=550$ and $\Delta_\text{max}(C_2)=450$ are sufficiently large.}
        \label{fig:T=0_A03_C02_vertices_used}
    \end{figure}
    
    \item For $k\neq4$, the rank is bounded as $\mathrm{rank}\leq 24-k$. All of the anomaly-free $4$-cliques with $\mathrm{rank}>20$ have one or more copies of the vertex $\{\SU(7),\,22\times\rep{7}+\rep{35}\}$. There is a single clique with a rank of $24$, consisting of four copies of this vertex with one bi-fundamental for each of the six pairs of vertices.
    
    \item The $\SU(N)$, $\SO(N)$ and $\Sp(N)$ gauge factors of largest rank which appear in anomaly-free cliques are $\SU(24)$, $\SO(18)$ and $\Sp(12)$, respectively, each appearing exactly once:
    \begin{equation}
        \begin{aligned}
            G &= \SU(24) \,, & \quad \H &= 3\times\rep{276} \,, & \quad \Delta &= 253 \,, & \quad \Gram &= \begin{psmallmatrix}
                9 & 3\\
                3 & 1
            \end{psmallmatrix} \,,\\
            G &= \SO(18) \,, & \H &= \rep{170} + \rep{256} \,, & \Delta &= 273 \,, & \Gram &= \begin{psmallmatrix}
                9 & 12\\
                12 & 16
            \end{psmallmatrix} \,,\\
            G &= \Sp(12) \,, & \H &= 2\times\rep{275} \,, & \Delta &= 250 \,, & \Gram &= \begin{psmallmatrix}
                9 & 3\\
                3 & 1
            \end{psmallmatrix} \,.
        \end{aligned}
    \end{equation}
    The above all have $k=1$: over anomaly-free cliques with $k\geq 2$ the largest-rank factors which appear are $\SU(20)$, $\SO(16)$ and $\Sp(10)$.

    \item The largest value of $b_i\cdot b_i$ for an anomaly-free clique occurs for a simple clique,
    \begin{equation}
        G = \Sp(2) \,, \quad \H = \rep{10} + \rep{35} + \rep{84} + \rep{154} \,, \quad \Delta=273 \,, \quad \Gram = \begin{psmallmatrix}
            9 & 126\\
            126 & 1764
        \end{psmallmatrix} \,,
    \end{equation}
    where $b_1\cdot b_1 = 42^2$. Amongst anomaly-free $(2+)$-cliques the largest value is $b_i\cdot b_i=18^2$ and occurs exactly twice, again for $\Sp(2)$:
    \begin{align}
        &\begin{aligned}
            G &= \SU(4)\times\Sp(2) \,, \quad \Delta = 273 \,, \quad \Gram = \begin{psmallmatrix}
                9 & 9 & 54\\
                9 & 9 & 54\\
                54 & 54 & 324
            \end{psmallmatrix} \,,\\
            \H &= (\rep{6},\rep{10}) + (\rep{4},\rep{35}) + (\rep{4},\rep{1}) + (\rep{10},\rep{1}) + (\rep{1},\rep{5}) + (\rep{1},\rep{14}) + (\rep{1},\rep{30}) + (\rep{1},\rep[\prime]{35}) \,,
        \end{aligned}\\[15pt]
        &\begin{aligned}
            G &= \Sp(2)\times E_6 \,, \quad \Delta = 273 \,, \quad \Gram = \begin{psmallmatrix}
                9 & 54 & 6\\
                54 & 324 & 36\\
                6 & 36 & 4
            \end{psmallmatrix} \,,\\
            \H &= (\rep{10},\rep{27}) + (\rep{10},\rep{1}) + (\rep{81},\rep{1}) \,.
        \end{aligned}
    \end{align}

    \item There are a few seeming coincidences in irrep dimensions which give anomaly-free theories with very few species of hypermultiplets, such as the following,
    \begin{equation}
    \label{eq:large_irrep_examples}
        \begin{aligned}
            G &= \SU(8) \,, & \quad \H &= \rep{336} \,, & \quad \Delta &= 273 \,, & \quad \Gram &= \begin{psmallmatrix*}
                9 & 24\\
                24 & 64
            \end{psmallmatrix*} \,,\\
            G &= F_4 \,, & \H &= \rep{52} + \rep{273} \,, & \Delta &= 273 \,, & \Gram &= \begin{psmallmatrix*}
                9 & 21\\
                21 & 49
            \end{psmallmatrix*} \,,\\
            G &= G_2 \,, & \H &= \rep{14} + \rep{273} \,, & \Delta &= 273 \,, & \Gram &= \begin{psmallmatrix*}
                9 & 117\\
                117 & 1521
            \end{psmallmatrix*} \,,\\
            G &= E_6 \,, & \H &= \rep{351} \,, & \Delta &= 273 \,, & \Gram &= \begin{psmallmatrix*}
                9 & 21\\
                21 & 49
            \end{psmallmatrix*} \,,\\
            G &= E_6 \,, & \H &= \rep[\prime]{351} \,, & \Delta &= 273 \,, & \Gram &= \begin{psmallmatrix*}
                9 & 24\\
                24 & 64
            \end{psmallmatrix*} \,,
        \end{aligned}
    \end{equation}
    all of which have no neutral hypers.

    \item Most representations which appear in anomaly-free cliques are low-dimensional, but occasional there can appear hypermultiplets in irreps which are anomalously large. For example, for $\SU(8)$ the low-dimensional irreps $\rep{8}$, $\rep{28}$, $\rep{36}$, $\rep{56}$, $\rep{63}$ and $\rep{70}$ all appear, then there is a gap before $\rep{336}=\ydiagram{3,1,1}$ appears exactly once in the first line of~\eqref{eq:large_irrep_examples}. For $\SU(7)$, again the low-dimensional irreps $\rep{7}$, $\rep{21}$, $\rep{28}$, $\rep{35}$ and $\rep{48}$ appear and then there is a gap until $\rep{210}=\ydiagram{2,1,1}$ appears exactly once, this time in a $2$-clique:
    \begin{equation}
        G = \SU(7)\times E_6 \,, \qquad \H = (\rep{210},\rep{1}) + (\rep{7},\rep{27}) \,, \qquad \Delta=273 \,.
    \end{equation}
    Similarly, for $\SU(6)$ we find that after a gap the irrep $\rep{189}=\ydiagram{2,2,1,1}$ appears exactly once:
    \begin{equation}
        G = \SU(6)\times\Sp(3) \,, \qquad \H = (\rep{189},\rep{1}) + (\rep{1},\rep[\prime]{14}) + (\rep{21},\rep{6}) \,, \qquad \Delta=273 \,.
    \end{equation}

    \item Amongst $(2+)$-cliques, the majority ($11,\!622$ out of $19,\!064$) have one or more bi-charged hypermultiplets which are \emph{not} bi-fundamentals. As might be expected, these non-bi-fundamentals typically involve the familiar two-index, adjoint, and spinor representations. Perhaps the most extreme representation is the $\rep{35}$ of $\Sp(2)$ which appears as part of a bi-charged hyper in exactly one anomaly-free clique:
    \begin{align}
        G &= \SU(4)\times\Sp(2) \,, \notag\\
        \H &= (\rep{6},\rep{10}) + (\rep{4},\rep{35}) + (\rep{4},\rep{1}) + (\rep{10},\rep{1}) + (\rep{1},\rep{5}) + (\rep{1},\rep{14}) + (\rep{1},\rep{30}) + (\rep{1},\rep[\prime]{35}) \,, \notag\\
        \Delta &= 273 \,, \notag\\
        \Gram &= \begin{psmallmatrix}
             9 &  9 &  54\\
             9 &  9 &  54\\
            54 & 54 & 324
        \end{psmallmatrix} \,.
    \end{align}
\end{itemize}

\subsection{Any \texorpdfstring{$T$}{T}}
\label{sec:results_general}

\begin{table}[t]
    \centering
    \begin{tabular}{llccc}
        \toprule
        $G_i$ & $\;\;\H_i$ & $\Delta_i$ & $b_i\cdot b_i$ & $b_0\cdot b_i$\\
        \midrule
        $\SU(2N)$ & $(2N+8)\times\rep{2N} + \rep{2N(2N+1)/2}$ & $N(2N+15)+1$ & $-1$ & $1$\\
        $\SU(2N)$ & $16\times\rep{2N} + 2\times\rep{2N(2N-1)/2}$ & $30N+1$ & $0$ & $2$ \\
        \bottomrule
    \end{tabular}
    \caption{Vertices which are removed from $\G$ by hand. These can always replace their corresponding $\Sp(N)$ type-\B{B} vertices to which they are effectively identical, only increasing $\Delta$ marginally.}
    \label{tab:SU(2N)_redundant}
\end{table}

We now turn to a general, $T$-agnostic classification with simple groups
\begin{equation}
\label{eq:simple_group_list}
    A_{4\to24} \,,\quad B_{3\to16} \,,\quad C_{3\to16} \,,\quad D_{4\to16} \,,\quad E_{6,7,8} \,,\quad F_4 \,,\quad G_2 \,.
\end{equation}
Solutions to the $B$-constraint for type-\A{A} vertices are generated out to $\Delta_\text{max}(G_i)$ according to
\begin{equation}
    \begin{aligned}
        &\Delta_\text{max}(A_4) = \Delta_\text{max}(B_3) = \Delta_\text{max}(C_3) = 375 \,, \\
        &\Delta_\text{max}(A_5) = \Delta_\text{max}(B_4) = \Delta_\text{max}(C_4) = \Delta_\text{max}(D_4) = 450 \,, \\
        &\Delta_\text{max}(D_5) = 700 \,, \quad \Delta_\text{max}(A_6) = 600 \,, \quad \Delta_\text{max}(A_7) = 800 \,,\\
        &\Delta_\text{max}(E_6) = \Delta_\text{max}(F_4) = 500 \,, \quad \Delta_\text{max}(G_2) = 350
    \end{aligned}
\end{equation}
and $\Delta_\text{max}(G_i)=1000$ otherwise. These choices for simple groups and $\Delta_\text{max}$ are sufficient to capture the bulk of all vertices relevant for anomaly-free theories, but we make no claim that this encompasses them all. In order to combat some of the combinatorics, the type-\B{B} vertices of table~\ref{tab:SU(2N)_redundant} are removed from $\G$ by hand. In an anomaly-free clique these $\SU(2N)$ vertices can always replace any or all of their $\Sp(N)$ counterparts to which they are effectively identical, having the same Gram matrix entries, hyper multiplicities and dimensions and only marginally larger $\Delta_i$. All statistics presented are with these vertices absent.

As discussed in section~\ref{sec:clique_anatomy}, the classification of anomaly-free cliques breaks into two steps. First, we discuss the enumeration of irreducible cliques which satisfy the bound $(\Delta+28n_-^\gram)(\C^\irr)\leq 273 - T_\text{min}(\C^\irr)$. In particular, we are able to retroactively justify the claim that there are no $\C_{<0}^\irr$ irreducible cliques if the eight vertices of~\eqref{eq:removed_vertices} are omitted. We are also able to identify irreducible cliques $\C_\infty^\irr$ which fall into infinite families parametrized only by their ranks and show that there are only finitely many such infinite families. Second, we discuss how anomaly-free cliques can be formed by taking disjoint unions of these irreducible cliques.

\subsubsection{Irreducible cliques}

\begin{figure}[t]
    \centering
    \includegraphics[width=\textwidth]{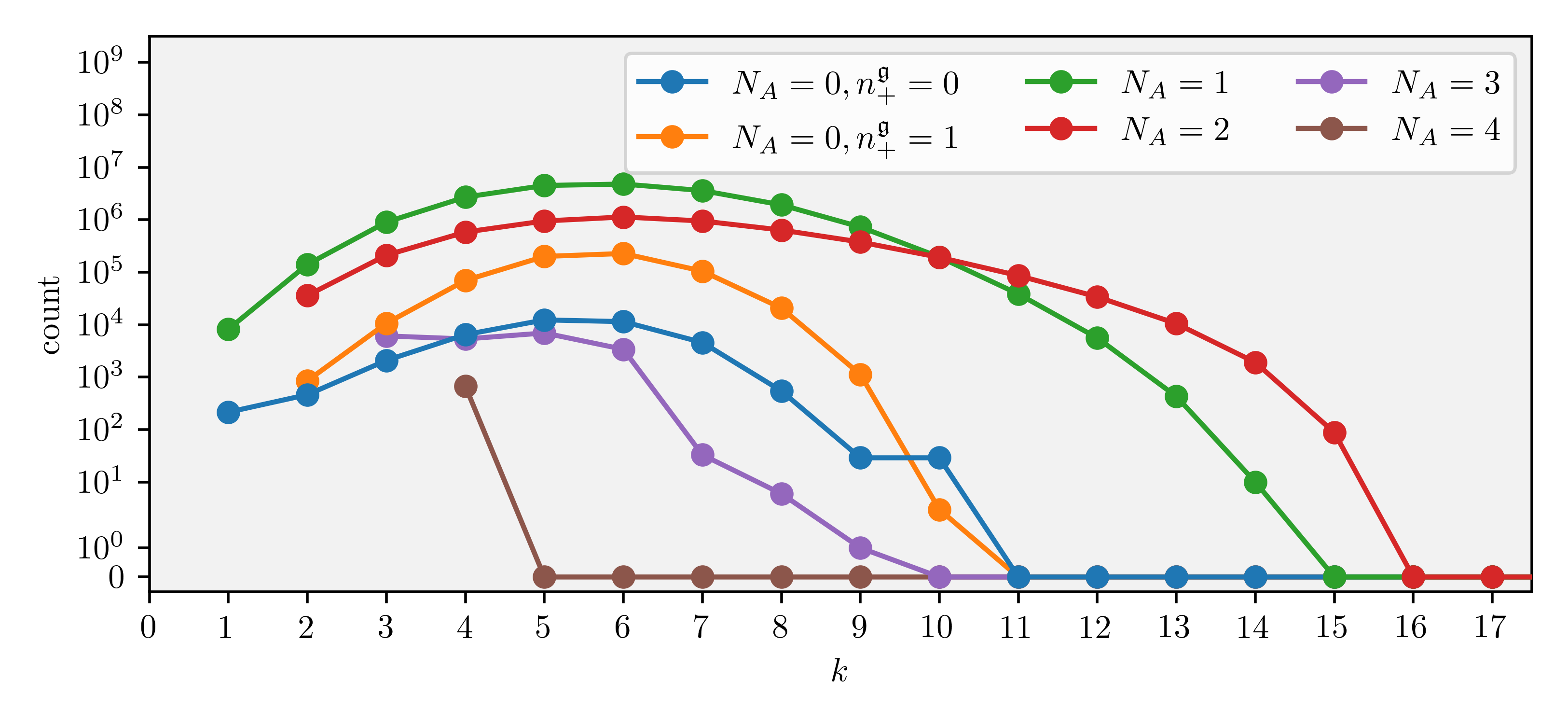}
    \caption{Number of irreducible cliques for the groups of~\eqref{eq:simple_group_list} which satisfy $(\Delta+28n_-^\gram)(\C^\irr)\leq 273 - T_\text{min}(\C^\irr)$ by number of type-\A{A} and total number of vertices. Type-\B{B} irreducible cliques are further subdivided by their value of $n_+^\gram$.}
    \label{fig:irreducible_clique_counts}
\end{figure}

Following the procedure outlined in section~\ref{sec:clique_construction}, an exhaustive list of irreducible cliques which satisfy $(\Delta+28n_-^\gram)(\C^\irr) \leq 273 - T_\text{min}(\C^\irr)$ is recursively generated. In figure~\ref{fig:irreducible_clique_counts} is shown the number of these cliques with different numbers of type-\A{A} and type-\B{B} vertices. Most bounded cliques have one or two type-\A{A} vertices and never more than four. The distributions peak around $k\approx 6$, indicating that restricting attention to gauge groups with few simple factors misses the majority of admissible hypermultiplet configurations.

Type-\B{B} vertices tend to participate in relatively few non-trivial edges in irreducible cliques. For example, the majority of type-\B{B} $5$-cliques satisfying the bound have a non-trivial edge structure which is linear:
\begin{equation*}
    \newcommand{\s}{0.5}
    \begin{aligned}
        \begin{tikzpicture}[baseline={([yshift=-0.5ex]current bounding box.center)}]
            \node[circle, inner sep=1.5, fill=typeB, draw=black, thick] (a) at (0,0) {};
            \node[circle, inner sep=1.5, fill=typeB, draw=black, thick] (b) at (\s,0) {};
            \node[circle, inner sep=1.5, fill=typeB, draw=black, thick] (c) at (2*\s,0) {};
            \node[circle, inner sep=1.5, fill=typeB, draw=black, thick] (d) at (3*\s,0) {};
            \node[circle, inner sep=1.5, fill=typeB, draw=black, thick] (e) at (4*\s,0) {};

            \draw[very thick] (a) to (b) to (c) to (d) to (e);
        \end{tikzpicture} \quad &:& \mathtt{50.499\%}& &\qquad
        \begin{tikzpicture}[baseline={([yshift=-0.5ex]current bounding box.center)}]
            \node[circle, inner sep=1.5, fill=typeB, draw=black, thick] (a) at (0,0) {};
            \node[circle, inner sep=1.5, fill=typeB, draw=black, thick] (b) at (-0.707*\s,-0.707*\s) {};
            \node[circle, inner sep=1.5, fill=typeB, draw=black, thick] (c) at (-0.707*\s,0.707*\s) {};
            \node[circle, inner sep=1.5, fill=typeB, draw=black, thick] (d) at (0.707*\s,-0.707*\s) {};
            \node[circle, inner sep=1.5, fill=typeB, draw=black, thick] (e) at (0.707*\s,0.707*\s) {};

            \draw[very thick] (b) to (a) to (c);
            \draw[very thick] (d) to (a) to (e);
        \end{tikzpicture} \quad &:& \mathtt{4.029\%}& &\qquad
        \begin{tikzpicture}[baseline={([yshift=-0.5ex]current bounding box.center)}]
            \node[circle, inner sep=1.5, fill=typeB, draw=black, thick] (a) at (0,0) {};
            \node[circle, inner sep=1.5, fill=typeB, draw=black, thick] (b) at (\s,0) {};
            \node[circle, inner sep=1.5, fill=typeB, draw=black, thick] (c) at (2*\s,0) {};
            \node[circle, inner sep=1.5, fill=typeB, draw=black, thick] (d) at (3*\s,0) {};
            \node[circle, inner sep=1.5, fill=typeB, draw=black, thick] (e) at (1.5*\s,0.866*\s) {};

            \draw[very thick] (a) to (b) to (c) to (d);
            \draw[very thick] (b) to (e) to (c);
        \end{tikzpicture} \quad &:& \mathtt{0.078\%}& \\
        \begin{tikzpicture}[baseline={([yshift=-0.5ex]current bounding box.center)}]
            \node[circle, inner sep=1.5, fill=typeB, draw=black, thick] (a) at (0,0) {};
            \node[circle, inner sep=1.5, fill=typeB, draw=black, thick] (b) at (\s,0) {};
            \node[circle, inner sep=1.5, fill=typeB, draw=black, thick] (c) at (2*\s,0) {};
            \node[circle, inner sep=1.5, fill=typeB, draw=black, thick] (d) at (2.5*\s,0.866*\s) {};
            \node[circle, inner sep=1.5, fill=typeB, draw=black, thick] (e) at (2.5*\s,-0.866*\s) {};

            \draw[very thick] (a) to (b) to (c) to (d);
            \draw[very thick] (c) to (e);
        \end{tikzpicture} \quad &:& \mathtt{29.237\%}& &\qquad
        \begin{tikzpicture}[baseline={([yshift=-0.5ex]current bounding box.center)}]
            \node[circle, inner sep=1.5, fill=typeB, draw=black, thick] (a) at (0,0) {};
            \node[circle, inner sep=1.5, fill=typeB, draw=black, thick] (b) at (\s,0) {};
            \node[circle, inner sep=1.5, fill=typeB, draw=black, thick] (c) at (2*\s,0) {};
            \node[circle, inner sep=1.5, fill=typeB, draw=black, thick] (d) at (2.866*\s,0.5*\s) {};
            \node[circle, inner sep=1.5, fill=typeB, draw=black, thick] (e) at (2.866*\s,-0.5*\s) {};

            \draw[very thick] (a) to (b) to (c) to (d) to (e) to (c);
        \end{tikzpicture} \quad &:& \mathtt{0.730\%}& &\qquad
        \begin{tikzpicture}[baseline={([yshift=-0.5ex]current bounding box.center)}]
            \node[circle, inner sep=1.5, fill=typeB, draw=black, thick] (a) at (0,0) {};
            \node[circle, inner sep=1.5, fill=typeB, draw=black, thick] (b) at (\s,0) {};
            \node[circle, inner sep=1.5, fill=typeB, draw=black, thick] (c) at (1.866*\s,0.5*\s) {};
            \node[circle, inner sep=1.5, fill=typeB, draw=black, thick] (d) at (1.866*\s,-0.5*\s) {};
            \node[circle, inner sep=1.5, fill=typeB, draw=black, thick] (e) at (2.732*\s,0) {};

            \draw[very thick] (a) to (b) to (c) to (d) to (e) to (c);
            \draw[very thick] (b) to (d);
        \end{tikzpicture} \quad &:& \mathtt{0.017\%}& \\
        \begin{tikzpicture}[baseline={([yshift=-0.5ex]current bounding box.center)}]
            \node[circle, inner sep=1.5, fill=typeB, draw=black, thick] (a) at (0,0.851*\s) {};
            \node[circle, inner sep=1.5, fill=typeB, draw=black, thick] (b) at (0.809*\s,0.263*\s) {};
            \node[circle, inner sep=1.5, fill=typeB, draw=black, thick] (c) at (0.5*\s,-0.688*\s) {};
            \node[circle, inner sep=1.5, fill=typeB, draw=black, thick] (d) at (-0.5*\s,-0.688*\s) {};
            \node[circle, inner sep=1.5, fill=typeB, draw=black, thick] (e) at (-0.809*\s,0.263*\s) {};

            \draw[very thick] (a) to (b) to (c) to (d) to (e) to (a);
        \end{tikzpicture} \quad &:& \mathtt{10.721\%}& &\qquad
        \begin{tikzpicture}[baseline={([yshift=-0.5ex]current bounding box.center)}]
            \node[circle, inner sep=1.5, fill=typeB, draw=black, thick] (a) at (-\s,0) {};
            \node[circle, inner sep=1.5, fill=typeB, draw=black, thick] (b) at (0,-0.707*\s) {};
            \node[circle, inner sep=1.5, fill=typeB, draw=black, thick] (c) at (\s,0) {};
            \node[circle, inner sep=1.5, fill=typeB, draw=black, thick] (d) at (0,0.707*\s) {};
            \node[circle, inner sep=1.5, fill=typeB, draw=black, thick] (e) at (0,0) {};

            \draw[very thick] (a) to (b) to (c) to (d) to (a) to (e) to (c);
        \end{tikzpicture} \quad &:& \mathtt{0.475\%}& &\qquad
        \begin{tikzpicture}[baseline={([yshift=-0.5ex]current bounding box.center)}]
            \node[circle, inner sep=1.5, fill=typeB, draw=black, thick] (a) at (-0.866*\s,0) {};
            \node[circle, inner sep=1.5, fill=typeB, draw=black, thick] (b) at (0,0.5*\s) {};
            \node[circle, inner sep=1.5, fill=typeB, draw=black, thick] (c) at (0,-0.5*\s) {};
            \node[circle, inner sep=1.5, fill=typeB, draw=black, thick] (d) at (0.866*\s,0) {};
            \node[circle, inner sep=1.5, fill=typeB, draw=black, thick] (e) at (0,-\s) {};

            \draw[very thick] (a) to (b) to (c) to (a) to[out=-70,in=160] (e) to[out=20,in=-110] (d) to (c);
            \draw[very thick] (d) to (b);
        \end{tikzpicture} \quad &:& \mathtt{0.015\%}& \\
        \begin{tikzpicture}[baseline={([yshift=-0.5ex]current bounding box.center)}]
            \node[circle, inner sep=1.5, fill=typeB, draw=black, thick] (a) at (-0.707*\s,0) {};
            \node[circle, inner sep=1.5, fill=typeB, draw=black, thick] (b) at (0,-0.707*\s) {};
            \node[circle, inner sep=1.5, fill=typeB, draw=black, thick] (c) at (0.707*\s,0) {};
            \node[circle, inner sep=1.5, fill=typeB, draw=black, thick] (d) at (0,0.707*\s) {};
            \node[circle, inner sep=1.5, fill=typeB, draw=black, thick] (e) at (-1.707*\s,0) {};

            \draw[very thick] (a) to (b) to (c) to (d) to (a) to (e);
        \end{tikzpicture} \quad &:& \mathtt{4.086\%}& &\qquad
        \begin{tikzpicture}[baseline={([yshift=-0.5ex]current bounding box.center)}]
            \node[circle, inner sep=1.5, fill=typeB, draw=black, thick] (a) at (0,-0.5*\s) {};
            \node[circle, inner sep=1.5, fill=typeB, draw=black, thick] (b) at (\s,-0.5*\s) {};
            \node[circle, inner sep=1.5, fill=typeB, draw=black, thick] (c) at (\s,0.5*\s) {};
            \node[circle, inner sep=1.5, fill=typeB, draw=black, thick] (d) at (0,0.5*\s) {};
            \node[circle, inner sep=1.5, fill=typeB, draw=black, thick] (e) at (-0.866*\s,0) {};

            \draw[very thick] (a) to (b) to (c) to (d) to (e) to (a);
            \draw[very thick] (a) to (d);
        \end{tikzpicture} \quad &:& \mathtt{0.098\%}& &\qquad
        \begin{tikzpicture}[baseline={([yshift=-0.5ex]current bounding box.center)}]
            \node[circle, inner sep=1.5, fill=typeB, draw=black, thick] (a) at (0,0) {};
            \node[circle, inner sep=1.5, fill=typeB, draw=black, thick] (b) at (-0.866*\s,-0.5*\s) {};
            \node[circle, inner sep=1.5, fill=typeB, draw=black, thick] (c) at (-0.866*\s,0.5*\s) {};
            \node[circle, inner sep=1.5, fill=typeB, draw=black, thick] (d) at (0.866*\s,-0.5*\s) {};
            \node[circle, inner sep=1.5, fill=typeB, draw=black, thick] (e) at (0.866*\s,0.5*\s) {};

            \draw[very thick] (b) to (a) to (d) to (e) to (a) to (c);
        \end{tikzpicture} \quad &:& \mathtt{0.014\%}&
    \end{aligned}
\end{equation*}
The following is an example of a type-\B{B} $5$-clique which has the rarest shape above, of which there are only $30$:
\begin{equation*}
    \newcommand{\s}{2}
    \begin{tikzpicture}[baseline={([yshift=-0.5ex]current bounding box.center)}]
        \node[circle, inner sep=1.5, fill=typeB, draw=black, thick] (a) at (0,0) {};
        \node[circle, inner sep=1.5, fill=typeB, draw=black, thick] (b) at (-0.866*\s,-0.5*\s) {};
        \node[circle, inner sep=1.5, fill=typeB, draw=black, thick] (c) at (-0.866*\s,0.5*\s) {};
        \node[circle, inner sep=1.5, fill=typeB, draw=black, thick] (d) at (0.866*\s,-0.5*\s) {};
        \node[circle, inner sep=1.5, fill=typeB, draw=black, thick] (e) at (0.866*\s,0.5*\s) {};

        \draw (c) to[out=0,in=120] (d) to[out=190,in=-10] (b) to (c) to[out=10,in=170] (e) to[out=180,in=60] (b);

        \draw[line width=3, white] (a) to (c);
        \draw[line width=3, white] (a) to (e);
        \draw[very thick] (b) to node[above=5, anchor=center, rotate=30] {\tiny$(\rep{6},\rep{11})$} (a);
        \draw[very thick] (c) to node[above=5, anchor=center, rotate=-30] {\tiny$(\rep{6},\rep{11})$} (a);
        \draw[very thick] (a) to node[above=5, anchor=center, rotate=-30] {\tiny$(\rep{11},\rep{5})$} (d);
        \draw[very thick] (a) to node[above=5, anchor=center, rotate=30] {\tiny$(\rep{11},\rep{5})$} (e);
        \draw[very thick] (e) to node[right=5, anchor=center, rotate=-90] {\tiny$(\rep{5},\rep{5})$} (d);

        \node[below=5] at (a) {\scriptsize\begin{minipage}{2cm}
            \centering
            $\SU(11)$\\
            $22\times\rep{11}$
        \end{minipage}};
        \node[below] at (b) {\scriptsize\begin{minipage}{2cm}
            \centering
            $\SU(6)$\\
            $12\times\rep{6}$
        \end{minipage}};
        \node[above] at (c) {\scriptsize\begin{minipage}{2cm}
            \centering
            $\SU(6)$\\
            $12\times\rep{6}$
        \end{minipage}};
        \node[below] at (d) {\scriptsize\begin{minipage}{2cm}
            \centering
            $\SU(5)$\\
            $16\times\rep{5}+2\times\rep{10}$
        \end{minipage}};
        \node[above] at (e) {\scriptsize\begin{minipage}{2cm}
            \centering
            $\SU(5)$\\
            $16\times\rep{5}+2\times\rep{10}$
        \end{minipage}};
    \end{tikzpicture}
    \;\Longleftrightarrow\;
    \begin{aligned}
        G &= \SU(11)\times\SU(6)\times\SU(6)\times\SU(5)\times\SU(5) \,,\\
        \H &= \;\scriptstyle (\rep{11},\rep{6},\rep{1},\rep{1},\rep{1}) + (\rep{11},\rep{1},\rep{6},\rep{1},\rep{1}) + (\rep{11},\rep{1},\rep{1},\rep{5},\rep{1}) + (\rep{11},\rep{1},\rep{1},\rep{1},\rep{5})\\
        &\qquad \scriptstyle + (\rep{1},\rep{1},\rep{1},\rep{5},\rep{5}) + (\rep{1},\rep{6},\rep{1},\rep{1},\rep{1}) + (\rep{1},\rep{1},\rep{6},\rep{1},\rep{1})\\
        &\qquad \scriptstyle + 2\times(\rep{1},\rep{1},\rep{1},\rep{10},\rep{1}) + 2\times(\rep{1},\rep{1},\rep{1},\rep{1},\rep{10}) \,,\\
        \Gram &= \begin{psmallmatrix}
            9-T &  0 &  0 &  0 & 2 & 2\\
              0 & -2 &  1 &  1 & 1 & 1\\
              0 &  1 & -2 &    &   &  \\
              0 &  1 &    & -2 &   &  \\
              2 &  1 &    &    & 0 & 1\\
              2 &  1 &    &    & 1 & 0
        \end{psmallmatrix} \,, \qquad n_\pm^\gram = (1,4) \,,\\
        \Delta &= 81 \,, \quad \Delta+28n_-^\gram = 193 \,, \quad T_\text{min} = 7 \,.
    \end{aligned}
\end{equation*}
Although the above example is not anomaly-free by itself, it can be made so by taking a disjoint union with two copies of the $1$-clique with vertex $\{G_2,\,\rep{7}\}$, for example.

Similarly, for type-\A{A}\B{B} cliques it is most common to have linear chains of type-\B{B} vertices connected to one of the type-\A{A} vertices. For example, the $21,\!837$ bounded irreducible cliques with exactly three type-\A{A} vertices can be organized by their non-trivial edge structure as follows:

\begin{equation*}
    \newcommand{\s}{0.5}
    \begin{aligned}
        \begin{tikzpicture}[baseline={([yshift=-0.5ex]current bounding box.center)}]
            \node[circle, inner sep=1.5, fill=typeA, draw=black, thick] (a) at (0,0) {};
            \node[circle, inner sep=1.5, fill=typeA, draw=black, thick] (b) at (-0.866*\s,0.5*\s) {};
            \node[circle, inner sep=1.5, fill=typeA, draw=black, thick] (c) at (-0.866*\s,-0.5*\s) {};

            \draw[very thick] (a) to (b) to (c) to (a);
        \end{tikzpicture} \quad &:& \mathtt{6104}& &\qquad
        \begin{tikzpicture}[baseline={([yshift=-0.5ex]current bounding box.center)}]
            \node[circle, inner sep=1.5, fill=typeA, draw=black, thick] (a) at (-0.866*\s,0) {};
            \node[circle, inner sep=1.5, fill=typeA, draw=black, thick] (b) at (0,0.5*\s) {};
            \node[circle, inner sep=1.5, fill=typeA, draw=black, thick] (c) at (0,-0.5*\s) {};
            \node[circle, inner sep=1.5, fill=typeB, draw=black, thick] (d) at (0.866*\s,0) {};

            \draw[very thick] (b) to (c) to (a) to (b) to (d) to (c);
        \end{tikzpicture} \quad &:& \mathtt{709}& &\qquad
        \begin{tikzpicture}[baseline={([yshift=-0.5ex]current bounding box.center)}]
            \node[circle, inner sep=1.5, fill=typeA, draw=black, thick] (a) at (0,0) {};
            \node[circle, inner sep=1.5, fill=typeA, draw=black, thick] (b) at (-0.866*\s,0.5*\s) {};
            \node[circle, inner sep=1.5, fill=typeA, draw=black, thick] (c) at (-0.866*\s,-0.5*\s) {};
            \node[circle, inner sep=1.5, fill=typeB, draw=black, thick] (d) at (\s,0) {};

            \draw[very thick] (a) to (b) to (c) to (a) to (d);
            \draw[very thick] (b) to[out=30,in=120] (d) to[out=-120,in=-30] (c);
        \end{tikzpicture} \quad &:& \mathtt{200}\\
        \begin{tikzpicture}[baseline={([yshift=-0.5ex]current bounding box.center)}]
            \node[circle, inner sep=1.5, fill=typeA, draw=black, thick] (a) at (0,0) {};
            \node[circle, inner sep=1.5, fill=typeA, draw=black, thick] (b) at (-0.866*\s,0.5*\s) {};
            \node[circle, inner sep=1.5, fill=typeA, draw=black, thick] (c) at (-0.866*\s,-0.5*\s) {};
            \node[circle, inner sep=1.5, fill=typeB, draw=black, thick] (d) at (\s,0) {};

            \draw[very thick] (a) to (b) to (c) to (a) to (d);
        \end{tikzpicture} \quad &:& \mathtt{4436}& &\qquad
        \begin{tikzpicture}[baseline={([yshift=-0.5ex]current bounding box.center)}]
            \node[circle, inner sep=1.5, fill=typeA, draw=black, thick] (a) at (0,0) {};
            \node[circle, inner sep=1.5, fill=typeA, draw=black, thick] (b) at (-0.866*\s,0.5*\s) {};
            \node[circle, inner sep=1.5, fill=typeA, draw=black, thick] (c) at (-0.866*\s,-0.5*\s) {};
            \node[circle, inner sep=1.5, fill=typeB, draw=black, thick] (d) at (0.866*\s,0.5*\s) {};
            \node[circle, inner sep=1.5, fill=typeB, draw=black, thick] (e) at (0.866*\s,-0.5*\s) {};

            \draw[very thick] (a) to (b) to (c) to (a) to (d);
            \draw[very thick] (a) to (e);
        \end{tikzpicture} \quad &:& \mathtt{28}& &\qquad
        \begin{tikzpicture}[baseline={([yshift=-0.5ex]current bounding box.center)}]
            \node[circle, inner sep=1.5, fill=typeA, draw=black, thick] (a) at (-0.866*\s,0) {};
            \node[circle, inner sep=1.5, fill=typeA, draw=black, thick] (b) at (0,0.5*\s) {};
            \node[circle, inner sep=1.5, fill=typeA, draw=black, thick] (c) at (0,-0.5*\s) {};
            \node[circle, inner sep=1.5, fill=typeB, draw=black, thick] (d) at (0.866*\s,0) {};
            \node[circle, inner sep=1.5, fill=typeB, draw=black, thick] (e) at (1.866*\s,0) {};

            \draw[very thick] (b) to (c) to (a) to (b) to (d) to (c);
            \draw[very thick] (d) to (e);
        \end{tikzpicture} \quad &:& \mathtt{3}\\
        \begin{tikzpicture}[baseline={([yshift=-0.5ex]current bounding box.center)}]
            \node[circle, inner sep=1.5, fill=typeA, draw=black, thick] (a) at (0,0) {};
            \node[circle, inner sep=1.5, fill=typeA, draw=black, thick] (b) at (-0.866*\s,0.5*\s) {};
            \node[circle, inner sep=1.5, fill=typeA, draw=black, thick] (c) at (-0.866*\s,-0.5*\s) {};
            \node[circle, inner sep=1.5, fill=typeB, draw=black, thick] (d) at (\s,0) {};
            \node[circle, inner sep=1.5, fill=typeB, draw=black, thick] (e) at (2*\s,0) {};

            \draw[very thick] (a) to (b) to (c) to (a) to (d) to (e);
        \end{tikzpicture} \quad &:& \mathtt{6920}& &\qquad
        \begin{tikzpicture}[baseline={([yshift=-0.5ex]current bounding box.center)}]
            \node[circle, inner sep=1.5, fill=typeA, draw=black, thick] (a) at (0,0) {};
            \node[circle, inner sep=1.5, fill=typeA, draw=black, thick] (b) at (-0.866*\s,0.5*\s) {};
            \node[circle, inner sep=1.5, fill=typeA, draw=black, thick] (c) at (-0.866*\s,-0.5*\s) {};
            \node[circle, inner sep=1.5, fill=typeB, draw=black, thick] (d) at (0.866*\s,0.5*\s) {};
            \node[circle, inner sep=1.5, fill=typeB, draw=black, thick] (e) at (0.866*\s,-0.5*\s) {};
            \node[circle, inner sep=1.5, fill=typeB, draw=black, thick] (f) at (1.866*\s,-0.5*\s) {};

            \draw[very thick] (a) to (b) to (c) to (a) to (d);
            \draw[very thick] (a) to (e) to (f);
        \end{tikzpicture} \quad &:& \mathtt{47}& &\qquad
        \begin{tikzpicture}[baseline={([yshift=-0.5ex]current bounding box.center)}]
            \node[circle, inner sep=1.5, fill=typeA, draw=black, thick] (a) at (0,0) {};
            \node[circle, inner sep=1.5, fill=typeA, draw=black, thick] (b) at (-0.866*\s,0.5*\s) {};
            \node[circle, inner sep=1.5, fill=typeA, draw=black, thick] (c) at (-0.866*\s,-0.5*\s) {};
            \node[circle, inner sep=1.5, fill=typeB, draw=black, thick] (d) at (0.866*\s,0.5*\s) {};
            \node[circle, inner sep=1.5, fill=typeB, draw=black, thick] (e) at (0.866*\s,-0.5*\s) {};
            \node[circle, inner sep=1.5, fill=typeB, draw=black, thick] (f) at (1.866*\s,-0.5*\s) {};
            \node[circle, inner sep=1.5, fill=typeB, draw=black, thick] (g) at (1.866*\s,0.5*\s) {};

            \draw[very thick] (a) to (b) to (c) to (a) to (d) to (g);
            \draw[very thick] (a) to (e) to (f);
        \end{tikzpicture} \quad &:& \mathtt{17}\\
        \begin{tikzpicture}[baseline={([yshift=-0.5ex]current bounding box.center)}]
            \node[circle, inner sep=1.5, fill=typeA, draw=black, thick] (a) at (0,0) {};
            \node[circle, inner sep=1.5, fill=typeA, draw=black, thick] (b) at (-0.866*\s,0.5*\s) {};
            \node[circle, inner sep=1.5, fill=typeA, draw=black, thick] (c) at (-0.866*\s,-0.5*\s) {};
            \node[circle, inner sep=1.5, fill=typeB, draw=black, thick] (d) at (\s,0) {};
            \node[circle, inner sep=1.5, fill=typeB, draw=black, thick] (e) at (2*\s,0) {};
            \node[circle, inner sep=1.5, fill=typeB, draw=black, thick] (f) at (3*\s,0) {};

            \draw[very thick] (a) to (b) to (c) to (a) to (d) to (e) to (f);
        \end{tikzpicture} \quad &:& \mathtt{3350}& &\qquad
        \begin{tikzpicture}[baseline={([yshift=-0.5ex]current bounding box.center)}]
            \node[circle, inner sep=1.5, fill=typeA, draw=black, thick] (a) at (0,0) {};
            \node[circle, inner sep=1.5, fill=typeA, draw=black, thick] (b) at (-0.866*\s,0.5*\s) {};
            \node[circle, inner sep=1.5, fill=typeA, draw=black, thick] (c) at (-0.866*\s,-0.5*\s) {};
            \node[circle, inner sep=1.5, fill=typeB, draw=black, thick] (d) at (0.866*\s,0.5*\s) {};
            \node[circle, inner sep=1.5, fill=typeB, draw=black, thick] (e) at (0.866*\s,-0.5*\s) {};
            \node[circle, inner sep=1.5, fill=typeB, draw=black, thick] (f) at (1.866*\s,-0.5*\s) {};
            \node[circle, inner sep=1.5, fill=typeB, draw=black, thick] (g) at (2.866*\s,-0.5*\s) {};

            \draw[very thick] (a) to (b) to (c) to (a) to (d);
            \draw[very thick] (a) to (e) to (f) to (g);
        \end{tikzpicture} \quad &:& \mathtt{16}& &\qquad
        \begin{tikzpicture}[baseline={([yshift=-0.5ex]current bounding box.center)}]
            \node[circle, inner sep=1.5, fill=typeA, draw=black, thick] (a) at (0,0) {};
            \node[circle, inner sep=1.5, fill=typeA, draw=black, thick] (b) at (-0.866*\s,0.5*\s) {};
            \node[circle, inner sep=1.5, fill=typeA, draw=black, thick] (c) at (-0.866*\s,-0.5*\s) {};
            \node[circle, inner sep=1.5, fill=typeB, draw=black, thick] (d) at (0.866*\s,0.5*\s) {};
            \node[circle, inner sep=1.5, fill=typeB, draw=black, thick] (e) at (0.866*\s,-0.5*\s) {};
            \node[circle, inner sep=1.5, fill=typeB, draw=black, thick] (f) at (1.866*\s,-0.5*\s) {};
            \node[circle, inner sep=1.5, fill=typeB, draw=black, thick] (g) at (2.866*\s,-0.5*\s) {};
            \node[circle, inner sep=1.5, fill=typeB, draw=black, thick] (h) at (1.866*\s,0.5*\s) {};

            \draw[very thick] (a) to (b) to (c) to (a) to (d) to (h);
            \draw[very thick] (a) to (e) to (f) to (g);
        \end{tikzpicture} \quad &:& \mathtt{6}\\
        & & & & & & & &\qquad
        \begin{tikzpicture}[baseline={([yshift=-0.5ex]current bounding box.center)}]
            \node[circle, inner sep=1.5, fill=typeA, draw=black, thick] (a) at (0,0) {};
            \node[circle, inner sep=1.5, fill=typeA, draw=black, thick] (b) at (-0.866*\s,0.5*\s) {};
            \node[circle, inner sep=1.5, fill=typeA, draw=black, thick] (c) at (-0.866*\s,-0.5*\s) {};
            \node[circle, inner sep=1.5, fill=typeB, draw=black, thick] (d) at (0.866*\s,0.5*\s) {};
            \node[circle, inner sep=1.5, fill=typeB, draw=black, thick] (e) at (0.866*\s,-0.5*\s) {};
            \node[circle, inner sep=1.5, fill=typeB, draw=black, thick] (f) at (1.866*\s,-0.5*\s) {};
            \node[circle, inner sep=1.5, fill=typeB, draw=black, thick] (g) at (2.866*\s,-0.5*\s) {};
            \node[circle, inner sep=1.5, fill=typeB, draw=black, thick] (h) at (1.866*\s,0.5*\s) {};
            \node[circle, inner sep=1.5, fill=typeB, draw=black, thick] (i) at (2.866*\s,0.5*\s) {};

            \draw[very thick] (a) to (b) to (c) to (a) to (d) to (h) to (i);
            \draw[very thick] (a) to (e) to (f) to (g);
        \end{tikzpicture} \quad &:& \mathtt{6}
    \end{aligned}
\end{equation*}
The structures with one or two type-\A{A} vertices are more varied.

\begin{figure}[p]
    \centering
    \includegraphics[width=\textwidth]{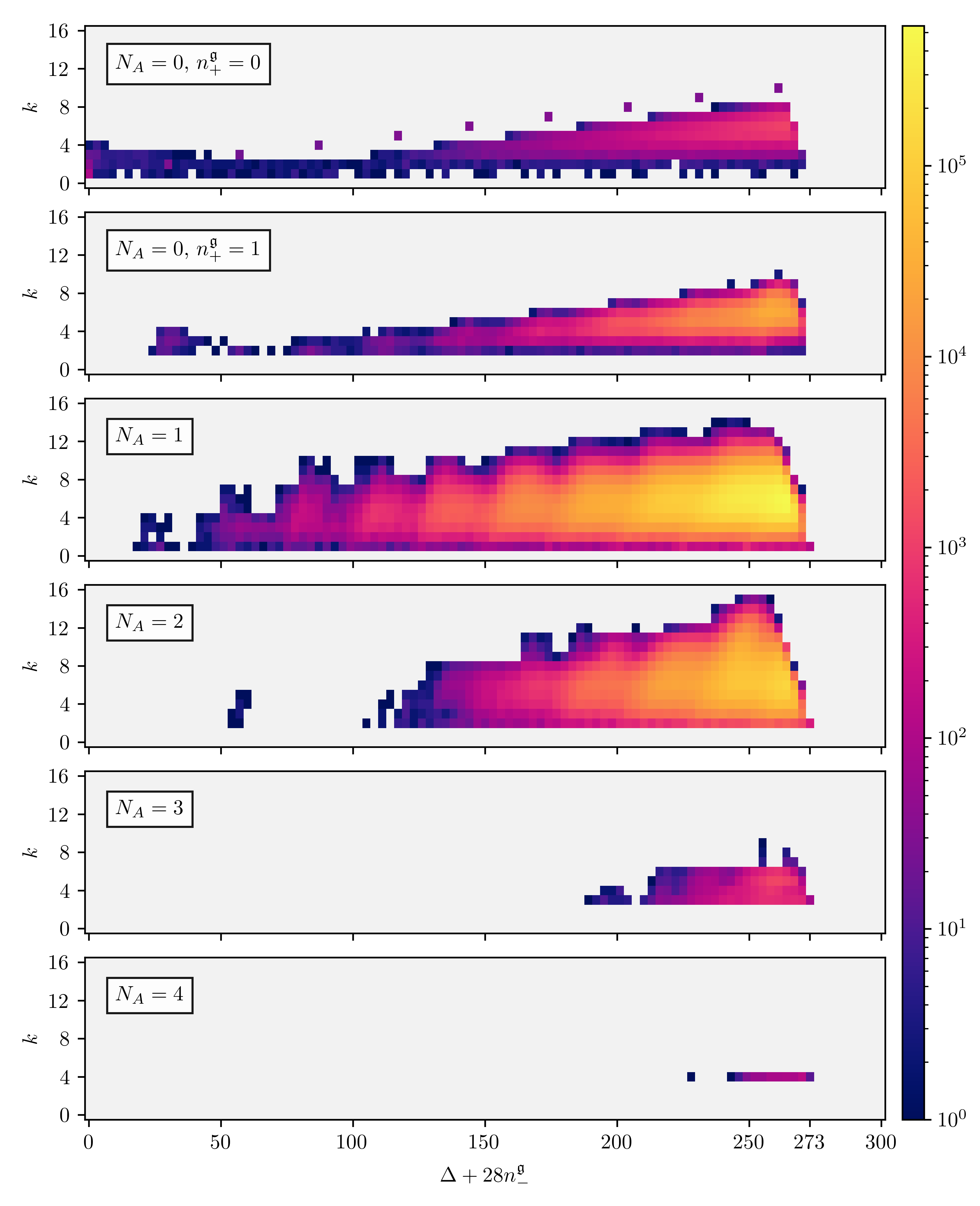}
    \caption{Distribution of $\Delta+28n_-^\gram$ vs.\ $k$ for irreducible cliques which satisfy the bound $\Delta+28n_-^\gram\leq 273 - T_\text{min}$ by number of type-\A{A} vertices. Type-\B{B} irreducible cliques are further subdivided by their value of $n_+^\gram$.}
    \label{fig:irreducible_cliques_dist}
\end{figure}

\medskip

The distributions of $\Delta+28n_-^\gram$ vs.\ $k$ are shown in figure~\ref{fig:irreducible_cliques_dist}, again separately for different numbers of type-\A{A} and type-\B{B} vertices. As expected, the number of irreducible cliques grows quickly with $\Delta+28n_-^\gram$ and most come close to saturating the upper bound. We emphasize that not all of these will be anomaly-free by themselves since it is not uncommon to have $T_\text{min}>n_-^\gram$ and $\Delta+28n_-^\gram+T_\text{min}\leq 273 < \Delta+29T_\text{min}$. As a general trend, $\Delta+28n_-^\gram$ inevitably grows with $k$ and only type-\B{B} cliques with $n_+^\gram=0$ even come close to saturating the lower bound $\Delta+28n_-^\gram\geq 0$. This retroactively justifies the first claim at the end of section~\ref{sec:clique_anatomy} and many of the steps in establishing pruning conditions in section~\ref{sec:clique_construction}.

In figure~\ref{fig:irreducible_cliques_dist} there are some clearly visible outliers. Amongst type-\B{B} cliques with $n_+^\gram=0$ there are isolated clusters which lie elevated above the rest; each of these corresponds to two families of irreducible cliques,
\begin{equation}
\label{eq:nneg=0_families}
    \newcommand{\s}{0.75}
    \newcommand{\w}{1}
    \begin{tikzpicture}[baseline={([yshift=-0.5ex]current bounding box.center)}]
        \node[inner sep=2, fill=typeB, draw=black, thick] (a) at (0,0) {};
        \node[inner sep=2, fill=typeB, draw=black, thick] (b) at (\s,0) {};
        \node[inner sep=2, fill=typeB, draw=black, thick] (c) at (2*\s,0) {};
        \node[circle, inner sep=1.5, fill=typeB, draw=black, thick] (d) at (-0.5*\s,0.866*\s) {};
        \node[circle, inner sep=1.5, fill=typeB, draw=black, thick] (e) at (-0.5*\s,-0.866*\s) {};
        \node[inner sep=2, fill=typeB, draw=black, thick] (f) at (4*\s,0) {};
        \node[inner sep=2, fill=typeB, draw=black, thick] (g) at (5*\s,0) {};
        \node[circle, inner sep=1.5, fill=typeB, draw=black, thick] (h) at (5.5*\s,0.866*\s) {};
        \node[circle, inner sep=1.5, fill=typeB, draw=black, thick] (i) at (5.5*\s,-0.866*\s) {};

        \draw[very thick] (d) to (a) to (e);
        \draw[very thick] (a) to (b) to (c) to node[fill=white] {$\cdots$} (f) to (g);
        \draw[very thick] (h) to (g) to (i);

        \begin{scope}[xshift=-3cm]
            \node[circle, inner sep=1.5, fill=typeB, draw=black, thick] (a2) at (0,-\w) {};
            \node[circle, inner sep=1.5, fill=typeB, draw=black, thick] (b2) at (0.707*\w,-0.707*\w) {};
            \node[circle, inner sep=1.5, fill=typeB, draw=black, thick] (c2) at (\w,0) {};
            \node[circle, inner sep=1.5, fill=typeB, draw=black, thick] (d2) at (0.707*\w,0.707*\w) {};
            \node[circle, inner sep=1.5, fill=typeB, draw=black, thick] (e2) at (0,\w) {};
            \node[circle, inner sep=1.5, fill=typeB, draw=black, thick] (f2) at (-0.707*\w,0.707*\w) {};
            \node[circle, inner sep=1.5, fill=typeB, draw=black, thick] (g2) at (-\w,0) {};
            \node (h2) at (-0.707*\w,-0.707*\w) {};
            \node (hh) at (-0.65*\w,-0.65*\w) {};
        \end{scope}

        \draw[very thick] (a2) to (b2) to (c2) to (d2) to (e2) to (f2) to (g2) to (h2) to (a2);
        \node[rotate=-45,fill=white] at (hh) {$\cdots$};

        \node[circle, inner sep=1.5, fill=typeB, draw=black, thick] (l1) at (5*\s+2,0.25) {};
        \node[inner sep=2, fill=typeB, draw=black, thick] (l2) at (5*\s+2,-0.25) {};
        \node[right] at (l1) {\scriptsize$=\{\SU(N)+2N\times\rep{N}\}$};
        \node[right] at (l2) {\scriptsize$=\{\SU(2N)+4N\times\rep{2N}\}$};
    \end{tikzpicture}
\end{equation}
where each edge represents one bi-fundamental and trivial edges have been suppressed. These have $\Delta=k$, $n_-^\gram=k-1$ and $T_\text{min}=\max\{9,k\}$, satisfying the bound for $k\leq10$. Although these all have a positive linear combination of $b_i$ which is null and orthogonal to all $b_I$ ($b_\text{null}=\sum_{i=1}^kb_i$ for the the cliques on the left and $b_\text{null}=b_1+b_2+b_{k-1}+b_k+2\sum_{i=3}^{k-2}b_i$ for those on the right), this is perfectly fine as long as $n_+^\Gram=0$ and $n_-^\Gram<T$. Only if one of the inequalities of equation~\eqref{eq:eigenvalue_bounds_G} is saturated must we have $b_\text{null}=0$ and thus violate $j\cdot b_i>0$. In particular, such irreducible components cannot appear together with any $\C_{\text{gen},\alpha}^\irr$ which has $n_+^\gram=1$.

There are also clear outliers for cliques with two, three and four type-\A{A} vertices which have $\Delta+28n_-^\gram$ well below the main distributions. The outlier with four type-\A{A} vertices is\footnote{Note that if one replaces one of the $(\rep{5},\repss{8}{s})$ hypers with $(\rep{5},\repss{8}{c})$ (recall conjugate representations are interchangeable), this example is triality-invariant.}
\begin{equation*}
    \newcommand{\s}{2.25}
    \begin{tikzpicture}[baseline={([yshift=-0.5ex]current bounding box.center)}]
        \node[circle, inner sep=1.5, fill=typeA, draw=black, thick] (a) at (0,0) {};
        \node[circle, inner sep=1.5, fill=typeA, draw=black, thick] (b) at (0,\s) {};
        \node[circle, inner sep=1.5, fill=typeA, draw=black, thick] (c) at (0.866*\s,-0.5*\s) {};
        \node[circle, inner sep=1.5, fill=typeA, draw=black, thick] (d) at (-0.866*\s,-0.5*\s) {};

        \draw[very thick] (a) to node[left=5, anchor=center, rotate=90] {\tiny$(\repss{8}{v},\rep{5})$} (b);
        \draw[very thick] (a) to node[above=5, anchor=center, rotate=-30] {\tiny$(\repss{8}{s},\rep{5})$} (c);
        \draw[very thick] (a) to node[above=5, anchor=center, rotate=30] {\tiny$(\rep{5},\repss{8}{s})$} (d);
        \draw[very thick] (b) to node[above right=4, anchor=center, rotate=-60] {\tiny$(\rep{5},\rep{5})$} (c);
        \draw[very thick] (d) to node[above left=4, anchor=center, rotate=60] {\tiny$(\rep{5},\rep{5})$} (b);
        \draw[very thick] (d) to node[below=-2] {\tiny$(\rep{5},\rep{5})$} (c);

        \node[below=4] at (a) {\scriptsize\begin{minipage}{2cm}
            \centering
            $\SO(8)$\\
            $5\times\repss{8}{v}+10\times\repss{8}{s}$
        \end{minipage}};
        \node[above] at (b) {\scriptsize\begin{minipage}{2cm}
            \centering
            $\SU(5)$\\
            $19\times\rep{5}+3\times\rep{10}$
        \end{minipage}};
        \node[below] at (c) {\scriptsize\begin{minipage}{2cm}
            \centering
            $\SU(5)$\\
            $19\times\rep{5}+3\times\rep{10}$
        \end{minipage}};
        \node[below] at (d) {\scriptsize\begin{minipage}{2cm}
            \centering
            $\SU(5)$\\
            $19\times\rep{5}+3\times\rep{10}$
        \end{minipage}};
    \end{tikzpicture}
    \hspace{-15pt}
    \;\Longleftrightarrow\;
    \begin{aligned}
        G &= \SU(5)\times\SU(5)\times\SU(5)\times\SO(8) \,,\\
        \H &= \;\scriptstyle (\rep{5},\rep{5},\rep{1},\rep{1}) + (\rep{5},\rep{1},\rep{5},\rep{1}) + (\rep{1},\rep{5},\rep{5},\rep{1})\\
        &\qquad \scriptstyle + (\rep{5},\rep{1},\rep{1},\repss{8}{v}) + (\rep{1},\rep{5},\rep{1},\repss{8}{s}) + (\rep{1},\rep{1},\rep{5},\repss{8}{s})\\
        &\qquad \scriptstyle + (\rep{5},\rep{1},\rep{1},\rep{1}) + (\rep{1},\rep{5},\rep{1},\rep{1}) + (\rep{1},\rep{1},\rep{5},\rep{1})\\
        &\qquad \scriptstyle + (\rep{10},\rep{1},\rep{1},\rep{1}) + (\rep{1},\rep{10},\rep{1},\rep{1}) + (\rep{1},\rep{1},\rep{10},\rep{1}) \,,\\
        \Gram &= \begin{psmallmatrix}
            9-T & 3 & 3 & 3 & 3\\
              3 & 1 & 1 & 1 & 2\\
              3 & 1 & 1 & 1 & 2\\
              3 & 1 & 1 & 1 & 2\\
              3 & 2 & 2 & 2 & 1
        \end{psmallmatrix} \,, \qquad n_\pm^\gram = (1,1) \,,\\
        \Delta &= 200 \,, \quad \Delta+28n_-^\gram = 228 \,, \quad T_\text{min} = 3 \,,
    \end{aligned}
\end{equation*}
and the small cluster of cliques with two type-\A{A} vertices near $\Delta+28n_-^\gram\approx 55$ includes
\begin{equation*}
    \newcommand{\s}{2.25}
    \begin{tikzpicture}[baseline={([yshift=-0.5ex]current bounding box.center)}]
        \node[circle, inner sep=1.5, fill=typeA, draw=black, thick] (a) at (0,0) {};
        \node[circle, inner sep=1.5, fill=typeA, draw=black, thick] (b) at (\s,0) {};

        \draw[very thick] (a) to node[above] {\tiny$(\rep{7},\rep{6})$} (b);

        \node[below] at (a) {\scriptsize\begin{minipage}{1.5cm}
            \centering
            $G_2$\\
            $6\times\rep{7} + \rep{14}$
        \end{minipage}};
        \node[below] at (b) {\scriptsize\begin{minipage}{2cm}
            \centering
            $\Sp(3)$\\
            $7\times\rep{6} + \rep[\prime]{14} + \rep{21}$
        \end{minipage}};
    \end{tikzpicture}
    \;\Longleftrightarrow\;\quad
    \begin{aligned}
        G &= G_2\times\Sp(3) \,,\\
        \H &= \;\scriptstyle (\rep{7},\rep{6}) + (\rep{1},\rep[\prime]{14}) + (\rep{1},\rep{21}) \,,\\
        \Gram &= \begin{psmallmatrix}
            9-T & 2 & 2\\
              2 & 2 & 2\\
              2 & 2 & 2
        \end{psmallmatrix} \,, \qquad n_\pm^\gram = (1,0) \,,\\
        \Delta &= 56 \,, \quad \Delta+28n_-^\gram = 56 \,, \quad T_\text{min} = 7 \,,
    \end{aligned}
\end{equation*}
and
\begin{equation*}
    \newcommand{\s}{2.25}
    \begin{tikzpicture}[baseline={([yshift=-0.5ex]current bounding box.center)}]
        \node[circle, inner sep=1.5, fill=typeA, draw=black, thick] (a) at (0,0) {};
        \node[circle, inner sep=1.5, fill=typeA, draw=black, thick] (b) at (\s,0) {};
        \node[circle, inner sep=1.5, fill=typeB, draw=black, thick] (c) at (2*\s,0) {};

        \draw[very thick] (a) to node[above=-3] {\tiny$(\rep{7},\rep{6})$} (b) to node[above=-3] {\tiny$(\rep{6},\rep{14})$} (c);
        \draw (a) to[out=30,in=150] (c);

        \node[below] at (a) {\scriptsize\begin{minipage}{1.5cm}
            \centering
            $G_2$\\
            $6\times\rep{7} + \rep{14}$
        \end{minipage}};
        \node[below] at (b) {\scriptsize\begin{minipage}{2cm}
            \centering
            $\Sp(3)$\\
            $21\times\rep{6} + \rep[\prime]{14}$
        \end{minipage}};
        \node[below] at (c) {\scriptsize\begin{minipage}{1.5cm}
            \centering
            $\SO(14)$\\
            $6\times\rep{14}$
        \end{minipage}};
    \end{tikzpicture}
    \;\Longleftrightarrow\;\quad
    \begin{aligned}
        G &= G_2\times\Sp(3)\times\SO(14) \,,\\
        \H &= \;\scriptstyle (\rep{7},\rep{6},\rep{1}) + (\rep{1},\rep{6},\rep{14}) + (\rep{14},\rep{1},\rep{1}) + (\rep{1},\rep[\prime]{14},\rep{1}) \,,\\
        \Gram &= \begin{psmallmatrix}
            9-T & 2 & 3 & -2\\
              2 & 2 & 2 &   \\
              3 & 2 & 1 &  2\\
             -2 &   & 2 & -4
        \end{psmallmatrix} \,, \qquad n_\pm^\gram = (1,1) \,,\\
        \Delta &= 28 \,, \quad \Delta+28n_-^\gram = 56 \,, \quad T_\text{min} = 8 \,,
    \end{aligned}
\end{equation*}
both of which are in fact anomaly-free.

There is also a clearly visible cluster of cliques for $n_+^\gram=0$ which have $\Delta+28n_-^\gram$ close to zero. These include the $\SO(2N+8)\times\Sp(N)$ cliques from table~\ref{tab:irreducible_cliques_infty} and
\begin{equation}
\label{eq:SUSO_chain}
    \newcommand{\s}{1}
    \begin{tikzpicture}[baseline={([yshift=-0.5ex]current bounding box.center)}]
        \node[inner sep=2, fill=typeB, draw=black, text=white, thick] (a) at (0,0) {\tiny8};
        \node[inner sep=2, fill=typeB, draw=black, text=white, thick] (b) at (\s,0) {\tiny16};
        \node[inner sep=2, fill=typeB, draw=black, text=white, thick] (c) at (2*\s,0) {\tiny24};
        \node[inner sep=2, fill=typeB, draw=black, text=white, thick] (d) at (4*\s,0) {\tiny$8(k-1)$};
        \node[circle, inner sep=2, fill=typeB, draw=black, text=white, thick] (e) at (5*\s,0) {\tiny$8k$};

        \draw[very thick] (a) to (b) to (c) to node[fill=white] {$\cdots$} (d) to (e);
        
        \node[inner sep=1.5, fill=typeB, draw=black, text=white, thick] (l1) at (5*\s+2,0.25) {\tiny$N$};
        \node[circle, inner sep=2, fill=typeB, draw=black, text=white, thick] (l2) at (5*\s+2,-0.25) {\tiny$N$};
        \node[right=7] at (l1) {\scriptsize$=\{\SU(N),\,2N\times\rep{N}\}$};
        \node[right=7] at (l2) {\scriptsize$=\{\SO(N),\,(N-8)\times\rep{N}\}$};
    \end{tikzpicture}
\end{equation}
where again each edge corresponds to one bi-fundamental and trivial edges have been suppressed. These have $\Delta=-27k-1$, $n_-^\gram=k$, $\Delta+28n_-^\gram=k-1\geq 0$ and $T_\text{min}=k$, satisfying the bound for $k\leq 137$ (given our choice of simple groups in~\eqref{eq:simple_group_list}, we have captured up to $k=4$). Many of the cliques which appear just above these with $n_+^\gram=0$ or in the cluster near $\Delta+28n_-^\gram\approx 30$ for type-\B{B} cliques with $n_+^\gram=1$ are slight variations on~\eqref{eq:SUSO_chain} where the mergings of hypermultiplets into bi-fundamentals is not as exact, such as
\begin{equation}
\label{eq:SUSO_chain_inexact}
    \newcommand{\s}{1}
    \begin{tikzpicture}[baseline={([yshift=-0.5ex]current bounding box.center)}]
        \node[inner sep=2, fill=typeB, draw=black, text=white, thick] (a) at (0,0) {\tiny11};
        \node[inner sep=2, fill=typeB, draw=black, text=white, thick] (b) at (\s,0) {\tiny23};
        \node[circle, inner sep=2, fill=typeB, draw=black, text=white, thick] (c) at (2*\s,0) {\tiny30};

        \draw[very thick] (a) to (b) to (c);
    \end{tikzpicture}
    \quad:\quad 
    \begin{aligned}
        G &= \SU(11)\times\SU(23)\times\SO(30) \,,\\
        \H &= 3\times(\rep{1},\rep{23},\rep{1}) + (\rep{11},\rep{23},\rep{1}) + (\rep{1},\rep{23},\rep{30}) \,,\\
        \Delta &= -70 \,, \quad \Delta+28n_-^\gram = 14 \, \quad T_\text{min} = 3 \,.
    \end{aligned}
\end{equation}
Cliques such as those in~\eqref{eq:SUSO_chain}, \eqref{eq:SUSO_chain_inexact} and table~\ref{tab:irreducible_cliques_infty} were considered and bounded in~\cite{Tarazi2021}.

\subsubsection{Anomaly-free cliques}

Having enumerated all irreducible admissible cliques satisfying the bounds of equation~\eqref{eq:delta_28n_double_sided_bound}, we can generate anomaly-free cliques by taking disjoint unions at will. There are two possible shapes:
\begin{equation}
\label{eq:anomaly-free-clique-types}
    \C^\af = \C_{\text{gen},n_+^\gram=1}^\irr \uplus \Big( \biguplus_\alpha \C_{\text{gen},n_+^\gram=0,\alpha}^\irr \Big) \,, \quad \C^\af = \Big( \biguplus_\alpha \C_{\infty,\alpha}^\irr \Big) \uplus \Big( \biguplus_\alpha \C_{\text{gen},n_+^\gram=0,\alpha}^\irr \Big) \,.
\end{equation}
On the left, exactly one of the cliques with $n_+^\gram=1$ is selected and all other irreducible cliques must have $n_+^\gram=0$. On the right, all irreducible cliques have $n_+^\gram=0$ and some members from the infinite families may be selected. The latter are perhaps less interesting, seeing as they necessarily contain only type-\B{B} vertices. In appendix~\ref{app:brane_probe} we review how each infinite family of this shape is truncated to a finite number of consistent theories if one makes some mild assumptions about the BPS spectrum.

The bound
\begin{equation}
    \sum_\alpha(\Delta+28n_-^\gram)(\C_\alpha^\irr) + \max_\alpha\{T_\text{min}(\C_\alpha^\irr)\} \leq 273
\end{equation}
can help to guide the construction of anomaly-free cliques, but ultimately they must satisfy $(\Delta+29T_\text{min})(\C^\af)\leq 273$. We provide the functionality to generate anomaly-free cliques for your favorite $\C_{\text{gen},n_+^\gram=1}^\irr$ alongside the data sets in~\cite{Loges:2023gh2}.

\subsection{On the (conditional) finiteness of the landscape}
\label{sec:finite_landscape}

In order to make headway on enumerating anomaly-free cliques it was necessary to make two key simplifying assumptions, namely that hypermultiplets are charged under at most two simple factors and that the vertices of~\eqref{eq:removed_vertices} are absent. Having done so, we have seen that the resulting ensemble of anomaly-free theories has both $\Delta$ and $T$ bounded above and below and has only a finite number of infinite families, each of which is truncated to a finite number of consistent theories with mild assumptions on the BPS spectrum (see appendix~\ref{app:brane_probe}). It is a natural to ask which of these assumptions may be dropped without changing the shape of the ensemble in any significant way. We have seen in section~\ref{sec:infinite_families} that reincorporating \emph{all} of the vertices of~\eqref{eq:removed_vertices} leads to a qualitatively different situation, where there are infinitely many infinite families where $\Delta$ and $T$ grow unbounded. However, only half of the eight vertices in~\eqref{eq:removed_vertices} actually lead to infinite families of this nature, while the others were avoided because they have $\Delta+28n_-^\gram < 0$.

As we saw above, if irreducible cliques with $\Delta+28n_-^\gram<0$ are absent then all anomaly-free cliques consist of disjoint unions of irreducible cliques for which $0\leq(\Delta+28n_-^\gram)(\C^\irr)\leq 273-T_\text{min}(\C^\irr)$: this implies an upper bound on the number of tensor multiplets of $273$. This bound can be saturated only when $(\Delta+28n_-^\gram)(\C^\irr) = 0$ for each irreducible component and in fact there is a unique anomaly-free theory that achieves $T=273$:
\begin{equation}
    G = \SO(8)^{273} \,, \quad T=273 \,, \quad \H = \emptyset \,, \quad \Delta = -7644 \,, \quad \Gram = \begin{psmallmatrix*}
        -264 & -2_{273}\\
        -2_{273} & -4\,\I_{273\times 273}
    \end{psmallmatrix*} \,.
\end{equation}
This curious solution to the consistency conditions has neither charged nor neutral hypermultiplets and far surpasses the maximal value for $T$ in known F-theory constructions, which to our knowledge is $T=193$~\cite{Aspinwall:1997ye,Candelas:1997eh}.

\medskip

In summary, we have found the following:
\begin{equation}
\label{eq:finiteness_summary}
    \left\{\footnotesize
    \begin{minipage}{4.2cm}
        $\U(1)$, $\SU(2)$, $\SU(3)$,\\
        $\{E_6,\,k\times\rep{27}\}$ ($k=0,1$),\\
        $\{E_7,\,\tfrac{k}{2}\times\rep{56}\}$ ($k=0,1,2,3$),\\
        $\{E_8,\,\emptyset\}$ and $\{F_4,\,\emptyset\}$ all absent \\
        and no $(3+)$-charged hypers
    \end{minipage}
    \right\} \;\implies\;\left\{\hspace{-12pt}
    \begin{minipage}{7cm}
        \footnotesize
        \begin{itemize}
            \setlength\itemsep{-0.25em}
            \item $T\leq 273$
            \item The number of anomaly-free theories\\ with $n_+^\Gram=1$ is finite
            \item The number of infinite families of anomaly-free theories with $n_+^\Gram=0$ is finite
            \item The number of consistent theories\\ with $n_+^\Gram=0$ is finite
        \end{itemize}
    \end{minipage}
    \right.
\end{equation}
We conjecture that this picture remains essentially unchanged if some of our working assumptions are relaxed:
\begin{align*}
    &\text{\textbf{C1)} \eqref{eq:finiteness_summary} continues to hold if the LHS is replaced by } \left\{\scriptsize
    \begin{minipage}{3.9cm}
        $\U(1)$, $\{E_6,\,k\times\rep{27}\}$ ($k=0,1$),\\
        $\{E_7,\,\tfrac{k}{2}\times\rep{56}\}$ ($k=0,1,2,3$),\\
        $\{E_8,\,\emptyset\}$ and $\{F_4,\,\emptyset\}$ all absent
    \end{minipage}
    \right\}\\[5pt]
    &\text{\textbf{C2)} \eqref{eq:finiteness_summary} continues to hold if the LHS is replaced by } \left\{\scriptsize
    \begin{minipage}{3cm}
        $\U(1)$, $\{E_6,\,\emptyset\}$,\\
        $\{E_7,\,\tfrac{k}{2}\times\rep{56}\}$ ($k=0,1$)\\
        and $\{E_8,\,\emptyset\}$ all absent
    \end{minipage}
    \right\}\\
    &\text{\phantom{\textbf{C2)}} and ``273'' on the RHS is replaced by some to-be-determined constant $T_\ast$}
\end{align*}
That is, for \textbf{(C1)} we propose that even when hypermultiplets charged under more than two simple factors are considered and $\SU(2)$ and $\SU(3)$ are reintroduced, there remains a finite number of (families of) irreducible, admissible cliques with $n_+^\gram=1$ ($n_+^\gram=0$) and all cliques with $\Delta+28n_-^\gram<0$ include one or more of the vertices from~\eqref{eq:removed_vertices}.
For \textbf{(C2)}, we strengthen this even further by reintroducing the four vertices of~\eqref{eq:removed_vertices} which have $\Delta+28n_-^\gram<0$ but \emph{do not} satisfy equation~\eqref{eq:infinite_family_inequalities}. Given their negative contribution to the sum in~\eqref{eq:delta_n_sum} the upper bound on $T$ would be quite a bit larger, but we cannot estimate $T_\ast$ without knowing how many times such vertices can co-exist.

\medskip

Even if \textbf{(C2)} holds, the landscape of anomaly-free theories is still wildly infinite, given the construction of section~\ref{sec:infinite_families} where we explicitly built an infinite family of anomaly-free theories for infinitely many choices of seed theory. However, it appears that the only technical obstacle to conclusively achieving a finite landscape is to bound the number of times that the vertices of table~\ref{tab:culprits} may appear, for example by showing that BPS strings kill all but a finite number from this ``doubly-infinite'' family. We find this extremely plausible since, as we saw for the example of equation~\eqref{eq:infinite_family_example}, for most seed theories the values for $T$ and $m$ must be enormous.

\section{Discussion}
\label{sec:discussion}

In this work we have focused on charting the landscape of anomaly-free 6D supergravities with eight supercharges. Leveraging ideas from graph theory, we have seen how to describe anomaly-free theories as cliques in the multigraph $\G$ and provided a general classification of their structure. One of us has found great utility in using graph theory previously for understanding the structure of intersecting brane models~\cite{Loges:2022mao}. There, vertices correspond to the homology classes of D6-branes in a toroidal orientifold and edges correspond to pair-wise compatibility in the sense of their preserving some common supersymmetry. One then searches for cliques which satisfy the tadpole bounds and thus give a consistent brane configuration. The potential applications in high-energy theory are many and it is our hope that such techniques become more widely used, if only as a useful organizational framework.

In the present work we saw how enumerating irreducible admissible cliques can be accomplished in a largely $T$-independent way using a standard branch-and-prune algorithm, which we then implemented in two cases, one for $T=0$ where we allowed for any semi-simple gauge group without $\U(1)$, $\SU(2)$ and $\SU(3)$ factors and the other with a restricted set of simple factors to keep the size of the multigraph within reason. More targeted enumerations, say with $T$ bounded or fixed or more narrow shapes for the gauge group, could also be done. Our analysis was subject to just a few practical limitations:
\begin{itemize}
    \item Hypermultiplets are charged under at most two simple factors. This was built into the multigraph structure since edges connect only two vertices.
    \item Abelian factors and some low-rank simple groups are omitted. Anomaly cancellation is somewhat different for $\U(1)$~\cite{Berkooz:1996iz}, but we see no reason our techniques could not be adapted to address abelian charges as well (although one would again have to contend with infinite families of charge assignments: see, e.g.,~\cite{Park:2011wv}).
    \item Some choices of hypermultiplets for exceptional groups are ignored. This is a necessary evil at present, given the infinite families they lead to.
\end{itemize}
However, we also emphasize the following features which, when taken together, bring us quite far in understanding the detailed structure of the landscape of anomaly-free theories:
\begin{itemize}
    \item The classification is done in a largely $T$-independent way by basing our understanding on the $T$-independent combination $\Delta+28n_-^\gram$.
    \item The allowed hypermultiplet representations are unrestricted.
    \item The form of the gauge group is unrestricted (other than the omitted abelian and simple factors).
\end{itemize}
In addition we have, modulo the conjectures of section~\ref{sec:finite_landscape}, a complete characterization of the infinite families of anomaly-free theories. In particular, the four simple theories of table~\ref{tab:culprits} necessarily appear in infinite families with $T>273$ and $T$ must grow at least linearly with their multiplicity.

\medskip

With an ensemble of anomaly-free theories at our disposal, there are many questions that one would like to answer about its members. Very broadly, which theories have a string theory realization? To have a geometric F-theory construction requires the Kodaira condition be satisfied: which/how many theories pass this non-trivial check? Similarly, assuming BPS completeness one can use the consistency of string probes to rule out theories (and in particular many of the infinite families~\cite{Kim:2019vuc,Tarazi2021}): which/how many anomaly-free theories are in fact consistent? There is also the question of determining the equivalence classes of anomaly-free theories which are related via Higgsing. We leave such questions for the bright future.

\acknowledgments

We thank Hee-Cheol Kim for insightful discussions, both Hee-Cheol Kim and Gary Shiu for comments on an earlier draft, and the 2023 Summer Program of the Simons Center for Geometry and Physics for their kind hospitality.
The work of Y.H.\ and G.L.\ is supported in part by MEXT Leading Initiative for Excellent Young Researchers Grant Number JPMXS0320210099.

\appendix

\section{Trace indices}
\label{app:trace_indices}

\begin{table}[p]
    \centering
    \newcommand{\mycol}{\multicolumn{1}{c|}{}}
    \newcommand{\myrow}[2]{\multicolumn{1}{c|}{\multirow{#1}{*}{#2}}}
    \begin{tabular}{crrrrl}
        \toprule
        $G$ & \multicolumn{1}{c}{$R$} & $A_R$ & $B_R$ & $C_R$ & Notes\\ \midrule
        \\[-12pt]
        \myrow{4}{$\SU(N)$} & $\rep{N}$        &     1 &     1 & 0\\
        \mycol              & $\rep{N(N-1)/2}$ & $N-2$ & $N-8$ & 3\\
        \mycol              & $\rep{N(N+1)/2}$ & $N+2$ & $N+8$ & 3\\
        \mycol              & $\rep{N^2-1}$    &  $2N$ & $2N$  & 6 & A\\
        \\[-5pt]
        \myrow{4}{$\SO(N)$} & $\rep{N}$    &         2 &          4 &                0\\
        \mycol              & $\rep{N(N-1)/2}$ & $2(N-2)$ &  $4(N-8)$ &               12 & A\\
        \mycol              & $\rep{(N-1)(N+2)/2}$ & $2(N+2)$ &  $4(N+8)$ &               12\\
        \mycol              & $\rep{2^{\lfloor\frac{N-1}{2}\rfloor}}$ & $2^{\lfloor\frac{N-5}{2}\rfloor}$ & $-2^{\lfloor\frac{N-5}{2}\rfloor}$ & $3\cdot 2^{\lfloor\frac{N-9}{2}\rfloor}$ & Q for $N\equiv 3,4,5\mod{8}$\\
        \\[-5pt]
        \myrow{3}{$\Sp(N)$} & $\rep{2N}$          &      1 &      1 & 0 & Q\\
        \mycol              & $\rep{(N-1)(2N+1)}$ & $2N-2$ & $2N-8$ & 3\\
        \mycol              & $\rep{N(2N+1)}$     & $2N+2$ & $2N+8$ & 3 & A\\
        \\[-5pt]
        \myrow{4}{$E_6$} & $\rep{27}$          &   6 & 0 &   3\\
        \mycol           & $\rep{78}$          &  24 & 0 &  18 & A\\
        \mycol           & $\rep{351}$         & 150 & 0 & 165\\
        \mycol           & $\rep[\prime]{351}$ & 168 & 0 & 210\\
        \\[-5pt]
        \myrow{4}{$E_7$} & $\rep{56}$   &  12 & 0 &   6 & Q\\
        \mycol           & $\rep{133}$  &  36 & 0 &  24 & A\\
        \mycol           & $\rep{912}$  & 360 & 0 & 372 & Q\\
        \mycol           & $\rep{1463}$ & 660 & 0 & 792\\
        \\[-5pt]
        \myrow{2}{$E_8$} & $\rep{248}$  &   60 & 0 &   36 & A\\
        \mycol           & $\rep{3875}$ & 1500 & 0 & 1548\\
        \\[-5pt]
        \myrow{4}{$F_4$} & $\rep{26}$  &   6 & 0 &   3\\
        \mycol           & $\rep{52}$  &  18 & 0 &  15 & A\\
        \mycol           & $\rep{273}$ & 126 & 0 & 147\\
        \mycol           & $\rep{324}$ & 162 & 0 & 207\\
        \\[-5pt]
        \myrow{4}{$G_2$} & $\rep{7}$  &  2 & 0 &   1\\
        \mycol           & $\rep{14}$ &  8 & 0 &  10 & A\\
        \mycol           & $\rep{27}$ & 18 & 0 &  27\\
        \mycol           & $\rep{64}$ & 64 & 0 & 152\\
        \\[-12pt]
        \bottomrule
    \end{tabular}
    \caption{Some common, low-dimensional irreps and their indices $A_R$, $B_R$ and $C_R$ for our chosen normalization, equation~\eqref{eq:ABC_def_in_app} (A$\,=\,$adjoint, Q$\,=\,$quaternionic).}
    \label{tab:irreps}
\end{table}

In this appendix we review the results of~\cite{Okubo82}, in which the calculation of the trace indices $A_R$, $B_R$ and $C_R$ for any representation is made explicit. In table~\ref{tab:irreps} are collected the indices $A_R$, $B_R$ and $C_R$ for some low-dimensional representations in our normalization. Recall from~\eqref{eq:ABC_def} that we are taking
\begin{equation}
\label{eq:ABC_def_in_app}
    \beta_i\tr_R(\F_i^2) = A_R\tr(\F_i^2) \,, \qquad \beta_i^2\tr_R(\F_i^4) = B_R\tr(\F_i^4) + C_R[\tr(\F_i^2)]^2 \,,
\end{equation}
with $\beta_i$ given in table~\ref{tab:lambdas}.

The index $A_R$ is simply proportional to the quadratic Casimir $I_2$:
\begin{equation}
    A_R = \frac{\dim R}{\dim\adj}I_2(R) \,.
\end{equation}
The indices $B_R$ and $C_R$ are related in a fairly simple way to $I_2$ and a particular quartic Casimir $J_4$ advocated for in~\cite{Okubo82} whose precise definition we omit. It is also helpful to introduce
\begin{equation}
    K(R) = \frac{\dim\adj}{2(\dim\adj+2)\dim R}\left(6 - \frac{I_2(\adj)}{I_2(R)}\right) \,.
\end{equation}
As explained in~\cite{Okubo82}, the utility of $J_4$ over other possible quartic Casimirs is that it vanishes identically iff the only quartic Casimir is $(I_2)^2$. This is the case for $A_1$, $A_2$ and the five exceptional groups, in which case
\begin{equation}
    \tr_R(F_i^4) = K(R)[\tr_R(F_i^2)]^2 \quad\implies\quad B_R=0 \,, \;\; C_R = K(R)A_R^2 \,.
\end{equation}
Otherwise, it can be shown that
\begin{equation}
    \frac{\tr_R(\F_i^4) - K(R)[\tr_R(\F_i^2)]^2}{\tr_{R_0}(\F_i^4) - K(R_0)[\tr_{R_0}(\F_i^2)]^2} = \frac{\dim R}{\dim R_0}\frac{J_4(R)}{J_4(R_0)}
\end{equation}
where $R_0$ is any reference irrep for which $J_4(R_0)\neq 0$. Unpacking this using~\eqref{eq:ABC_def_in_app} gives
\begin{equation}
\label{eq:BC_index}
    \begin{aligned}
        B_R &= \frac{\dim R}{\dim R_0}\frac{J_4(R)}{J_4(R_0)}\,B_{R_0} \,,\\
        C_R &= \frac{\dim R}{\dim R_0}\frac{J_4(R)}{J_4(R_0)}\Big[ C_{R_0} - K(R_0)A_{R_0}^2 \Big] + K(R)A_R^2 \,.
    \end{aligned}
\end{equation}

It only remains to explain how to compute $I_2(R)$ and $J_4(R)$. If we write the highest-weight vector as $\lambda=(\lambda_1,\ldots,\lambda_n)$ and define $W_1$, $W_2$, $W_3$, $\sigma_j^{(0)}$ and $\sigma_j$ by
\begin{itemize}
    \item $A_{n\geq 1}$: $W_1=1$, $W_2=(n+1)^2+1$, $W_3=\frac{2(n+1)^2-3}{n+1}$ and
    \begin{equation}
        \sigma_j^{(0)} = \tfrac{n+2}{2} - j \,, \qquad \sigma_j = \sum_{k=j}^n(\lambda_k+1) - \sum_{k=1}^n\tfrac{k(\lambda_k+1)}{n+1} \,, \qquad j\in\{1,2,\ldots,n+1\}
    \end{equation}
    \item $B_{n\geq 3}$: $W_1 = 1$, $W_2 = \frac{2n^2+n+2}{8}$, $W_3 = \frac{4n+1}{8}$ and
    \begin{equation}
        \sigma_j^{(0)} = n - j + \tfrac{1}{2} \,, \qquad \sigma_j = \sigma_j^{(0)} + \sum_{k=j}^n\lambda_k - \tfrac{\lambda_n}{2} \,, \qquad j\in\{1,2,\ldots,n\}
    \end{equation}
    \item $C_{n\geq 2}$: $W_1 = \frac{1}{2}$, $W_2 = \frac{2n^2+n+2}{8}$, $W_3 = \frac{4n+1}{8}$ and
    \begin{equation}
        \sigma_j^{(0)} = n - j + 1 \,, \qquad \sigma_j = \sigma_j^{(0)} + \sum_{k=1}^j\lambda_k \,, \qquad j\in\{1,2,\ldots,n\}
    \end{equation}
    \item $D_{n\geq 4}$: $W_1 = 1$, $W_2 = \frac{2n^2-n+2}{8}$, $W_3 = \frac{4n-1}{8}$ and
    \begin{equation}
        \sigma_j^{(0)} = n - j \,, \qquad \sigma_j = \sigma_j^{(0)} + \sum_{k=j}^n\lambda_k - \tfrac{\lambda_{n-1}+\lambda_n}{2} \,, \qquad j\in\{1,2,\ldots,n\}
    \end{equation}
\end{itemize}
then we have
\begin{equation}
\begin{aligned}
    I_2(R) &= W_1\sum_j\big((\sigma_j)^2 - (\sigma_j^{(0)})^2\big) \,,\\
    J_4(R) &= W_2\sum_j(\sigma_j)^4 + W_3\Big(\sum_j(\sigma_j)^2\Big)^2 + \frac{\dim(\adj)J_4(\adj)}{240} \,,
\end{aligned}
\end{equation}
for the classical groups, where
\begin{equation}
\begin{aligned}
    A_{n\geq 1}: & & J_4(\adj) &= \tfrac{1}{3}(n^2 - 1)(n-2)(n+3)(n+4) \,,\\
    B_{n\geq 3}: & & J_4(\adj) &= \tfrac{1}{12}(n^2 - 1)(2n - 1)(2n + 3)(2n - 7) \,,\\
    C_{n\geq 2}: & & J_4(\adj) &= \tfrac{1}{6}(n^2 - 1)(2n - 1)(2n + 3)(n + 4) \,,\\
    D_{n\geq 4}: & & J_4(\adj) &= \tfrac{1}{6}(n^2 - 1)(2n + 1)(2n - 3)(n - 4) \,.
\end{aligned}
\end{equation}
The fact that $J_4(\adj)=0$ for $D_4$ even though $J_4$ is not identically zero is related to triality. Otherwise, the adjoint can serve as the reference irrep in~\eqref{eq:BC_index}. The expression for $I_2(R)$ above is nothing but the usual
\begin{equation}
    I_2(R) \propto \langle \lambda,\lambda+2\rho\rangle = \langle\lambda+\rho,\lambda+\rho\rangle - \langle \rho,\rho\rangle \,,
\end{equation}
where $\langle \,,\,\rangle$ is the inner product in the weight space and $\rho$ is half the sum of positive roots. This can be used directly for the exceptional groups.

\bigskip

With our normalization $A_R$ and $B_R$ are always integers and $C_R$ is almost always an integer. In fact, $C_R$ is almost always a multiple of three so that $b_i\cdot b_i\in\Z$ is guaranteed. The only exceptions are $C_R\in\frac{1}{2}\Z$ for $A_1\sim\SU(2)$, $A_2\sim\SU(3)$, $B_3\sim\SO(7)$ and $D_4\sim\SO(8)$, and $C_R\in\Z$ for $G_2$. However, the modular conditions for $A_1$, $A_2$ and $G_2$ ensure that $b_i\cdot b_i\in\Z$ and also one observes that
\begin{equation}
\begin{aligned}
    B_3,D_4 &: \quad & \frac{C_R}{3} + \frac{B_R}{4} &\in \Z \,,
\end{aligned}
\end{equation}
so that $b_i\cdot b_i\in\Z$ for these two groups provided the $B$-constraint is satisfied. Similarly, one can understand why $b_0\cdot b_i\in\Z$ by noticing the following patterns:
\begin{equation}
\begin{aligned}
    A_{n\geq3},C_{n\geq2} &: & A_R-B_R &\equiv 0\mod{6} \,,\\
    B_{n\geq3},D_{n\geq4} &: & A_R+B_R &\equiv 0\mod{6} \,,\\
    E_6,E_7,E_8,F_4 &: & A_R &\equiv 0\mod{6} \,,\\
    A_1,A_2,G_2 &: \quad & A_R - 2C_R &\equiv 0 \mod{6} \,.
\end{aligned}
\end{equation}
Observations such as these were made precise in~\cite{Kumar:2010ru} and~\cite{Monnier2018}.

\section{Admissibility}
\label{app:admissibility}
In this appendix, we discuss details of the consistency conditions in section~\ref{sec:consistency_conditions}. In appendix~\ref{app:lattices}, we explain how the unimodular condition~\eqref{eq:unimodular_embedding} can be checked given only the Gram matrix $\Gram$. In appendix~\ref{app:j_exists}, we argue how to check for the existence of $j$ with \eqref{eq:j_condition}.

\subsection{Unimodular lattice embeddings}
\label{app:lattices}

Given only the Gram matrix $\Gram$ for the anomaly lattice $\Lambda$ with signature $(n_+^\Gram,n_-^\Gram)$, when does there exist an embedding of $\Lambda$ into a unimodular lattice of signature $(1,T)$? Luckily, the classification of unimodular lattices of indefinite signature is very simple. For all $T$, there is, up to $\mathrm{O}(1,T)$ transformations, a unique odd unimodular lattice $\mathrm{I}_{1,T}$ of signature $(1,T)$ which has a simple description as $\Z^{1+T}$ equipped with the inner product $\Omega=\diag(1,-1,\ldots,-1)$. Only for $T\equiv 1\mod{8}$ does there exist an even unimodular lattice $\mathrm{II}_{1,T}$ of signature $(1,T)$, again unique up to isomorphism. The lattice $\mathrm{II}_{1,1}$ is often written as $U$, which can be simply described as $\Z^2$ with inner product $\Omega=\begin{psmallmatrix*}
    0 & 1\\ 1 & 0
\end{psmallmatrix*}$. Using this, we can leverage the uniqueness to write
\begin{equation}
    \mathrm{II}_{1,8k+1} = U\oplus\underbrace{(-E_8)\oplus\ldots\oplus(-E_8)}_{k} \,,
\end{equation}
where $(-E_8)$ is the $E_8$ root lattice with reversed signature.

\medskip

For any integer lattice $\Lambda$ there is a canonical way to view $\Lambda$ as a sublattice of $\Lambda^\ast$ (the dual lattice of $\Lambda$) and thus define the discriminant group,
\begin{equation}
    \Lambda^\ast / \Lambda \;\cong\; (\Z/a_1\Z)\oplus\cdots\oplus(\Z/a_r\Z) \,.
\end{equation}
The so-called Smith invariants $a_j$ are positive integers satisfying $a_j|a_{j+1}$ that in some sense capture to what extent the lattice $\Lambda$ is not unimodular: a unimodular lattice has $\Lambda^\ast/\Lambda\cong 0$. They can be computed by bringing the Gram matrix $\Gram(\Lambda)$ to Smith normal form,
\begin{equation}
    S\,\Gram(\Lambda)\,T = D \equiv \diag(a_1,a_2,\ldots,a_r,0,\ldots,0) \,, \qquad a_j|a_{j+1} \,,
\end{equation}
where $S$ and $T$ are both unimodular matrices. We will write $\ell(\Lambda^\ast/\Lambda)$ for the minimal number of generators of $\Lambda^\ast/\Lambda$, i.e.\ the number of $a_j$ which are larger than one.

\medskip

Determining if $\Lambda$ can be embedded into a unimodular lattice of a particular signature breaks into two parts. The first is to embed $\Lambda$ in an integer lattice $M_\Lambda$ with the same signature but smaller discriminant group. The second is to determine whether $M_\Lambda$ has a \emph{primitive}\footnote{An embedding $\Lambda\hookrightarrow\Lambda'$ is primitive if $\Lambda'/\Lambda$ is torsion-free, i.e.\ $\Lambda'/\Lambda\cong \Z^n$ for some $n$.} embedding into a unimodular lattice of signature $(1,T)$: precisely when this is possible was understood in~\cite{nikulin1980integral}.

Let us now flesh out these ideas in turn. $\Lambda$ can be extended to an integer lattice $M_\Lambda$ of the same signature through a sequence
\begin{equation}
    \Lambda = \Lambda_0 \hookrightarrow \Lambda_1 \hookrightarrow\cdots\hookrightarrow \Lambda_n = M_\Lambda \,,
\end{equation}
where at each step the discriminant group is replaced by a proper subgroup, $\Lambda_{m+1}^\ast/\Lambda_{m+1}<\Lambda_m^\ast/\Lambda_m$. As we now review, we can always ensure that the Smith invariants $a_j$ for $M_\Lambda$ are square-free, at most two $a_j$ are larger than one and at most one $a_j\neq0$ is even. The matrices $S$ and $T$ which bring $\Gram(\Lambda_m)$ to Smith normal form provide convenient bases for identifying vectors to be added to $\Lambda_m$; if we make a change of basis using the unimodular matrix $(S^{-1})^\top$ then the new Gram matrix is
\begin{equation}
    \Gram_\text{new}(\Lambda_m) = (S^{-1})\Gram(\Lambda_m)(S^{-1})^\top = DT(S^{-1})^\top \,.
\end{equation}
On the RHS the multiplication by $D$ means that the $j^\text{th}$ row of $\Gram_\text{new}$ is divisible by $a_j$ (and since it is symmetric, also the corresponding column). Therefore if $a_j=s^2b$ with $b$ square-free, replacing the basis vector $e_j$ by $e_j/s$ results in an integer lattice $\Lambda_{m+1}$ which contains $\Lambda_m$ as an index-$s$ sublattice. In doing so the invariant $a_j$ is replaced by its square-free part.

After having made all invariants square-free, we can ensure that each prime $p$ divides at most two of the $a_j$. If $\Lambda_m^\ast/\Lambda_m$ contains a subgroup $(\Z/p\Z)^3$ then like before the matrix $S$ for $\Lambda_m$ provides a basis containing three basis vectors $\{e_j\}_{j=1,2,3}$ with $e_i\cdot e_j=pg_{ij}$, $p\nmid g_{ii}$ and $p|e_j\cdot e_k$ for \emph{all} basis vectors $e_k$. Adding the vector
\begin{equation}
    e_\text{new} = \frac{z_1e_1+z_2e_2+z_3e_3}{p} \,, \qquad z_j\in\{0,1,\ldots,p-1\} \,,
\end{equation}
to $\Lambda_m$ produces a lattice $\Lambda_{m+1}$ which remains integral provided the $z_j$ are chosen so that
\begin{equation}
    z_ig_{ij}z_j \equiv 0 \mod{p} \,.
\end{equation}
The Chavalley-Warning theorem~\cite{chevalley1935bemerkung} ensures that the above always has a non-trivial solution. $e_\text{new}$ generates a $(\Z/p\Z)$ subgroup of $\Lambda_m^\ast/\Lambda_m$ and thus $\Lambda_{m+1}^\ast/\Lambda_{m+1}$ is related to $\Lambda_m^\ast/\Lambda_m$ by the replacement $(\Z/p\Z)^3\to(\Z/p\Z)$. When $\Lambda_m^\ast/\Lambda_m$ only contains a $(\Z/p\Z)^2$ subgroup then by the same argument adding the vector $e_\text{new} = (z_1e_1+z_2e_2)/p$ produces an integer lattice provided $z_i$ are a non-trivial solution of $z_ig_{ij}z_j\equiv 0\mod{p}$. Only for $p=2$ is there guaranteed to be a non-trivial solution ($z_1=z_2=1$). Whenever a non-trivial solution \emph{does} exist then adding the vector $e_\text{new}$ to $\Lambda_m$ yields an integer lattice $\Lambda_{m+1}$ which has discriminant group with the $(\Z/p\Z)^2$ factor removed. At the end of this process the lattice $M_\Lambda$ will have discriminant group of the form
\begin{equation}
    A\equiv M_\Lambda^\ast/M_\Lambda \cong \bigoplus_{p\,:\,\text{prime}} A_p = \bigoplus_{p\,:\,\text{prime}} (\Z/p\Z)^{\ell(A_p)} \,, \quad \ell(A_2)\leq 1 \,,\;\; \ell(A_{p>2})\leq 2 \,,
\end{equation}
and thus $\ell(M_\Lambda^\ast/M_\Lambda)=\max_p \ell(A_p) \leq 2$. As usual, two is the oddest prime and for the next step it may be necessary to \emph{not} remove all powers of two, settling instead for $A_2\in\{0,(\Z/2\Z),(\Z/8\Z)\}$. That this may be necessary is related to the fact that two is the only prime power that cannot be written as the difference of two squares.

\medskip

After having constructed the lattice $M_\Lambda$, we would like to determine when $M_\Lambda$ may be primitively embedded into either $\mathrm{I}_{1,T}$ or $\mathrm{II}_{1,T}$. A complete understanding of primitive embeddings into indefinite unimodular lattices first appeared in the work of Nikulin~\cite{nikulin1980integral} and leverages $p$-localization and the classification of finite quadratic forms. For $T\equiv 1\mod{8}$ the condition
\begin{equation}
    (1+T)-(n_+ + n_-) \geq \ell(A)
\end{equation}
is necessary in order for the even lattice $M_\Lambda$ to have an embedding into $\mathrm{II}_{1,T}$, and having the inequality be strictly satisfied is sufficient. We direct the interested reader to Theorem~1.12.2 of~\cite{nikulin1980integral} for an exact characterization in the case of equality. There are similar conditions for embedding into $\mathrm{I}_{1,T}$ (see Theorems~1.16.5, 1.16.7 and~1.16.8 of~\cite{nikulin1980integral}), where now
\begin{equation}
    (1+T) - (n_++n_-) \geq \begin{cases}
        \ell(A) & (1-n_+)-(T-n_-) \equiv 0\mod{8}\\
        \max\{\ell(A) ,\, \ell(A_2) + 1\} & \text{otherwise}
    \end{cases}
\end{equation}
is necessary and
\begin{equation}
    (1+T) - (n_++n_-) > \begin{cases}
        \ell(A) & (1-n_+)-(T-n_-) \equiv 0\mod{8}\\
        \max\{\ell(A) ,\, \ell(A_2) + 2\} & \text{otherwise}
    \end{cases}
\end{equation}
is sufficient. One corollary of these results is that since we can always ensure $\ell(A_2)\leq 1$ and $\ell(A_{p>2})\leq 2$ via the sequence of changes discussed above, $(1+T)-(n_++n_-) > 3 \geq \max\{\ell(A),\ell(A_2)+2\}$, i.e.\ $T\geq n_+ + n_- + 3$, is sufficient to ensure a primitive embedding exists. This fact is used in section~\ref{sec:infinite_families}. Also, it is more efficient to simply check the necessary conditions during the branch-and-prune algorithm, only checking for an embedding $\Lambda\hookrightarrow \mathrm{I}_{1,T}$ or $\Lambda\hookrightarrow\mathrm{II}_{1,T}$ in detail for the irreducible cliques which satisfy the bound $\Delta+28n_-^\gram\leq 273$.

\subsection{Positivity}
\label{app:j_exists}

As discussed in section~\ref{sec:consistency_conditions}, one of the consistency conditions that we require is the existence of a vector $j\in\R^{1,T}$ satisfying $j\cdot j>0$ and $j\cdot b_i>0$, which ensures that the gauge kinetic terms have the correct sign. When $n_+^\Gram=1$ or $n_-^\Gram=T$, the vectors $b_I$ are uniquely fixed up to $\operatorname{O}(1,T)$ transformations and the existence of $j$ is unambiguous. In contrast, we prove that if $n_+^\Gram=0$ and $n_-^\Gram<T$ then there is always enough freedom in choosing the vectors $b_I$ such that $j$ exists satisfying not only $j\cdot j>0$ and $j\cdot b_i>0$, but also $j\cdot b_0>0$. We will also discuss how to use linear and quadratic programming to check for the existence of $j$ when $n_+^\Gram=1$ or $n_-^\Gram=T$ using the integer matrix $\Gram$ directly, without the need to find the vectors $b_I$ as an intermediate step.

It is helpful to first rephrase this into a geometric condition on the vectors $b_i$ alone. It will be useful to refer to the future and past light-cones,
\begin{equation}
    C^\pm = \big\{ x\in\R^{1,T}\;\big|\; x\cdot x = 0 \,, \; \pm x^0 \geq 0 \big\} \,,
\end{equation}
and the convex-hull $\conv(b_i)=\{\sum_i\lambda_i b_i\,|\,\lambda_i\in[0,1]\,\text{and}\,\sum_i\lambda_i=1\}$ of the vectors $b_i$.\footnote{To include the condition $j\cdot b_0>0$, replace $\conv(b_i)$ with $\conv(b_I)$ in everything that follows.} Note that both $C^+$ and $C^-$ contain the origin.

\begin{claim}\label{clm:j_exists}
    Given a collection of vectors $b_i\in\R^{1,T}$, a time-like vector $j\in\R^{1,T}$ satisfying $j\cdot b_i>0$ for each $i$ exists iff $\conv(b_i)\cap C^+=\emptyset$ or $\conv(b_i)\cap C^-=\emptyset$.
\end{claim}

\begin{proof}
    First note that the statements $\conv(b_i)\cap C^\pm = \emptyset$ are invariant under $\operatorname{O}(1,T)$ transformations.
    For $\Rightarrow$, pick $\Omega=\diag(1,-1,\ldots,-1)$ and choose coordinates where $j=(1,0,\ldots,0)$. Then $\conv(b_i)\cap C^-=\emptyset$ immediately follows from $j\cdot b_i = (b_i)^0 > 0$.
    For $\Leftarrow$, let's assume that $\conv(b_i)\cap C^-=\emptyset$, without loss of generality. If $\conv(b_i)$ lies entirely in the interior of $\conv(C^-)$ then clearly all $b_i$ are past-directed time-like vectors, in which case $j$ can be chosen to be any past-directed time-like vector. Otherwise, $\conv(b_i)$ and $\conv(C^-)$ are two disjoint convex sets and since $\conv(b_i)$ is compact, the hyperplane separation theorem implies that there are hyperplanes through the origin such that all $b_i$ are strictly to one side. Because such hyperplanes meet $C^-$ only at the origin, it is clear that they are space-like surfaces, in which case $j$ can be chosen to be one of the future-directed time-like normal vectors.
\end{proof}

As a special case, if there is a non-negative linear combination of $b_i$ which vanishes, i.e.\ $\sum_i x_i^2b_i = 0\in\conv(b_i)$, then $j$ clearly does not exist since $j\cdot\sum x_i^2b_i = 0$ contradicts $j\cdot b_i>0$. With the above rephrasing we can address the case where neither inequality of equation~\eqref{eq:eigenvalue_bounds_G} is saturated.

\begin{claim}\label{clm:case3}
    Given a Gram matrix $\Gram$ with $n_+=0$ and $n_-<T$, one can always find vectors $b_I,j\in\R^{1,T}$ which realize $\Gram$ and satisfy $j\cdot j>0$ and $j\cdot b_I>0$.
\end{claim}

\begin{proof}
    In light of claim~\ref{clm:j_exists}, it suffices to show that there is enough freedom in picking the $b_I$ to ensure that ${\conv(b_I)\cap C^-=\emptyset}$. The key observation is that since $\operatorname{span}(b_I)=\{\sum_I\lambda_I b_I|\lambda_I\in\mathbb{R}\}$ contains no time-like vectors, $\conv(b_I)$ can intersect $C^+\cup C^-$ at most along a light-like line segment. If the $b_I$ can be chosen so that this line segment does not contain the origin, then we are done since then one of $\conv(b_I)\cap C^\pm$ will be empty.

    Start by diagonalizing $\Gram$; there are $n_-^\Gram$ vectors $z_a^-$ and $n_0^\Gram$ vectors $z_b^0$ satisfying
    \begin{equation}
        \Gram_{IJ}(z_a^-)^I = \lambda_a^-(z_a^-)^J \,, \qquad \Gram_{IJ}(z_b^0)^I = 0 \,, 
        \qquad a=1,\ldots,n_-^\Gram \,,
        \qquad b=1,\ldots,n_0^\Gram \,,
    \end{equation}
    with $\lambda_a^-<0$ and normalization $\sum_I[(z_a^-)^I]^2 = \sum_I[(z_b^0)^I]^2 = 1$. The vectors $v_a^-\equiv (z_a^-)^Ib_I\in\R^{1,T}$ and $v_b^0\equiv (z_b^0)^Ib_I\in\R^{1,T}$ are all orthogonal. However, two light-like vectors are orthogonal iff they are proportional, so it must be that all $v_b^0$ are proportional. That is, all $v_b^0$ lie along the same direction in an $\R^{1,1}$ subspace and the $v_a^-$ form a basis of an orthogonal $\R^{0,n_-^\Gram}$ subspace. We can choose coordinates where $\Omega=\diag(1,-1,\ldots,-1)$ and
    \begin{equation}\label{eq:v_vector_coords}
        v_a^- = \sqrt{-\lambda_a^-}\,(0,0,\hat{e}_a) + g_a^-(1,1,\vec{0}) \,, \qquad v_b^0 = g_b(1,1,\vec{0}) \,,
    \end{equation}
    with $\hat{e}_a$ a collection of $n_-^\Gram$ orthonormal unit vectors and $g_a^-$ and $g_b$ are a set of constants we are free to choose.\footnote{From a choice for the parameters $g_a^-$ and $g_b$ one can reconstruct the $b_I$ by solving $v_a^-=(z_a^-)^Ib_I$ and $v_b^0=(z_b^0)^Ib_I$.} Actually, for our purposes it suffices to simply pick $g_a^-=0$. We would then like to show that there always exists a choice for $g_b$ such that $\conv(b_I)\cap C^-=\emptyset$.    
    
    The eigenvectors have inner products
    \begin{equation}
        v_a^-\cdot v_{a^\prime}^- = \lambda_a^-\delta_{aa^\prime} \,, \qquad v_a^-\cdot v_b^0 = 0 \,, \qquad v_b^0\cdot v_{b^\prime}^0 = 0 \,.
    \end{equation}
    However, it is helpful to define a new positive-semidefinite inner product on $\operatorname{span}(b_I)$ by\footnote{Essentially by dropping the zeroth coordinate in~\eqref{eq:v_vector_coords} and flipping signs.}
    \begin{equation}
        \langle v_a^-,v_{a^\prime}^-\rangle = |\lambda_a^-|\delta_{aa^\prime} \,, \qquad \langle v_a^-,v_b^0\rangle = 0 \,, \qquad \langle v_b^0,v_{b^\prime}^0\rangle = g_bg_{b^\prime} \,.
    \end{equation}
    Let $v$ be an arbitrary vector in $\conv(b_I)$,
    \begin{equation}
        \begin{aligned}
            v &= c_-^av_a^- + c_0^bv_b^0 = v^Ib_I \,, &\qquad v^I &= c_-^a(z_a^-)^I + c_0^b(z_b^0)^I \,,\\
            & & v^I &\geq 0 \,,\qquad \sum_Iv^I = 1 \,.
        \end{aligned}
    \end{equation}
    for which
    \begin{equation}
        \langle v,v\rangle = \sum_a|\lambda_a^-|(c_-^a)^2 + \Big(\sum_bg_bc_0^b\Big)^2 \geq 0 \,.
    \end{equation}
    Clearly $\langle v,v\rangle=0$ iff $c_-^a=0$ and $\sum_bg_bc_0^b=0$, so if $g_b$ can be chosen such that $\sum_bg_bc_0^b\neq0$ for all $v\in\conv(b_I)$ then $0\notin\conv(b_I)$ and $j$ exists. The inequalities
    \begin{equation}
    \label{eq:c0b_condition}
        v_I = c_0^b(z_b^0)^I \geq 0 \;\;\text{and}\;\; \sum_Iv^I=\sum_Ic_0^b(z_b^0)^I = 1
    \end{equation}
    define a region for $c_0^b$ which is either empty or convex and disjoint from the origin.\footnote{Clearly, the origin $c_0^b=0$ does not satisfy \eqref{eq:c0b_condition}. It is also obvious that if $c_0^{\prime b}$ and $c_0^{\prime\prime b}$ satisfy \eqref{eq:c0b_condition} then so too does $\lambda c_0^{\prime b}+(1-\lambda)c_0^{\prime\prime b}$ $(0\leq\lambda\leq1)$.} In the former case the existence of $g_b$ is vacuously true and in the latter case there exists a hyperplane separating the convex region from the origin and $g_b$ can be chosen to be the unit normal with appropriate orientation. For this choice of $g_b$, $\langle v,v\rangle$ is never zero for $v\in\conv(b_I)$, i.e.\ $v=0$ is not in $\conv(b_I)$ and thus $j$ exists.
\end{proof}

This proof is constructive, since it shows how to correct the vectors $b_I$ in order to satisfy $j\cdot b_I>0$. For example, the $k=12$ theory in~\eqref{eq:12-clique-example} has the following Gram matrix for $T=15$,
\begin{equation}
    \Gram = \begin{psmallmatrix}
        -6 & -1 &  1 & -1 &  1 & \cdots & -1 &  1\\
        -1 & -1 &  1 &  0 &  0 & \cdots &  0 &  0\\
         1 &  1 & -1 &  0 &  0 & \cdots &  0 &  0\\
        -1 &  0 &  0 & -1 &  1 & \cdots &  0 &  0\\
         1 &  0 &  0 &  1 & -1 & \cdots &  0 &  0\\[-5pt]
        \vdots & \vdots & \vdots & \vdots & \vdots & \ddots & \vdots & \vdots\\
        -1 &  0 &  0 &  0 &  0 & \cdots & -1 &  1\\
         1 &  0 &  0 &  0 &  0 & \cdots &  1 & -1\\
    \end{psmallmatrix}_{13\times13} \,,
\end{equation}
and there is a natural choice for the vectors $b_I$, namely ($\ell=1,2,\ldots,6$)
\begin{equation}
    b_0 = (3; 1^9, 1^6) \,, \quad b_{2\ell-1} = -b_{2\ell} = (0; 0^9, 0^{\ell-1}, 1, 0^{6-\ell}) \,,
\end{equation}
but this clearly disallows $j\cdot b_i>0$. However, since $n_+^\Gram = 0$ and $n_-^\Gram = 6 < T$, we have just seen that it is always possible to add appropriately-chosen multiples of a null vector to each $b_I$ in order to ensure that $\conv(b_I)$ no longer contains the origin. For example,
\begin{equation}
    b_0 = (3; 1^9, 1^6) \,, \quad b_{2\ell-1} = (3; 1^9, 0^{\ell-1}, 1, 0^{6-\ell}) \,, \quad b_{2\ell} = (3; 1^9, 0^{\ell-1}, -1, 0^{6-\ell}) \,,
\end{equation}
works and allows for $j=(1;0^9,0^6)$.

\medskip

The above proof also makes clear why $(n_+,n_-)=(0,T)$ must be treated separately, since there isn't enough room to have orthogonal $\R^{1,1}$ and $\R^{0,T}$ subspaces. When either $n_-^\Gram=T$ or $n_+^\Gram=1$ the existence of $j$ is not guaranteed, but can be determined using linear and quadratic programs, respectively.
\begin{itemize}
    \item $(n_+,n_-)=(0,T)$: Much of the discussion for claim~\ref{clm:case3} goes through. However, as just mentioned, we must take $g_b=0$ so that if $z\in\ker\Gram$ then $z^Ib_I=0$. In other words, $j$ exists iff $0\notin\conv(b_j)$. We can phrase the existence of $j$ using the following linear program,
    \begin{equation}
        \begin{array}{rl}
            \text{maximize:} \quad& \sum_I c^b(z_b^0)^I\\
            \text{subject to:} \quad& c^b(z_b^0)^i \geq 0\\
            \text{subject to:} \quad& c^b(z_b^0)^0 = 0\\
            & c^b\in[-1,1]
        \end{array}
    \end{equation}
    where, as before, $z_b^0$ are a basis for $\ker\Gram$ indexed by $b$. The linear constraint $c^b(z_b^0)^0 = 0$ ensures that the $c_0^b$ parametrize vectors with no $b_0$ component. $j$ exists iff the maximum value of this linear program is zero, necessarily occuring at $c_0^b=0$ since the vectors $z_b^0$ are linearly independent.

    \item $n_+=1$: Let $z_+$ be the eigenvector of $\Gram$ with positive eigenvalue: $z_+$ defines an orientation since wlg we can say that the time-like vector $z_+^Ib_I$ is future-directed.
    The following two quadratic programs,
    \begin{equation}
        \begin{array}{rl}
            \text{maximize:} \quad& c^i\gram_{ij}c^j\\
            \text{subject to:} \quad& \pm z_+^I\Gram_{Ij}c^j \geq 0\\
            & c^i \in [0,1]
        \end{array}
    \end{equation}
    can be used together to determine if $\conv(b_i)\cap C^\pm = \emptyset$: $\conv(b_i)\cap \conv(C^\pm)$ are non-empty if the maximum value is strictly positive or if there is a non-trivial solution with $c^i\gram_{ij}c^j=0$.
\end{itemize}

\section{Multigraph algorithms}
\label{app:multigraph_algorithms}

In this appendix, we provide details about the construction of the multigraph outlined in section~\ref{sec:graph_construction} and the pruning conditions of section~\ref{sec:clique_construction}. In~\ref{app:vertex_construction}, a way to enumerate the vertices satisfying the B-constraint~\eqref{eq:B_constraint} is presented. In appendix~\ref{app:edge_construction}, an efficient algorithm to enumerate the edges connecting type-\A{A} vertices is provided.
In appendix~\ref{app:pruning}, we provide details regarding the refined pruning condition~\eqref{eq:prune_rule_improved}.

\subsection{Vertex construction}
\label{app:vertex_construction}

In this section we discuss an algorithm to find all solutions of the $B$-constraint~\eqref{eq:B_constraint} given a collection of irreps for a simple group and subject to a bound $H\leq H_\text{max}$. That is, given a set $\{(H_{R_j},B_{R_j})\}_{j=1}^m$ we would like to find all choices for $n_{R_j}\geq 0$ such that
\begin{equation}
    \sum_{j=1}^m n_{R_j}H_{R_j} \leq H_\text{max} \;,\qquad \sum_{j=1}^m n_{R_j}B_{R_j} = B_\adj \,.
\end{equation}
This is accomplished in several steps. First, sort the irreps $\{R_j\}$ in order of decreasing $|B_R|/H_R$. Next, compute the bounds\footnote{Note that if $\sum_{j>\ell}n_{R_j}B_{R_j}=B$ is unattainable then $H_\text{min}^{(\ell)}(B)=\inf\emptyset=+\infty$.}
\begin{equation}
    H_\text{min}^{(\ell)}(B) = \inf\Big\{ \sum_{j>\ell}n_{R_j}H_{R_j} \;\Big|\; n_{R_j}\in\Z_{\geq 0} \,, \;\; \sum_{j>\ell}n_{R_j}B_{R_j} = B \Big\} \,,
\end{equation}
which represent the minimum number of hypers needed to achieve a ``$B$-sum'' of $B$ using only the irreps $\{R_j\}_{j>\ell}$. These bounds can be computed relatively quickly using the following recursion:
\begin{equation}
    \begin{aligned}
        H_\text{min}^{(m)}(B) &= \begin{cases}
            0 & B=0\\
            +\infty & B\neq 0
        \end{cases} \,,\\
        H_\text{min}^{(\ell-1)}(B) &= \inf_{n\geq0}\Big\{ H_\text{min}^{(\ell)}(B-nB_{R_\ell}) + nH_{R_\ell}\Big\} \,.
    \end{aligned}
\end{equation}
We can then construct solutions to the $B$-constraint by iteratively incorporating irreps one-by-one. Given a intermediate solution $n_R=(n_{R_1},n_{R_2},\ldots,n_{R_{\ell-1}},0,0,\ldots,0)$ involving only irreps $\{R_j\}_{j=1}^{\ell-1}$ one forms configurations $n_R'=(n_{R_1},n_{R_2},\ldots,n_{R_{\ell-1}},n_{R_\ell},0,\ldots,0)$ by picking all values of $n_{R_\ell}$ subject to the upper bound on $H$. The configuration $n_R'$ is then kept only if
\begin{equation}
\label{eq:B-constraint_upper_bound}
    \sum_{j=1}^\ell n_{R_j}H_{R_j} + H_\text{min}^{(\ell)}\Big( B_\adj - \sum_{j=1}^\ell n_{R_j}B_{R_j} \Big) \leq H_\text{max} \,,
\end{equation}
since the contributions of the yet-to-be-added hypermultiplets are controlled by the pre-computed bounds: the $H_\text{min}^{(\ell)}$ term in the above equation gives the minimum number of $\{R_j\}_{j>\ell}$ hypermultiplets needed to make up the difference in the current ``$B$-sum'' and land exactly on the target value of $B_\adj$. In this way one can construct all solutions of the $B$-constraint recursively, starting from $n_R=(0,0,\ldots,0)$.

As an example, figure~\ref{fig:B-constraint} shows three steps in this process for $\SU(5)$ with $\Delta_\text{max}=300$ ($H_\text{max}=324$). There are $m=26$ irreps with $H_R\leq H_\text{max}$ and as $\ell$ increases the upper bound on $H=\sum_{j=1}^\ell n_{R_j}H_{R_j}$ as a function of $B=\sum_{j=1}^\ell n_{R_j}B_{R_j}$ becomes more stringent, eventually constraining solutions to have $B=B_\adj$ exactly. For $\ell=25$ (green) the only irrep yet to be incorporated is the $\rep{40}$, which has $H_{\rep{40}}=40$ and $B_{\rep{40}}=-2$ (the smallest $|B_R|/H_R$ ratio) and thus
\begin{equation}
    H_\text{min}^{\ell=25}(B) = \begin{cases}
        -20B & B=0,-2,-4,\ldots\\
        \infty & \text{otherwise}
    \end{cases} \,.
\end{equation}
This explains why the green points in figure~\ref{fig:B-constraint} fall in the thin wedge $0 \leq 20(B - B_\adj) \leq H_\text{max}-H$ and only on the lines $B-B_\text{adj}=0,2,4,\ldots$.

\begin{figure}[t]
    \centering
    \includegraphics[width=\textwidth]{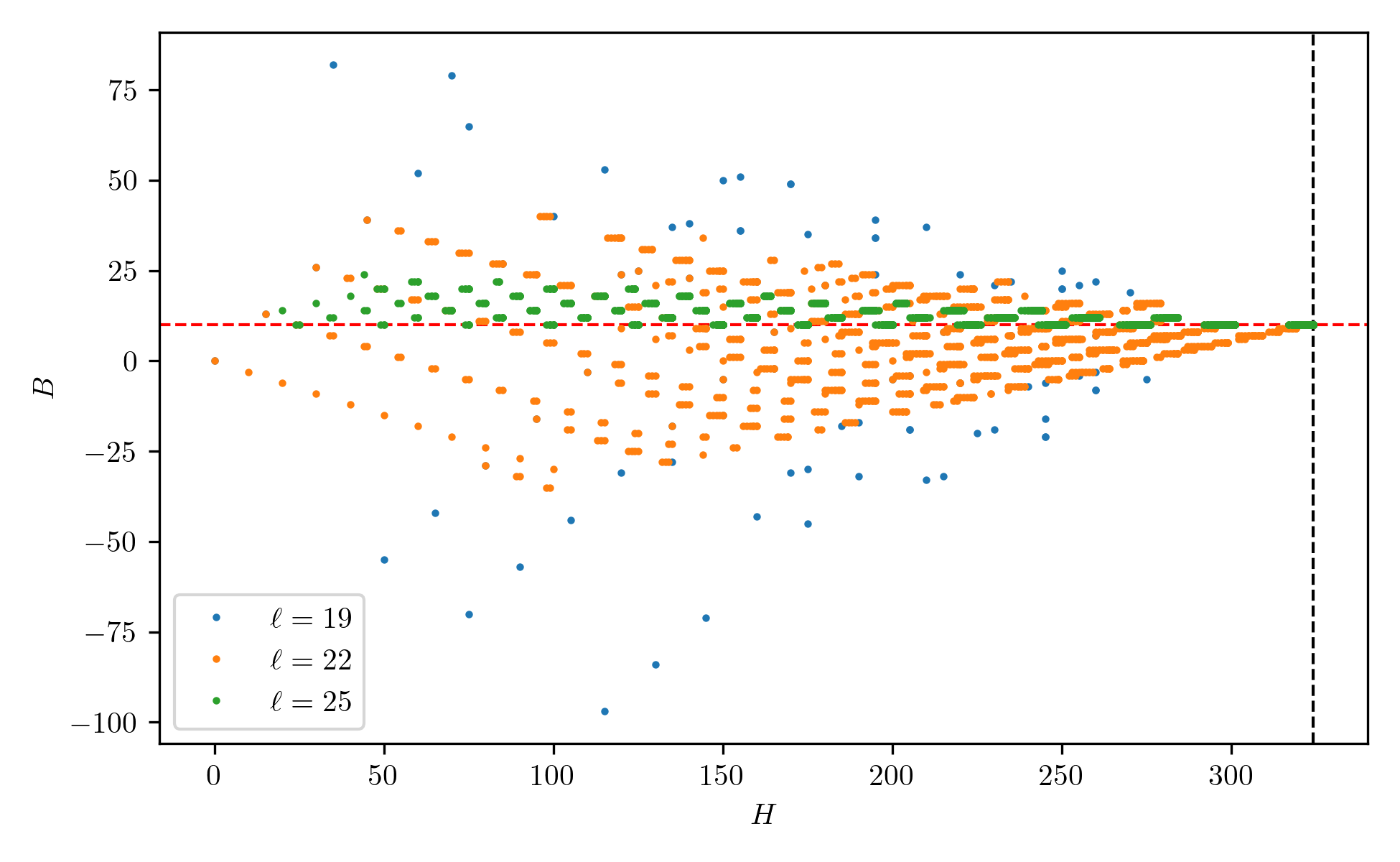}
    \caption{Three steps in the construction of $\SU(5)$ vertices below $\Delta_\text{max}=300$ ($H_\text{max}=324$). Partial, candidate solutions of the $B$-constraint using irreps $R_{j\leq\ell}$ are shown at $(H,B)=\sum_{j=1}^\ell n_{R_j}(H_{R_j},B_{R_j})$. As $\ell$ increases the bound of equation~\eqref{eq:B-constraint_upper_bound} becomes stronger, ultimately constraining solutions to lie exactly on the red, dashed line $B=B_\adj$.}
    \label{fig:B-constraint}
\end{figure}

\subsection{Edge construction}
\label{app:edge_construction}
Consider a fixed type-\A{A} vertex $\v_i$: during the construction of the multigraph $\G$ we would like to identify other type-\A{A} vertices which can potentially be connected to $\v_i$ via an edge much faster than na\"ively checking all $\mathcal{O}(|\V|)$ of them. As discussed in section~\ref{sec:graph_construction}, the key fact from~\eqref{eq:bibj_lower_bound} is that there need to be enough bi-charged hypermultiplets to have
\begin{equation}
\label{eq:bibj_bound_app}
    (b_i\cdot b_j)^2 \geq (b_i\cdot b_i)(b_j\cdot b_j) \,.
\end{equation}
Our task here will be to bound $b_j\cdot b_j$ given only $\v_i$ and the \emph{gauge group} $G_j$. Recall that the index $A_R^i$ is always strictly positive and that the left-hand side above is computed according to equation~\eqref{eq:anomaly_cancellation},
\begin{equation}
    b_i\cdot b_j = \sum_{R,S} n_{(R,S)}^{ij} A_R^i A_S^j \,,
\end{equation}
where $R$ and $S$ are representations of $G_i$ and $G_j$, respectively. By using
\begin{equation}
    \sum_S n_{(R,S)}^{ij}H_S^j \leq n_R^i \,,
\end{equation}
we can bound the sum as
\begin{equation}
\label{eq:bibj_upper_bound}
    b_i\cdot b_j = \sum_R A_R^i \Big(\sum_S n_{(R,S)}^{ij}A_S^j\Big) \leq \sum_R A_R^i A_\text{max}^{G_j}(n_R^i) \equiv \sigma_{i,j}^\text{max} \,,
\end{equation}
where $A_\text{max}^{G_j}(n)$ is defined as
\begin{equation}
    A_\text{max}^{G_j}(n) \equiv \max\!\Big\{ \sum_S x_SA_S^j \;\Big|\; \sum_S x_SH_S^j \leq n \,,\; x_S\in\Z_{\geq 0}\Big\} \,.
\end{equation}
This is a standard knapsack problem where the items are the representations $S$ of $G_j$, their values and weights are $A_S^j$ and $H_S^j$, respectively, and the maximum weight is $n$. The function $A_\text{max}^{G_j}(n)$ can be pre-computed for each group $G_j$ using\footnote{Note that if $\sum_Sx_SH_S^j=n$ is unattainable then $\hat{A}_\text{max}^{G_j}(n) = \sup\emptyset = -\infty$.}
\begin{equation}
    \begin{aligned}
        A_\text{max}^{G_j}(n) &= \begin{cases}
            0 & n = 0 \,,\\
            \max\!\big\{A_\text{max}^{G_j}(n-1),\hat{A}_\text{max}^{G_j}(n)\big\} & n > 0 \,,
        \end{cases}\\
        \hat{A}_\text{max}^{G_j}(n) &\equiv \sup\Big\{ \sum_S x_SA_S^j \;\Big|\; \sum_S x_SH_S^j = n \,,\; x_S\in\Z_{\geq 0}\Big\} \,,
    \end{aligned}
\end{equation}
where $\hat{A}_\text{max}^{G_j}(n)$ can be computed recursively via
\begin{equation}
    \hat{A}_\text{max}^{G_j}(n) = \begin{cases}
        0 & n = 0 \,,\\
        \sup\limits_{S\,:\,H_S^j\leq n}\big\{A_\text{max}^{G_j}(n-H_S^j)+A_S^j\big\} & n > 0 \,.
    \end{cases}
\end{equation}
For example, picking $G_j=\SU(4)$ for which the low-lying irreps have
\begin{equation}
    A_{\rep{4}} = 1 \,, \quad A_{\rep{6}} = 2 \,, \quad A_{\rep{10}} = 6 \,, \quad A_{\rep{15}} = 8 \,, \quad \ldots
\end{equation}
one finds the sequences $\hat{A}_\text{max}^{\SU(4)}$ and $A_\text{max}^{\SU(4)}$ begin
\begin{equation*}
    \newlength\mylen
    \settowidth\mylen{$-\infty$}
    \begin{array}{c | *{16}{wc{\mylen}} }
        n & 0 & 1 & 2 & 3 & 4 & 5 & 6 & 7 & 8 & 9 & 10 & 11 & 12 & 13 & 14 & 15\\ \midrule
        \hat{A}_\text{max}^{\SU(4)} & 0 & -\infty & -\infty & -\infty & 1 & -\infty & 2 & -\infty & 2 & -\infty & 6 & -\infty & 4 & -\infty & 7 & 8\\
        A_\text{max}^{\SU(4)} & 0 & 0 & 0 & 0 & 1 & 1 & 2 & 2 & 2 & 2 & 6 & 6 & 6 & 6 & 7 & 8
    \end{array}
\end{equation*}
We can use this to learn, for example, that the only $\SU(4)$ vertices which can potentially have a non-trivial edge to the $\SU(7)$ vertex corresponding to
\begin{equation}\label{eq:ffhat_example}
    G_i = \SU(7) \,, \qquad \H_i = 13\times\rep{7} + 5\times\rep{21} + \rep{35} + \rep{48} \,, \qquad b_i\cdot b_i = 8 \,,
\end{equation}
must have
\begin{equation}
\begin{aligned}
    b_j\cdot b_j &\leq \frac{1}{8}(b_i\cdot b_j)^2\\
    &\leq \frac{1}{8}\Big[ A_{\rep{7}}A_\text{max}^{\SU(4)}(13) + A_{\rep{21}}A_\text{max}^{\SU(4)}(5) + A_{\rep{35}}A_\text{max}^{\SU(4)}(1) + A_{\rep{48}}A_\text{max}^{\SU(4)}(1) \Big]^2\\
    &= \frac{1}{8}\big(1\cdot 6 + 5\cdot 1 + 10\cdot 0 + 14\cdot 0\big)^2\\
    &< 16 \,.
\end{aligned}
\end{equation}
Here~\eqref{eq:bibj_bound_app} and~\eqref{eq:bibj_upper_bound} are used in the first and second lines, respectively.
This bound is quite restrictive, being satisfied by only $189$ of the $\SU(4)$ type-\A{A} vertices.

\subsection{Clique pruning}
\label{app:pruning}

Recall from section~\ref{sec:clique_construction} that during the pruning step we would like to identify irreducible cliques for which $\Delta+28n_-^\gram$ cannot be brought under the bound of $273-T_\text{min}$ via any sequence of future branchings. Given an irreducible clique $\C_0^\irr$, we saw previously that any irreducible clique $\widetilde{\C}^\irr$ which contains $\C_0^\irr$ as a sub-clique can be formed by adjoining some number of irreducible cliques via a set $\E_\text{new}$ of new non-trivial edges between active vertices, as exemplified in~\eqref{eq:prune_decompose_1}. From this we found the following lower bound on $(\Delta+28n_-^\gram)(\widetilde{\C}^\irr)$ for any choice of $\lambda\in[0,1]$:
\begin{equation}
    \begin{aligned}
        (\Delta+28n_-^\gram)(\widetilde{\C}^\irr) &\geq 28(n_-^{\overline{\gram}}-n_-^\gram)(\C_0^\irr) + f_\lambda(\C_0^\irr) + \sum_{\alpha=1}^m f_{1-\lambda}(\C_\alpha^\irr) \,,\\
        f_y(\C^\irr) &\equiv (\Delta+28n_-^\gram)(\C^\irr) - y\sum_{\substack{\v\in\C^\irr\\ \text{active}}}\sum_Rn_R^\avbl(\v) H_R
    \end{aligned}
\end{equation}
For type-\A{A} cliques, choosing $\lambda=1$ and dropping the non-negative $f_0(\C_\alpha^\irr)$ terms is sufficient since branching quickly peters out regardless. When augmenting type-\A{A}\B{B} and type-\B{B} cliques by type-\B{B} vertices it behooves us to better understand exactly when $f_y(\C^\irr)$ is non-negative for type-\B{B} irreducible cliques. That we can take $y>0$ and maintain $f_y(\C^\irr)\geq 0$ reflects the fact that cliques with active vertices generally do not come close to saturating the lower bound $\Delta+28n_-^\gram\geq 0$.

\begin{figure}[t]
    \centering
    \includegraphics[width=\textwidth]{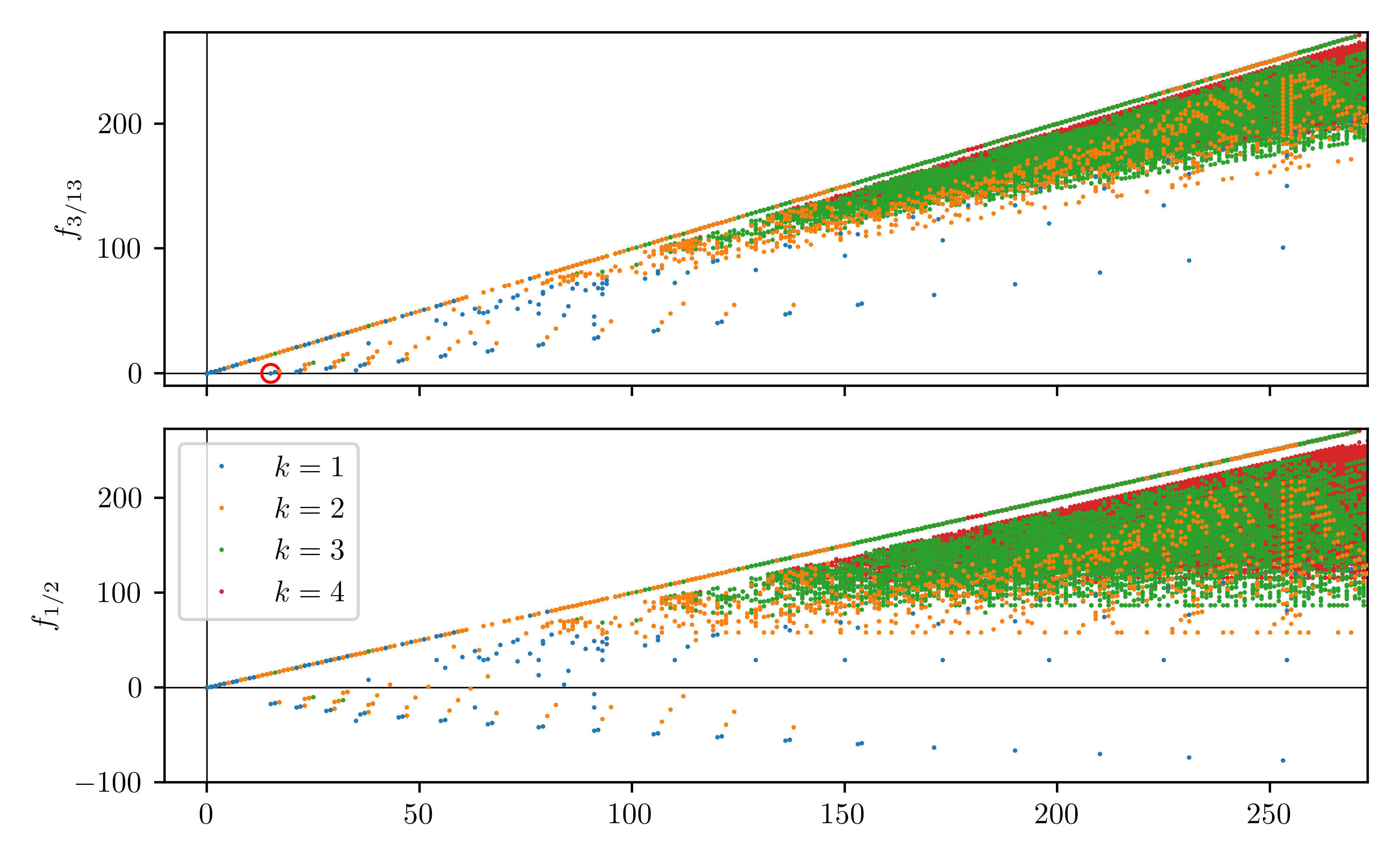}
    \caption{The functions $f_{3/13}$ (top) and $f_{1/2}$ (bottom) for type-\B{B} cliques with $k\leq4$ and the simple groups of section~\ref{sec:results_general}. The red circle indicates the clique of equation~\eqref{eq:first_neg_clique} which is the first to cross to $f_y<0$ as $y$ increases from zero.}
    \label{fig:fy_typeB}
\end{figure}

For $y=0$, $f_y(\C_{\B{B}}^\irr)$ is non-negative for all type-\B{B} cliques. As $y$ increases, the value of $f_y$ decreases for cliques with at least one active vertex. If we restrict attention to the simple groups considered in section~\ref{sec:results_general}, all irreps have $H_R\geq 5$ and the first clique to cross to negative values is the simple clique
\begin{equation}
\label{eq:first_neg_clique}
    \C = \{ \SO(13),5\times\rep{13} \} \,,\quad f_y(\C) = 15 - 65y
\end{equation}
at $y=\frac{3}{13}$. In the upper panel of figure~\ref{fig:fy_typeB} we see the values of $f_{3/13}(\C^\irr)$ for type-\B{B} cliques with $k\leq 4$. From this we conclude that
\begin{equation}
    f_y(\C_{\B{B}}^\irr) \geq 0 \,, \quad\forall y\in\big[0,\tfrac{3}{13}\big] \,,
\end{equation}
and obtain an improvement over~\eqref{eq:prune_rule_naive}:
\begin{equation}
    f_{10/13}(\C_0^\irr) + 28(n_-^{\overline{\gram}} - n_-^\gram)(\C_0^\irr) + T_\text{min}(\C_0^\irr) > 273 \quad\implies\quad \text{prune }\C_0^\irr \,.
\end{equation}
However, this is clearly sub-optimal, especially considering that the clique $\C_0^\irr$ may not even be able to form non-trivial edges with the simple clique in~\eqref{eq:first_neg_clique}.

If we extend all the way to $y=\frac{1}{2}$ then most irreducible type-\B{B} cliques have $f_{1/2}\geq 0$: see the lower panel of figure~\ref{fig:fy_typeB}. At least for the finite list of simple groups chosen for the general classification of section~\ref{sec:results_general}, there is a short list of exceptions which have $f_{1/2}<0$, all with exactly one active vertex. For example, we find
\begin{equation}
    \begin{aligned}
        f_{1/2}\left( \begin{tikzpicture}[baseline={([yshift=-0.5ex]current bounding box.center)}]
            \node[circle, inner sep=1.5, fill=typeB, draw=black, thick] (a) at (0,0) {};
            \draw (a) circle (0.15);
            \node[below=2] at (a) {\tiny\begin{minipage}{3cm}
                \centering
                $\SU(N)$\\
                $(N-8)\times\rep{N}+\rep{N(N+1)/2}$
            \end{minipage}};
        \end{tikzpicture} \right) &= -\frac{7N}{2}+29 \,, \\
        f_{1/2}\left( \begin{tikzpicture}[baseline={([yshift=-0.5ex]current bounding box.center)}]
            \node[circle, inner sep=1.5, fill=typeB, draw=black, thick] (a) at (0,0) {};
            \draw (a) circle (0.15);
            \node[below=2] at (a) {\tiny\begin{minipage}{1.5cm}
                \centering
                $\SO(N)$\\
                $(N-8)\times\rep{N}$
            \end{minipage}};
        \end{tikzpicture} \right) &= - \frac{7N}{2}+28 \,, \\
        f_{1/2}\left( \begin{tikzpicture}[baseline={([yshift=-0.5ex]current bounding box.center)}]
            \node[circle, inner sep=1.5, fill=typeB, draw=black, thick] (a) at (0,0) {};
            \draw (a) circle (0.15);
            \node[below=2] at (a) {\tiny\begin{minipage}{1cm}
                \centering
                $E_7$\\
                $\frac{k}{2}\times\rep{56}$
            \end{minipage}};
        \end{tikzpicture} \right) &= \begin{cases}
            14k - 105 & k\leq 7\\
            -31 & k=8
        \end{cases} \,\\
        f_{1/2}\left( \begin{tikzpicture}[baseline={([yshift=-0.5ex]current bounding box.center)}]
            \node[circle, inner sep=1.5, fill=typeB, draw=black, thick] (a) at (0,0) {};
            \node[circle, inner sep=1.5, fill=typeB, draw=black, thick] (b) at (2,0) {};
            \node[below=2] at (a) {\tiny\begin{minipage}{1cm}
                \centering
                $\SU(N)$\\
                $2N\times\rep{N}$
            \end{minipage}};
            \node[below=2] at (b) {\tiny\begin{minipage}{1.5cm}
                \centering
                $\SO(N+8)$\\
                $N\times\rep{(N+8)}$
            \end{minipage}};
            \draw (a) circle (0.15);
            \draw[very thick] (a) to node[above=-2] {\tiny$(\rep{N},\rep{N+8})$} (b);
        \end{tikzpicture} \right) &= -\frac{7N}{2} + 29 \,.\\
    \end{aligned}
\end{equation}

The worst-case scenario is where each active vertex of $\C_0^\irr$ is individually joined to as many of these cliques as possible. This means that the ``$+\epsilon$'' correction in equation~\eqref{eq:prune_lower_bound} can be computed as a sum over active vertices of $\C_0^\irr$,
\begin{equation}
\label{eq:+epsilon}
    ``{+\epsilon}" = \sum_{\substack{\v\in\C_0^\irr\\\text{active}}}\kappa\Big(\big\{(\e_\alpha,\C_\alpha^\irr)\big\};n_R^\avbl(\v)\Big) \,,
\end{equation}
where $\e_\alpha$ are the viable non-trivial edges from the active vertex $\v$ to active vertices of a clique $\C_\alpha^\irr$ with negative $f_{1/2}$. The function $\kappa$ denotes the solution to the knapsack problem where the items are the edge-clique pairs $(e_\alpha,\C_\alpha^\irr)$, the items' ``values'' are $|f_{1/2}(\C_\alpha^\irr)|$, their ``weights'' are their irrep usage at $\v$ and the knapsack's maximum capacity is the irrep availability at $\v$, given the pre-existing non-trivial edges in $\C_0^\irr$. Often the knapsack fits very few items (e.g.\ only one or two) and a standard recursive approach much like that described in appendix~\ref{app:edge_construction} to compute $A_\text{max}^{G_j}(n)$ quickly produces the value of $\kappa$.

\section{Brane probe constraint}
\label{app:brane_probe}

In this appendix we review how any infinite family with fixed Gram matrix is truncated to a finite number by the BPS string probe constraint, as considered in~\cite{Kim:2019vuc,Kim:2019ths,Lee:2019skh,Tarazi2021,Martucci:2022krl}.
To be precise, any infinite family with fixed Gram matrix must have gauge groups $G=G_\text{fixed}\times G_\text{aux}$ where $G_\text{fixed}$ is fixed and $G_\text{aux}$ consists of a fixed number of simple factors but parameterized by some collection of ranks $N_a$, and hypermultiplets chosen so that $\Gram$ is independent of the $N_a$.
This includes the first infinite family discussed in section~\ref{sec:infinite_families}, as well any infinite family which is of the form described on the right-hand side of equation~\eqref{eq:anomaly-free-clique-types}. 

Suppose that there exists a BPS string with the 2-form charge $Q=(q_0,\cdots,q_T)$.
The BPS string breaks the 6D Lorentz symmetry $\SO(1,5)$ to $\SO(1,1)\times \SO(4)$. The $\SO(1,1)$ symmetry is identified as the worldsheet Lorentz symmetry while on the other hand the $\SO(4)\simeq \SU(2)_R\times \SU(2)_\ell$ is identified as $R$-symmetry and flavor symmetry of the worldsheet. Moreover, the bulk gauge symmetry $G_i$ becomes the flavor symmetry of the worldsheet.
There are left and right-moving central charges $c_L, c_R$, and the levels of the current algebra $k_\ell$ and $k_i$ correspond to $\SU(2)_\ell$ and $G_i$.

From the anomaly inflow argument, one can compute the central charges as~\cite{Kim:2019vuc}
\begin{equation}
    \begin{aligned}
        c_L &= 3Q\cdot Q + 9Q\cdot b_0 + 2 \,,
        & c_R &= 3Q\cdot Q + 3Q\cdot b_0 \,,\\
        k_\ell &= \frac{1}{2}\left(Q\cdot Q - Q\cdot b_0 + 2\right)\,,
        \quad& k_i &= Q\cdot b_i\,,
    \end{aligned}
\end{equation}
where the contributions from the center of mass mode have been subtracted. Unitarity on the worldsheet requires
\begin{align}
    c_L &= 3Q\cdot Q+9Q\cdot b_0+2\geq0 \,, \label{eq:cL_positivity}
    \\
    c_R &= 3Q\cdot Q+3Q\cdot b_0\geq0 \,. \label{eq:cR_positivity}
\end{align}    
We focus on the situation where the current algebra of $\SU(2)_\ell$ and $G_i$ is realized on the left-movers.
Then we have the following positivity bounds:
\begin{align}
    k_\ell &= 
    \frac{1}{2}\left(Q\cdot Q - Q\cdot b_0+2\right)\geq0\,,\label{eq:kl_positivity}
    \\
    k_i &= Q\cdot b_i\geq0\,.\label{eq:ki_positivity}
\end{align}
Moreover, the central charge of the current algebra must be less than $c_L$. This leads to the condition
\begin{align}
    \sum_{i=1}^k \frac{k_i\dim G_i}{k_i+h^\vee_i}\leq c_L\,, \label{eq:cL_bound}
\end{align}
where $h^\vee_i$ is the dual Coxeter number of $G_i$, and $\dim G_i$ is the dimension of $G_i$.

The only remaining question is what is the charge $Q$ of the BPS string. Motivated by the BPS completeness hypothesis,\footnote{Explicit examples of string theories with sixteen supercharges were constructed in~\cite{Montero:2022vva} where BPS strings populate an index-2 sublattice of the charge lattice. However, the non-BPS string has the charge ``1/2'' in this case, and one can apply the constraint for the charge $1$ BPS string.} requiring that the condition~\eqref{eq:cL_bound} must be satisfied for \emph{any} $Q\in\Gamma$ which satisfies (\ref{eq:cL_positivity}, \ref{eq:cR_positivity}, \ref{eq:kl_positivity}, \ref{eq:ki_positivity}) leads to extremely strong bounds on the $N_a$~\cite{Kim:2019vuc,Tarazi2021}. However, for our purposes it is enough to assume that there exists \emph{at least} one $Q\in\Gamma$ which satisfies (\ref{eq:cL_positivity}, \ref{eq:cR_positivity}, \ref{eq:kl_positivity}, \ref{eq:ki_positivity}) and with each component $Q^\alpha$ bounded as the ranks $N_a$ increase. For example, it is enough for the BPS charges to form a sublattice of $\Gamma$ of bounded index and the $N_a$ grow. With the following simple scaling argument one can show that the $N_a$ are truncated to a finite number. Since the Gram matrix is constant over the infinite family, the vectors $b_0,b_i$ are fixed and we have
\begin{equation}
    c_L,c_R,k_\ell,c_i \sim \mathcal{O}(Q^2+b_I Q) \sim \mathcal{O}(N_a^0) \,.
\end{equation}
In contrast,
\begin{equation}
\label{eq:LHS-scaling}
    \sum_{i=1}^k \frac{k_i\dim G_i}{k_i+h^\vee_i} \geq \sum_{i=1}^k \frac{\dim G_i}{1+h^\vee_i} \sim \mathcal{O}(N_a) \,,
\end{equation}
where we have used that $h^\vee_{A_{N-1}}=N$, $h^\vee_{B_N}=2N-1$, $h^\vee_{C_N}=N+1$, $h^\vee_{D_N}=2N-2$ for the dual Coxeter number and $\dim G_{A_{N-1}}=N^2-1$, $\dim G_{SO(N)}=N(N-1)/2$, $\dim G_{C_N}=N(2N+1)$ for the dimension of the group. Since the left- and right-hand sides of~\eqref{eq:cL_bound} scale differently, for large enough $N_a$ we necessarily violate unitarity and the theories are inconsistent.

The only loophole in the above argument is having $k_i=0$ for all $G_i$ in $G_\text{aux}$ which have growing rank, since then the scaling of~\eqref{eq:LHS-scaling} is incorrect: the only nonzero terms would be constant with $N_a$. However, this cannot happen if the BPS charge lattice has finite index in $\Gamma$, i.e.\ has full rank, since then one can find a charge $Q\approx \lambda j$ for some $\lambda\gg1$, for which $c_L,c_R,k_\ell\sim \lambda^2(j\cdot j)>0$ and $k_i=Q\cdot b_i\approx \lambda(j\cdot b_i)>0$ is strict.

\bibliographystyle{JHEP}
\bibliography{refs}

\end{document}